\documentclass[11pt,a4paper,twoside]{article}
\usepackage{geometry} %% change the margin
\usepackage{amsmath,amsthm}
\usepackage{amssymb}
\usepackage{color}                   %Double column single space version default
\usepackage{rotating}
\usepackage{mdframed}
\usepackage{graphicx,amssymb}
\usepackage{algorithm}
\usepackage{algorithmic}
\usepackage{color,appendix}
\usepackage{diagbox}
\usepackage{multirow}
\usepackage{subcaption}
\usepackage{makecell}
\usepackage{hyperref}
\usepackage[affil-it]{authblk}
\usepackage{soul}
\usepackage{geometry}
\makeatletter

\@addtoreset{algorithm}{section}
\makeatother

\newtheorem{theorem}{Theorem}[section]

\newtheorem{proposition}[theorem]{Proposition}

\newtheorem{remark}[theorem]{Remark}

\renewcommand{\theequation} {\thesection.\arabic{equation}}

\def\R{ {\mathbb R} }
\def\Z{ {\mathbb Z} }

\numberwithin{equation}{section}

\setstcolor{red}
\title{Polynomial graph filters of multiple  shifts and  distributed implementation of  inverse filtering}

\author{Nazar Emirov}
\affil{Department of Computer Science, Boston College, Chestnut Hill, Massachusetts 02467, USA. \vspace{.1in} }
\author{Cheng Cheng\thanks{Corresponding author}}
\affil{School of Mathematics, Sun Yat-sen University, Guangzhou, Guangdong 510275, China.\vspace{.1in}}
\author[3]{Junzheng Jiang}
\affil{School of Information and Communication, Guilin University of Electronic Technology, Guilin, Guangxi 541004, China. \vspace{.1in}}
\author[4]{Qiyu Sun}
\affil{Department of Mathematics, University of Central Florida, Orlando, Florida 32816, USA. \vspace{.05in}}
\affil[ ]{ {nazar.emirov@bc.edu}, {chengch66@mail.sysu.edu.cn}, {jzjiang@guet.edu.cn}, {qiyu.sun@ucf.edu} }

\date{}

\begin{document}

	\maketitle

\begin{abstract}
Polynomial graph filters  and their inverses play  important roles in graph signal processing.  In this paper,
we introduce the concept of multiple commutative graph shifts and polynomial graph filters, which
could  play similar roles  in graph signal processing
   as the  one-order delay and finite impulse response filters  in classical multi-dimensional signal processing.
 We implement the filtering procedure associated with a polynomial graph filter of multiple shifts  at the vertex level in  a distributed network  on which each vertex is equipped with a  data processing subsystem for limited  computation power and data storage, and a communication subsystem for direct data exchange only to its adjacent vertices. In this paper, we also consider the  implementation of  inverse filtering  procedure associated with a polynomial graph filter of multiple  shifts, and we propose two  iterative approximation algorithms  applicable in a distributed network and in a central facility. We also demonstrate the effectiveness of the proposed algorithms to implement the inverse filtering procedure  on  denoising time-varying  graph signals
  and  a dataset of  US hourly temperature at $218$ locations.
\end{abstract}

%%%Graphical abstract
%\begin{graphicalabstract}
%%\includegraphics{grabs}
%\end{graphicalabstract}

%%%Research highlights
%\begin{highlights}
%\item Research highlight 1
%\item Research highlight 2
%\end{highlights}


%% \linenumbers

%% main text
%\section{}
%\label{}

%% The Appendices part is started with the command \appendix;
%% appendix sections are then done as normal sections
%% \appendix

%% \section{}
%% \label{}

%% If you have bibdatabase file and want bibtex to generate the
%% bibitems, please use
%%
%%  \bibliographystyle{elsarticle-harv}
%%  \bibliography{<your bibdatabase>}

%% else use the following coding to input the bibitems directly in the
%% TeX file.

\section{Introduction}  %\noindent

	Graph signal processing  provides an innovative framework to handle  data residing on    spatially distributed networks (SDN), %such as
% smart grids, neural networks, social networks and many other
%	irregular domains \cite{shuman13}--\cite{Ortega18}.
such as the wireless sensor networks,  smart grids and social network and many other
	irregular domains,  \cite{cheng_SDS16,  Hebner17, Ortega18, aliaksei14, shuman13, Yick08}.
% (\cite{shuman13, aliaksei14,  Ortega18}).
Graphs provide a flexible tool
to model the underlying topology of the networks, and the
edges present the interrelationship between data elements.
 %The graph topology in the underlying framework offers a  flexible tool to model the interrelationship between data elements.
  For instance,  an edge between two vertices may indicate the availability of a direct data exchanging channel between sensors of a distributed network,
the functional connectivity between neural regions in brain,
  or
  the correlation between temperature records of neighboring locations. By leveraging graph  spectral theory and applied harmonic analysis, graph signal processing
 has been  extensively exploited, and
many important concepts in classical signal processing, such as  Fourier transform and  wavelet filter banks, % and graph filters, %leading to graph Fourier transform, graph wavelet filter banks and graph filters,
have been extended to graph  setting  %signal processing
{\cite{cheng_SDS16,hammod11,narang12,Ortega18,aliaksei13,sandryhaila14,aliaksei14,shuman13,narang13}}.

	Let  ${\mathcal G}:=(V,E)$ be an  undirected  and unweighted graph with vertex set $V=\{1, \ldots, N\}$
	and edge set $E\subset V \times V$,
and define the
	geodesic distance $\rho(i,j)$   between vertices $i,j\in V$
	by  the number of edges in
a shortest path connecting  $i, j\in V$
%by the number of edges in a shortest path connecting them,
	and set  $ \rho(i,j)=\infty$ if vertices $i, j\in V$ belong to its different connected components.
   A {\em graph filter} %${\bf H}$
    on the graph $\mathcal G$  maps one graph signal   ${\bf x}=(x(i))_{i\in V}$
   linearly to another graph signal ${\bf y}={\bf H}{\bf x}$,
   and it is usually represented  by a matrix ${\bf H}=(H(i,j))_{ i,j\in V}$.
%	\vspace{-0.5em}\begin{equation}\label{filter.def00} {\bf H}=(H(i,j))_{ i,j\in V}.
%	\vspace{-0.6em}\end{equation}
Graph filters
and their implementations are fundamental  in graph signal processing,
 and they have been used in
	denoising, smoothing, consensus of multi-agent systems, the estimation of time series  and many other applications
{\cite{isufi19,shuman18,Waheed18,moura14}}.
 In the classical signal processing,  filters are categorized into two families, finite impulse response (FIR) filters and  infinite impulse response (IIR) filters. The FIR concept has been extended to graph filters with the duration of an FIR filter being replaced by the
  geodesic-width of a graph filter. Here the {\em geodesic-width} $\omega({\bf H})$
   of a graph filter
  ${\bf H}=(H(i,j))_{ i,j\in V}$ % ${\bf H}$ in \eqref{filter.def00}
	is the smallest nonnegative integer $\omega({\bf H})$ such that
%\vspace{-.4em}\begin{equation*}\label{bandwidth.def}
$	H(i,j)=0 $  hold for  all $i,j\in V$   with $\rho(i,j)>\omega({\bf H})$
%\vspace{-.4em}\end{equation*}
{ \cite{cheng2021,cheng_SDS16, ekambaram13,  jiang19, David2019}}.
% \cite{cheng_SDS16}, \cite{ ekambaram13}--\cite{cheng2021}. % \cite{cheng_SDS16, ekambaram13,  jiang19, David2019, cheng2021}.
	%  has significant advantage on their implement, is
%For a  graph filter ${\bf H}$ with limited geodesic-width  $\omega({\bf H})$,
%the   value   of the output signal of  the filtering procedure ${\bf x}\longmapsto {\bf H}{\bf x}$
% is a ``weighted'' sum of values of the input signal ${\bf x}$
%  at its  $\omega({\bf  H})$-hop vertices, i.e.,
%{\bf neighboring} vertices within geodesic distance no more than  $\omega({\bf  H})$
%\cite{jiang19,David2019}
%{\color{red} (references \cite{hammod11, aliaksei13,narang13, shuman18} are not related to this paragraph ) add some references here}
%since
% the   value   of the output signal
% at each vertex is a ``weighted'' sum of input signal values at its
%{\bf neighboring} vertices within geodesic distance no more than  $\omega({\bf  H})$. {\color{red} reference here}
%

	An	elementary graph filter  is a {\em graph shift}, which has one as its geodesic-width \cite{Coutino17, jiang19, aliaksei13,   segarra17}.
In this paper, we introduce the
concept of {\bf multiple} commutative graph shifts ${\bf S}_1, \dots,  {\bf S}_d$,
i.e.,
%	\vspace{-0.4em}
\begin{equation}\label{commutativityS}
	{\bf S}_k{\bf S}_{k'}={\bf S}_{k'}{\bf S}_k,\  1\le k,k'\le d,
%	\vspace{-0.6em}
\end{equation}
   and
 we consider the implementation of filtering and inverse filtering associated with a {\em polynomial graph filter}
%	\vspace{-0.6em}
\begin{equation}\label{MultiShiftPolynomial}
	{\bf H}=h({\bf S}_1, \ldots, {\bf S}_d)=\sum_{ l_1=0}^{L_1} \cdots \sum_{ l_d=0}^{L_d}  h_{l_1,\dots,l_d}{\bf S}_1^{l_1}\cdots {\bf S}_d^{l_d},
%	\vspace{-0.6em}
\end{equation}
   where the polynomial
%\vspace{-0.6em}
 \begin{equation*}\label{MultiShiftPolynomial.polynomial}
h(t_1, \ldots, t_d)=\sum_{ l_1=0}^{L_1} \cdots \sum_{ l_d=0}^{L_d}  h_{l_1,\dots,l_d} t_1^{l_1} \ldots t_d^{l_d} %\vspace{-0.4em}
\end{equation*}
 % of degree $\sum_{k=1}^d L_k$.
%h_{l_1,\dots,l_d},  0\le l_k\le L_k, 1\le k\le d$ are the polynomial coefficients for the terms  ${\bf S}_1^{l_1}\cdots {\bf S}_d^{l_d}$.
 in variables $t_1,\cdots,t_d$ has polynomial coefficients $h_{l_1,\dots,l_d}$, $0\le l_k\le L_k, 1\le k\le d$.
The commutativity of graph shifts ${\bf S}_1, \dots,  {\bf S}_d$
guarantees that the polynomial graph filter ${\bf H}$  in \eqref{MultiShiftPolynomial}  is independent on
equivalent expressions of the multivariate polynomial $h$, and also the well-definedness of their joint spectrum, see Appendix
\ref{jointspectrum.appendix}.
The concept of commutative graph shifts  ${\bf S}_1, \dots,  {\bf S}_d$
    may play a similar role  in graph signal processing
   as the  one-order delay $z_1^{-1}, \ldots, z_d^{-1}$ in classical multi-dimensional signal processing,
   and  in practice  graph shifts may have specific features and physical interpretation,
  see
   \ref{commutative.section} and Section \ref{Numeric.section} for
    some illustrative examples. % and  for examples of commutative graph filters in noise reduction.
We remark that the commutative assumption on  graph shifts ${\bf S}_1, \ldots, {\bf S}_d$ is trivial for $d=1$ and
 polynomial graph filters  of a single shift have been widely used in graph signal processing
 % \cite{Waheed18,  shuman18, segarra17}, \cite{mario19}--\cite{Emirov19}.
 \cite{ Emirov19,  mario19, Leus17,  Lu18, segarra17,shuman18,Waheed18}.

  Polynomial graph filters ${\bf H}$ in  \eqref{MultiShiftPolynomial}
	have geodesic-width $\omega({\bf H})$ no more than the degree $\sum_{k=1}^d L_k$ of the polynomial $h$.  Our study of
 polynomial graph filters
 of  multiple shifts
 is {\bf mainly} motivated by signal processing on  time-varying signals,
  such as video and data collected by a sensor network over a period of time,
which carry different correlation characteristics %and has different properties
 for different dimensions/directions. In such a scenario, graph  filters
 should be designed to reflect spectral characteristic  on the  vertex domain and also on the temporal domain,
 hence  polynomial graph filters
 of  multiple commutative shifts are  {\bf  preferable}, see
  \cite{Kurokawa17, Ortega18, aliaksei14} and also Subsections \ref{denoising.subsection} and \ref{denoisingweather.subsection}.
  The design of polynomial filters of multiple graph shifts
and their inverses with specific features
and physical interpretation for engineering applications
 is beyond the scope of this paper
and it will be discussed in our future work.
 %for the demonstration on denoising.  	
%  an digital image
% may have different patterns on horizontal, vertical and diagonal directions
% and hence graph filters for imaging processing should  be designed to include horizontal, vertical and diagonal shifts.
%   the graph to  describe the underlying topology is usually a two-dimensional rectangular lattice,
% and the spectral characteristic of the image which  can be shifted horizontally, vertically and diagonally depending on the image structure.
%	Recently, there is an  increasing demand to process multi-dimensional signals on graphs [citations], which shows different characteristics along different dimensions. For instance, the time-varying signals on graph present different correlation characteristics across the graph domain and temporal domain, the pixels of images show different oscillation features across horizonal, vertical and cross directions.
%		Product graph has been employed to model the domain of the multi-dimensional signal, for example, the Cartesian product graph is used to model the time-varying graph signal, by using the Kronecker sum of the general graph and ring graph. However, the existing product graph model and the corresponding theory can not distinguish different information along different directions, since  there is only one shift, i.e. the Laplacian/adjacency matrix, used for constructing linear operators to process the time-varying graph signals that prevents the separate operation.
%In order to overcome this deficiency,
 %which motivates the consideration of graph filters with multiple graph shifts.
  Our discussion
is {\bf also} motivated by  directional frequency analysis in \cite{Kurokawa17}, feature separation in
\cite{Fan19}
% forvmulti-dimensional graph signals
and
  graph filtering in  \cite{Radu19} for  time-varying graph signals. % with stationary interactions between nodes in time.

%and  hence  the corresponding  filtering procedure  can be implemented at the
%vertex level.
%``{\bf locally}" implemented in the vertex domain.
 For polynomial graph filters of a {\bf single}  shift,
 algorithms have been proposed to implement
their filtering procedure in finite steps, with each step including data exchanging between  {\bf adjacent} vertices only, see
\cite{ Emirov19, Leus17, Shi15, shuman18, Waheed18, moura14} and also Algorithm \ref{singleshiftprocedure.algorithm}.
%\cite{moura14, isufi17, shuman18,   Leus17,  Shi15, Emirov19}.
% The  first main contribution is that in Section \ref{implementation.section}
% we extend Algorithm \ref{singleshiftprocedure.algorithm}
% to implement filtering procedure associated with polynomial graph filters
% of multiple shifts, see Algorithm \ref{distributed_Hx.algorithm}.
 The first main contribution is  to extend  the one-hop implementation in Algorithm \ref{singleshiftprocedure.algorithm}
 to the filtering procedure associated with polynomial graph filters
 of {\bf multiple shifts},
 see Algorithm \ref{distributed_Hx.algorithm} in Section \ref{implementation.section}.
%\cc{An SDN has a large amount of agents
%and each agent equipped with a data processing subsystem
%having limited data storage and computation power and a
%communication subsystem for data exchanging to its adjacent agents.
%The implementation of data processing on an SDN is a distributed %processing
% task and
% it should be designed at agent/vertex level. }
Therefore the  filtering procedure associated with polynomial graph filters
can be implemented
 on an SDN  on which  each agent is equipped with a data processing subsystem
having limited data storage and computation power,  and with a
communication subsystem for data exchanging to its adjacent agents.

  %	In the classical signal processing,  IIR filters can be designed to provide better spectral characteristic than FIR filters of
% the same order do, and  the corresponding filtering procedures can be  implemented
% by the combination of an FIR filtering  and an inverse FIR filtering.

%	In the classical signal processing,  IIR filters
%%can be designed to provide better spectral characteristic than FIR filters of
%% the same order do, and
%the corresponding filtering procedures
% can be  implemented
% by the combination of an FIR filtering  and an inverse FIR filtering.

 Inverse filtering associated with  the graph filter having small  geodesic-width
 plays an important role in graph signal processing,
 %	Along with the filtering operators on graph by using the graph filter, the inverse filtering operator is also a fundamental issue in graph signal processing framework.
such as  denoising, graph semi-supervised learning,  non-subsampled filter banks
 and signal reconstruction
%\cite{shuman18, jiang19}, \cite{ mario19}--\cite{Emirov19}, \cite{Shi15}--\cite{sihengTV15}.
{\cite{siheng_inpaint15, sihengTV15,  Emirov19,mario19, Leus17,  jiang19, Lu18, Onuki16,Shi15,shuman18}}.
%(the graph signal sampling and reconstruction is not viewed as the  inverse filtering in general.),
  The challenge arisen in the inverse filtering  is on its implementation, as
  the inverse filter  ${\bf H}^{-1}$ usually has  full geodesic-width even if the original filter  ${\bf H}$ has small geodesic-width.
% Although the filter operator can be implemented in distributed manner when  it possesses  certain localized characteristic, its inverse is not localized in general according to the matrix inverse property. As a result, the inverse filtering can not realized in distributed manner.
%	Therefore, the inverse filtering presents big challenge, particularly when the graph is of large order [].
%	For a graph filter  ${\bf H}$ of  small bandwidth, the matrix  ${\bf H}^{-1}$ associated with the inverse filtering procedure usually has  full bandwidth.
For the case that  the filter ${\bf  H}$  is strictly positive definite,
%has its eigenvalues contained in the positive axis,
 the inverse filtering procedure ${\bf b}\longmapsto {\bf H}^{-1}{\bf b}$ can be
implemented by applying the iterative gradient descent method
%\cc{Delete: \st{with each iteration being implementable}}
in a distributed network,  see \cite{siheng_inpaint15,
 qiu2017,Shi15} and Remark \ref{remark.grad}.
 %chosen so that $\gamma {\bf H}$ has all eigenvalues contained in the open interval $(0, 2)$.
Denote the identity matrix by ${\bf I}$.
 To consider implementation of inverse filtering of an arbitrary invertible filter ${\bf H}$ with small geodesic-width, in Section \ref{iterativeapproximation.section}
 we  start from selecting a graph filter  ${\bf G}$ with small geodesic-width to approximate the inverse
filter ${\bf H}^{-1}$ so that the spectral radius of ${\bf I}-{\bf H}{\bf G}$  is strictly less than 1, and then we propose an  {\bf exponential convergent} algorithm \eqref{iterativedistributedalgorithm.eqn1}
	and \eqref{iterativedistributedalgorithm.eqn2} to  implement the inverse filtering procedure
	with each iteration  mainly including two  filtering procedures associated with filters ${\bf H}$ and ${\bf G}$, see  Theorem \ref{convergence_ICPA.thm}. % for the {\bf exponential convergence}.
% , we  establish the exponential convergence of
%the  iterative approximation  algorithm
%\eqref{iterativedistributedalgorithm.eqn1}
%	and \eqref{iterativedistributedalgorithm.eqn2}
%when the  graph filter  ${\bf G}$ provides a good approximation to the inverse
%filter ${\bf H}^{-1}$.
% in the sense that the spectral radius of ${\bf I}-{\bf H}{\bf G}$  is strictly less than 1.

% is selected to satisfy
%	\vspace{-.5em}\begin{equation}\label{sigmaHG.eq} \rho({\bf I}-{\bf H}{\bf G})<1,
%	\vspace{-.5em}\end{equation}
%where $\rho({\bf A})$ is a spectral radius of an matrix ${\bf A}$.
%The performance of the  proposed iterative algorithm %highly
%depends on the selection of the approximation filter ${\bf G}$ to the inverse filter ${\bf H}^{-1}$.
For  an invertible polynomial graph filter of a {\bf single}   shift,
  there are several methods to implement the inverse filtering in a distributed network  {\bf approximately} %anner
% \cite{shuman18,  segarra17,mario19,  Leus17, Emirov19, Shi15}.
\cite{Emirov19, mario19,  Leus17, Shi15,shuman18}, see Remark \ref{Chebyshev.remark}.
 The second main contribution of this paper is that we introduce   optimal  polynomial  filters
and multivariate Chebyshev polynomial  filters to provide  good approximations to the inverse
  of an invertible polynomial graph filter ${\bf H}$ of {\bf multiple}  shifts, see Section  \ref{ipaa.section}.
  % The second main contribution of this paper is on  finding a good approximation to the inverse
%  of an invertible polynomial graph filter ${\bf H}$ of multiple graph shifts.
%In Section  \ref{ipaa.section}, %of commutative shifts,
%we introduce  optimal  polynomial  filters
%and Chebyshev polynomial  filters  to  approximate the inverse filter ${\bf H}^{-1}$.
%\cc{Rewrite:   Apply the obtained optimal/Chebyshev polynomial filter   in the iterative approximation algorithm in Section  \ref{iterativeapproximation.section}, the induced iterative optimal polynomial
%approximation algorithm \eqref{optimaliterativedistributedalgorithm.eqn1}
%and the iterative Chebyshev polynomial approximation algorithm
%\eqref{Chebysheviterativedistributedalgorithm.eqn1} can be implemented at the vertex level in a distributed network and have exponential convergence, see Algorithms \ref{IOPA.algorithm}
%and \ref{ICPA.algorithm} and
% Theorems \ref{IOPAconvergence.thm} and \ref{icpa.thm}.
%}
Then, based on the iterative approximation algorithm in Section  \ref{iterativeapproximation.section},
we propose the iterative optimal polynomial
approximation algorithm \eqref{optimaliterativedistributedalgorithm.eqn1}
and the iterative Chebyshev polynomial approximation algorithm
\eqref{Chebysheviterativedistributedalgorithm.eqn1} to implement  the inverse filtering
procedure ${\bf b}\longmapsto {\bf H}^{-1}{\bf b}$,
 see
 Theorems \ref{IOPAconvergence.thm} and \ref{icpa.thm} for their exponential convergence. % in Section \ref{ipaa.section}.
More importantly, as shown in Algorithms \ref{IOPA.algorithm}
and \ref{ICPA.algorithm}, each iteration in the  proposed iterative  algorithms
mainly contains
 two  filtering procedures involving
data exchanging between {\bf adjacent} vertices only and hence they can be implemented in a {\bf distributed network} of large size,
 where  each vertex is equipped with systems for  limited data storage, computation power and
	data exchanging facility to its adjacent vertices.
%with direct communication between adjacent vertices.
%  and the Chebyshev polynomial approximation filters has been successfully applied  to the inverse filtering procedure, see \cite{shuman18} and references therein. {\color{blue} cite our sampta paper, and mention that for the polynomial of a single graph shift, the Chebyshev polynomial approximation with the iterative algorithm has been announced there. }

The paper is organized as follows. In Section \ref{implementation.section}, we consider  distributed
 implementation of the filtering procedure associated with polynomial graph filters
 of  multiple shifts at the vertex level. In Section   \ref{iterativeapproximation.section}, we propose an
  iterative approximation algorithm to implement an inverse filtering procedure.
In Section  \ref{ipaa.section}, we propose the iterative optimal polynomial
approximation algorithm
and the iterative Chebyshev polynomial approximation algorithm
 to implement  the inverse filtering
procedure associated with a polynomial filter.
The effectiveness of these two iterative algorithms to implement the inverse filtering procedure
 is demonstrated in Section
\ref{Numeric.section}.
In  \ref{commutative.section}, we introduce  two illustrative  families of commutative graph shifts
 on circulant graphs and product graphs respectively, and we define
  joint spectrum \eqref{jointspectrum.def} of multiple commutative graph shifts, which     %are commutative
 is crucial
   for us to  develop the iterative algorithms   for inverse filtering
     in Section  \ref{ipaa.section}.
In  \ref{commutative.section}, we also consider the problem when
a graph filter  is a polynomial of  multiple commutative graph  shifts
and how to measure the distance between a graph filter and the set of polynomial filters of commutative graph shifts.

	\section{Polynomial filter and distributed implementation}\label{implementation.section}

%	\begin{figure}[t]  %[t]
%		%\setlength{\belowcaptionskip}{-1cm}
%		\begin{center}
%			\includegraphics[width=90mm, height=22mm]{single_shift_diagram}
%			%\includegraphics[width=82mm, height=48mm]{two_shift_diagram}%{multI_shift_fig5_new}%{Pol_Implementation_Diagram}
%			\caption{Block diagram to implement the filtering procedure  ${\bf x}\longmapsto {\bf H}{\bf x}$ corresponding to a polynomial filter ${\bf H}=\sum^L_{l=0} h_l {\bf S}^l$.}
%%				of thshift ${\bf S}$. } %{\bf Too much space between figure caption and the text.}			
%			\label{onepolynomialfiltering_structure}
%		\end{center}
%		\vspace{-1em}
%	\end{figure}

%%%%%%%%%%%%%%%%%%%%%
   \begin{algorithm}[t]
\caption{ Backward iteratively  synchronous realization of the filtering procedure ${\bf x}\longmapsto {\bf H}{\bf x}$ for a polynomial filter ${\bf H}=\sum^L_{l=0} h_l {\bf S}^l$ at a vertex $i\in V$. }
\label{singleshiftprocedure.algorithm}
\begin{algorithmic}  %[1]

\STATE {\bf Inputs}: Polynomial coefficients
$h_0, h_1, \ldots, h_L$, entries $S(i,j), j\in {\mathcal N}_i$ in the $i$-th row of the shift ${\bf S}$, and
the  %signal
 value $x(i)$  of the input signal ${\bf x}=(x(i))_{i\in V}$ at the vertex $i$.

%\STATE {\bf Operation}: Evaluate $m_k=\mu(B(k, r))$, compute ${\bf F}_k= {\bf H}_{0,k}^T{\bf H}_{0,k}+ {\bf H}_{1,k}^T{\bf H}_{1,k}$,
% find its inverse  $({\bf F}_k)^{-1}$, and then compute $ {\bf G}^L_{l; k}:=({\bf F}_k)^{-1} {\bf H}_{l,k}^T, l=0, 1$.

%  , and a local approximation
%${\bf G}_k=(\tilde f_k(i,j))_{i,j\in B(k, 2r)}$  to the matrix ${\bf H}$
% and compute ${\bf F}_k= {\bf H}_{0,k}^T{\bf H}_{0,k}+ {\bf H}_{1,k}^T{\bf H}_{1,k}$ and
%$({\bf F}_k)^{-1}=(\tilde f_k(i,j))_{i,j\in B(k, 2r)}$

\STATE {\bf Initialization}:   $z^{(0)}(i)=h_L  x(i)$ and $n=0$.

\STATE{\bf 1)} Send $z^{(n)}(i)$ to its adjacent vertices $j\in {\mathcal N}_i$
and receive $z^{(n)}(j)$ from its  adjacent vertices $j\in {\mathcal N}_i$.

\STATE{\bf 2)} Update
 $z^{(n+1)}(i)=h_{L-n-1} x(i) + \sum\limits_{j\in {\mathcal N}_i} S(i, j) z^{(n)}(j).$

\STATE{\bf 3)}
 Set $n=n+1$ and return to  Step   {\bf 1)} if $n\le L-1$.

\STATE {\bf Output}: The  value $ \tilde x(i)=z^{(L)}(i)$  is  the output signal ${\bf H}{\bf x}=(\tilde x(i))_{i\in V}$ at the vertex $i$.
\end{algorithmic}  %\vspace{-.03in}
\end{algorithm}

Let ${\mathcal G}=(V, E)$ be a  connected, undirected and unweighted graph of order $N$.
 Graph shifts  ${\bf S}$ on ${\mathcal G}$ are building blocks of a polynomial filter.
 Our familiar examples of graph  shifts  %on a graph ${\mathcal G}$
are the adjacency matrix  ${\bf A}_{\mathcal G}$,
	%the degree matrix  ${\bf D}_{\mathcal G}$ of the graph ${\mathcal G}$,
	Laplacian matrix ${\bf L}_{\cal G}:={\bf D}_{\mathcal G}-{\bf A}_{\mathcal G}$,
symmetric	normalized Laplacian matrix ${\bf L}^{\rm sym}_{\cal G}={\bf D}_{\mathcal G}^{-1/2}{\bf L}_{\cal G}{\bf D}_{\mathcal G}^{-1/2}$
 %node-variant graph filters \cite{segarra17} and edge-variant graph filters \cite{Coutino17}
  and their
	variants, where ${\bf D}_{\mathcal G}$ is the degree matrix of the graph ${\mathcal G}$
\cite{Coutino17, jiang19, aliaksei13,  segarra17}.
% Illustrative examples of graph  shifts on a graph ${\mathcal G}$ are the adjacency matrix  ${\bf A}_{\mathcal G}$,
%	%the degree matrix  ${\bf D}_{\mathcal G}$ of the graph ${\mathcal G}$,
%	Laplacian matrix ${\bf L}_{\cal G}:={\bf D}_{\mathcal G}-{\bf A}_{\mathcal G}$,
%symmetric	normalized Laplacian matrix ${\bf L}^{\rm sym}_{\cal G}={\bf D}_{\mathcal G}^{-1/2}{\bf L}_{\cal G}{\bf D}_{\mathcal G}^{-1/2}$
% %node-variant graph filters \cite{segarra17} and edge-variant graph filters \cite{Coutino17}
%  and their
%	variants.
The filtering procedure
	${\bf x}\longmapsto {\bf S}{\bf x}$ associated with a graph shift ${\bf S}=(S(i,j))_{i,j\in V}$
	is a local operation that updates  signal value  at each vertex $i\in V$
	by a ``weighted" sum of signal values at  {\bf adjacent} vertices  $j\in {\mathcal N}_i$,
%	\vspace{-.5em}
\begin{equation*}
\label{weightedsum}
\tilde{x}(i)=\sum_{j\in\mathcal{N}_i}S(i,j)x(j), %\vspace{-.5em}
\end{equation*}
	where ${\bf x}=(x(i))_{i\in V}$, ${\bf S} {\bf x}=(\tilde {x}(i))_{i\in V}$, and ${\mathcal N}_i$
is the set of  adjacent vertices of $i\in V$. %, see Algorithm \ref{shiftprocedure.algorithm}.
%%%%%%%%%%%%%%%%%%%%%%%
% \begin{algorithm}[t]
%\caption{Realization of the filtering procedure for a graph shift ${\mathbf S}$ at a vertex $i\in V$. }
%\label{shiftprocedure.algorithm}
%\begin{algorithmic}  %[1]
%
%\STATE {\bf Inputs}: The $i$-th row $s(i,j), j\in {\mathcal N}_i$ of the shift ${\bf S}$ and
%the signal value $x(i)$  of the input signal ${\bf x}$ at the vertex $i$.
%
%%\STATE {\bf Operation}: Evaluate $m_k=\mu(B(k, r))$, compute ${\bf F}_k= {\bf H}_{0,k}^T{\bf H}_{0,k}+ {\bf H}_{1,k}^T{\bf H}_{1,k}$,
%% find its inverse  $({\bf F}_k)^{-1}$, and then compute $ {\bf G}^L_{l; k}:=({\bf F}_k)^{-1} {\bf H}_{l,k}^T, l=0, 1$.
%
%%  , and a local approximation
%%${\bf G}_k=(\tilde f_k(i,j))_{i,j\in B(k, 2r)}$  to the matrix ${\bf H}$
%% and compute ${\bf F}_k= {\bf H}_{0,k}^T{\bf H}_{0,k}+ {\bf H}_{1,k}^T{\bf H}_{1,k}$ and
%%$({\bf F}_k)^{-1}=(\tilde f_k(i,j))_{i,j\in B(k, 2r)}$
%
%\STATE{\bf 1)} Send $x(i)$ to all adjacent vertices $j\in {\mathcal N}_i$ and receive $x(j)$ from adjacent vertices $j\in {\mathcal N}_i$.
%\STATE{\bf 2)} Use \eqref{weightedsum} to calculate  $\tilde{x}(i)$.
%
%\STATE {\bf Output}: The signal value $\tilde x(i)$  of the output signal ${\bf x}$ at the vertex $i$.
%\end{algorithmic}  %\vspace{-.03in}
%\end{algorithm}
%%%%%%%%%%%%%%%%%%%%%%%%%%%%%%%%%%%%%%%%%%%%
		The above local implementation of filtering procedure    has been extended to a polynomial graph filter
	${\bf H}= \sum_{l=0}^L h_l {\bf S}^l $
	of the shift  ${\bf S}$,
%{\color{re \cite{Emirov19,Leus17, shuman18, Shi15}(more references here?)},
	\begin{equation} \label{iterativedistributedalgorithm.eqn100}
	\left\{ \begin{array}{l}
	{\bf z}^{(0)}= h_L  {\bf x}, \\
	{\bf z}^{(n+1)}= h_{L-n-1} {\bf x}+ {\bf S} {\bf z}^{(n)},\  n=0, \ldots, L-1,\\
	{\bf H}{\bf x}={\bf z}^{(L)},
	\end{array} \right.
	\end{equation}
	where the filtering procedure
${\bf x}\longmapsto {\bf H}{\bf x}$ is divided into $(L+1)$-steps with the procedure
	in each step
	being a local operation \cite{ Emirov19,  Leus17,  shuman18, Waheed18}.
%(\cite{moura14, isufi17, shuman18,  Leus17,  Emirov19}).
 The %block diagram of the above implementation
%is shown in  Figure \ref{onepolynomialfiltering_structure}, and
 parallel  realization of the above implementation \eqref{iterativedistributedalgorithm.eqn100} at the vertex level  is presented in   Algorithm  \ref{singleshiftprocedure.algorithm}.
	In this section, we  extend the above synchronized implementation at the vertex level
to the filtering procedure associated with  a polynomial graph filter ${\bf H}$
of multiple   shifts, and propose a recursive algorithm containing
  about $\sum_{m=0}^{d-1} \prod_{k=1}^{m+1}(L_k+1)$ %{\color{blue} or $L_k+1$}
steps with
 the output  value  at each vertex  in each step being updated from some weighted sum of the input  values at  adjacent vertices
 of its preceding step,
 see Algorithm \ref{distributed_Hx.algorithm}.
	
 Let ${\bf S}_k=(S_k(i,j))_{i,j\in V}, 1\le k\le d$, be commutative graph shifts and ${\bf H}$ be the polynomial graph filter  in \eqref{MultiShiftPolynomial} with $d\ge 2$.  For $1\le m\le d-1$
%  $N$ be the order of the graph ${\mathcal G}$.
%\cc{Reviewer suggests to define $v_m$ first, and also change the notation. }
%	\vspace{-.5em}\begin{eqnarray} \label{multi_index_order}
%	v_{m}(l_1, \ldots, l_{m}) & \hskip-0.08in = & \hskip-0.08in  1+ l_{m}+ (L_{m}+1) l_{m-1}\nonumber\\
%& \hskip-0.08in  &  +\cdots+
%\prod_{k'=2}^{m}(L_{k'}+1)l_1
%	\vspace{-.5em}\end{eqnarray}
we use
	\begin{equation} \label{multi_index_order}
	v_{m}(l_1, \ldots, l_{m}) =   l_{m}+ l_{m-1}(L_{m}+1)  +\cdots+ l_1
\prod_{k=2}^{m}(L_{k}+1)
\end{equation}
to denote the lexicographical order of $(l_1, \ldots, l_{m})$ with $0 \leq l_k \leq L_k, 1\le k\le m$.
Now we define a matrix ${\bf U}_{d-1}$   of size $N\times \prod_{k=1}^{d-1}(L_{k}+1)$
with its $v_{d-1}(l_1, \ldots, l_{d-1})$-th column given by
	\begin{equation} \label{algorithm.eqn0}
{\bf U}_{d-1}\big(:, v_{d-1}(l_1, \ldots, l_{d-1})\big)
	=\sum^{L_d}_{l_d=0} h_{l_1, \ldots, l_{d-1}, l_d} {\bf S}^{l_d}_d {\bf x}.
	\end{equation}
%
%
%
%
%
%% Define a matrix ${\bf U}_{d-1}$   of size $N\times \prod_{k=1}^{d-1}(L_{k}+1)$
%%with its $v_{d-1}(l_1, \ldots, l_{d-1})$-th column given by
%%	\begin{equation} \label{algorithm.eqn0}
%%{\bf U}_{d-1}\big(:, v_{d-1}(l_1, \ldots, l_{d-1})\big)
%%	=\sum^{L_d}_{l_d=0} h_{l_1, \ldots, l_{d-1}, l_d} {\bf S}^{l_d}_d {\bf x},
%%	\end{equation}
%%where  for $1\le m\le d-1$,
%%%	\vspace{-.5em}\begin{eqnarray} \label{multi_index_order}
%%%	v_{m}(l_1, \ldots, l_{m}) & \hskip-0.08in = & \hskip-0.08in  1+ l_{m}+ (L_{m}+1) l_{m-1}\nonumber\\
%%%& \hskip-0.08in  &  +\cdots+
%%%\prod_{k'=2}^{m}(L_{k'}+1)l_1
%%%	\vspace{-.5em}\end{eqnarray}
%%	\begin{equation} \label{multi_index_order}
%%	v_{m}(l_1, \ldots, l_{m}) =   l_{m}+ l_{m-1}(L_{m}+1)  +\cdots+ l_1
%%\prod_{k=2}^{m}(L_{k}+1)
%%\end{equation}
%%is  the lexicographical order of $(l_1, \ldots, l_{m})$ with $0 \leq l_k \leq L_k, 1\le k\le m$.
%
%
%
%
%
Follow the procedure in \eqref{iterativedistributedalgorithm.eqn100}, we can
evaluate ${\bf U}_{d-1}(:,  v_{d-1}(l_1, \ldots, l_{d-1}))$
	in $(L_d+1)$-steps with the filtering procedure
	in each step
	being a local operation,
 see Step 1 in Algorithm \ref{distributed_Hx.algorithm} for the distributed implementation at vertex level.
Moreover, one may verify that
 \begin{equation} \label{algorithm.eqn1}
{\bf H}{\bf x}= \sum_{ l_1=0}^{L_1}\cdots \sum_{l_{d-1}=0}^{L_{d-1}}
 {\bf S}_1^{l_1} \cdots  {\bf S}_{d-1}^{l_{d-1}}  {\bf U}_{d-1}(:, v_{d-1}( l_1, \ldots,  l_{d-1}))
\end{equation}
by \eqref{MultiShiftPolynomial} and \eqref{algorithm.eqn0}.
By induction on $m=d-2, \ldots, 1$, we define matrices ${\bf U}_m$ of size $N\times \prod_{k'=1}^m(L_{k'}+1)$
 by  %For $2\le m\le d-1$  we  define matrices ${\bf U}_m$
	\begin{equation} \label{algorithm.eqn2}
 {\bf U}_m\big(:, v_{m}(l_1, \ldots, l_{m})\big) 	=   \sum_{ l_{m+1}=0}^{L_{m+1}} {\bf S}_{m+1}^{l_{m+1}}  {\bf U}_{m+1}\big(:, v_{m+1}( l_1, \ldots, l_{m}, l_{m+1})\big)
	\end{equation}
%	
%	\vspace{-.5em}\begin{eqnarray} \label{algorithm.eqn2}
%\hskip-0.18in  & \hskip-0.08in 	 & \hskip-0.08in {\bf U}_m\big(:, v_{m}(l_1, \ldots, l_{m})\big)\nonumber\\
%% & \hskip-0.08in 	=  & \hskip-0.08in \sum^{L_{m}}_{l_{m}=0}
%%\cdots \sum^{L_d}_{l_d=0} h_{l_1, \ldots, l_{m-1}, l_m, \ldots,  l_d}\nonumber\\
%% & \hskip-0.08in 	  &  \qquad \quad \times {\bf S}_{m}^{l_m}\cdots {\bf S}^{l_d}_d {\bf x}\nonumber\\
%\hskip-0.18in & \hskip-0.08in 	=  & \hskip-0.08in \sum_{ l_{m+1}=0}^{L_{m+1}}
% {\bf S}_{m+1}^{l_{m+1}}  {\bf U}_{m+1}\big(:, v_{m+1}( l_1, \ldots, l_{m}, l_{m+1})\big)
%	\vspace{-.5em}\end{eqnarray}	
%
where  $0 \leq l_k \leq L_k, 1\le k\le m$.
 By induction on $m=d-2, \ldots, 1$
 we obtain from  \eqref{algorithm.eqn2} that
every column of the matrix ${\bf U}_m$
 can be %${\bf U}_m(:, l_1, \ldots, l_{m-1}), 0\le l_1\le L_1, \ldots, 0\le l_{m-1}\le L_{m-1}$
 evaluated from
 ${\bf U}_{m+1}$ %(:, l_1, \ldots, l_{m-1}, l_m), 0\le l_1\le L_1, \ldots, 0\le l_{m}\le L_{m}$
in $(L_{m+1}+1)$-steps,  see Step 3 in Algorithm \ref{distributed_Hx.algorithm}  for the distributed implementation at vertex level.
%%with the filtering procedure
%%	in each step	being a local operation,
%	\vspace{-.4em}\begin{equation} \label{iterativedistributedalgorithm.eqn11}
%	\left\{ \begin{array}{l}
%	{\bf w}^{(0)}=  {\bf U}_{m+1}(:, l_1, \ldots, l_{m-1}, L_m), \\ % {\bf y}_{m(l_1,\cdots,l_{d'-1}, L_{d'})}^{(d'+1)}, \\
%	{\bf w}^{(n)}= {\bf U}_{m+1}(:, l_1, \ldots, l_{m-1}, L_m-n)
% + {\bf S}_{m} {\bf w}^{(n-1)}\\
%  \ \ {\rm for} \ n=1, \ldots, L_{m},\\
%	{\bf U}_{m}(:, l_1, \ldots, l_{m-1})={\bf w}^{(L_m)}.
%	\end{array} \right.
%	\vspace{-.4em}\end{equation}
By \eqref{algorithm.eqn1}
and
\eqref{algorithm.eqn2},  we can prove
  \begin{equation} \label{algorithm.eqn3}
{\bf H}{\bf x}= \sum_{ l_1=0}^{L_1}\cdots \sum_{l_{m}=0}^{L_{m}}
 {\bf S}_1^{l_1} \cdots  {\bf S}_{m}^{l_{m}}  {\bf U}_{m}(:, v_m( l_1, \ldots,  l_{m}))
\end{equation}
by induction on $m=d-2, \ldots, 1$.
Taking $m=1$ in \eqref{algorithm.eqn3} yields
	\begin{equation}\label{algorithm.eqn4}
	{\bf H}{\bf x}=\sum_{ l_1=0}^{L_1} {\bf S}_1^{l_1} {\bf U}_1(:, l_1). % {\bf y}_{m(l_1,\cdots,l_{d-1})}^{(d)}.
	%\Big(\sum_{ l_2=0}^{L_2}  h_{m(l_1, \ldots, l_{d-1}, l_d)}{\bf S}_d^{l_d}\Big)
	%=:\sum_{l_1=0}^{L_1} {\bf H}_{l_1}({\bf S}_2) {\bf S}_1^{l_1},
	\end{equation}
By \eqref{algorithm.eqn4},
 we finally evaluated    the output   ${\bf H}{\bf x}$ of the filtering procedure
%	\vspace{-.5em}\begin{equation} \label{firstformula.eq02}
%	{\bf H}{\bf x}= \sum^{L_1}_{l_1=0}  {\bf S}^{l_{1}}_{1} {\bf y}_{m(l_1)}^{(2)},
%	\vspace{-.5em}\end{equation}
%	which can be obtained
from the matrix ${\bf U}_1$
	in $(L_1+1)$-steps with
the filtering procedure
	in each step
	being a local operation, see Step 4 in Algorithm \ref{distributed_Hx.algorithm} for the  implementation at vertex level.

	\begin{algorithm}[h]
		\caption{Synchronous realization of the filtering procedure ${\bf x}\longmapsto {\bf H}{\bf x}$
			for the polynomial filter  ${\bf H}$ of multiple graph shifts
%			${\bf H}=
%			\sum^{L_1}_{l_1=0} \cdots \sum_{l_d=0}^{L_d} h_{l_1,\ldots, l_d} {\bf S}^{l_1}_1 \cdots {\bf S}_d^{l_d}$
at a vertex $i\in V$. % where $d\ge 2$.
%			of three commutative graph shifts ${\bf S}_1, {\bf S}_2$ and ${\bf S}_2$.
}
		\label{distributed_Hx.algorithm}
		\begin{algorithmic}  %[1]
			\STATE {\bf Inputs}: \ Polynomial coefficients $h_{l_1,\ldots, l_d}, 0\le l_1\le L_1, \ldots, 0\le l_d\le L_d$ of the polynomial filter ${\bf H}$ in \eqref{MultiShiftPolynomial}, entries   $S_k(i,j), j\in {\mathcal N}_i$ of the $i$-th row of
			graph shifts ${\bf S}_k=(S_k(i,j))_{i,j\in V}, 1\le k\le d$, and the  value $x(i)$ of the input graph signal ${\bf x}=(x(k))_{k\in V}$ at vertex $i$.\\
			
			\STATE {\bf Step 1}:\  { Find the $i$-th row of the matrix ${\bf U}_{d-1}$.}\\
%=(U_{d-1}(i,j))_{i\in V, 0\le j\le  \prod_{k=1}^{d-1}(L_{k}+1)-1} $.

\hspace{0.5cm} {\bf for} $p=0, 1, \ldots,  \prod_{k=1}^{d-1}(L_{k}+1)-1 $\\ %(L_1+1)\cdots (L_{d-1}+1)-1$\\
\hspace{0.8cm} {\bf Step 1a}: Write $p= v_{d-1} (l_1, \ldots, l_{d-1})$
%\sum^{d-1}_{k=1} \big(\prod_{k'=k+1}^{d}(L_{k'}+1)\big)l_k+l_{d-1}$
for some $0\le l_k\le L_k, 1\le k\le d-1$.

\hspace{0.8cm} {\bf Step \hspace{-.05cm}1b}:\ Apply  Algorithm  \ref{singleshiftprocedure.algorithm}
with polynomial coefficients and  entries of the graph shift being replaced by
 polynomial coefficients $h_{l_1, \ldots, l_{d-1},  l_d}, 0\le l_d\le L_d$,  and entries
$S_d(i,j),j\in N_i$ in the $i$-th row of the shift ${\bf S}_d$, and denote
the corresponding output by $z^{(L_d)}(i)$.
%, and the value $x(i)$ of the input signal ${\bf x} = (x(i))i\in V$ at the vertex $i$.

% $z^{(0)}(i)=h_{l_1, \ldots, l_{d-1},  L_d} * x(i)$ and $n=0$. \\ %({\bf U}^3(i, j(M_3+1)))_{i\in \mathcal{N}_q^3}$\\
%			
%\hspace{0.8cm} {\bf Step 1b}:\  send $z^{(n)}(i)$ to its adjacent vertices $j\in {\mathcal N}_i$
%and receive $z^{(n)}(j)$ from its  adjacent vertices $j\in {\mathcal N}_i$.\\
%
%\hspace{0.8cm} {\bf  Step 1c}:\ update
% $z^{(n+1)}(i)=h_{l_1, \ldots, l_{d-1},  L_d-n-1}*x(i) + \sum\limits_{j\in {\mathcal N}_i} s_d(i, j) * z^{(n)}(j).$\\
%
% \hspace{0.8cm} {\bf  Step 1d}:\ set $n=n+1$ and return to  {\bf Step   1b} if $n\le L_d-1$.

\hspace{0.8cm} {\bf  Step 1c}:\  Set ${\bf U}_{d-1}(i,p)=z^{(L_d)}(i)$.\\

\hspace{0.5cm}{\bf end}

			\STATE {\bf Step 2}:\  {\bf if}   $d=2$,
 set  ${\bf W}(i, j)={\bf U}_{d-1}(i, j), 0\le j\le L_1$ and do {\bf Step 4}, {\bf otherwise} do {\bf Step 3}.

\STATE {\bf Step 3}:  {\ Find  the $i$-th row of the matrix ${\bf U}_{m}, d-2\ge  m\ge 1$.}\\

\hspace{0.5cm} {\bf for} $m=d-2, \ldots, 2, 1$\\
\hspace{0.8cm}
{\bf for} $p=0, 1, \ldots, \prod_{k=1}^{m}(L_{k}+1)-1$\\ % (L_1+1)\cdots (L_{m}+1)-1$\\
%\hspace{0.7cm} Write $j-1=\sum^{m-1}_{k=1} \big(\prod_{k'=k+1}^{m}(L_{k'}+1)\big)l_k+l_{m-1}$ for some $0\le l_k\le L_k, 1\le k\le m-1$,\\
		%\hspace{1cm}  % $l_1=\lfloor (j-1)/(L_2+1)\rfloor$, $l_2=j-1-(L_2+1)l_1$
\hspace{1cm} {\bf  Step 3a}: Apply   Algorithm  \ref{singleshiftprocedure.algorithm}
with polynomial coefficients, entries of the graph shift and the value of input being  replaced by
 polynomial coefficients $h_{l}=1, 0\le l\le L_{m+1}$,  entries
$S_{m+1}(i,j),j\in N_i$ in the $i$-th row of the shift ${\bf S}_{m+1}$,
and the value $z^{(0)}(i)={\bf U}_{m+1}\big(i, p(L_{m+1}+ 1) + L_{m+1}\big) $ of the $(p(L_{m+1}+ 1) + L_{m+1})$-column of the matrix
${\bf U}_{m+1}$,
 and denote
the corresponding  output by $z^{(L_{m+1})}(i)$.  %$z^{(L_d-1)}(i)$.

%{\bf Step 3a}:\  set  $w^{(0)}(i)={\bf U}_{m+1}\big(i,p(L_{m+1}+ 1) + L_{m+1}\big) $ and $n=0$. \\ %({\bf U}^3(i, j(M_3+1)))_{i\in \mathcal{N}_q^3}$\\
%			
%\hspace{1cm} {\bf Step 3b}:\  send $w^{(n)}(i)$ to its adjacent vertices $j\in {\mathcal N}_i$
%and receive $w^{(n)}(j)$ from its  adjacent vertices $j\in {\mathcal N}_i$.\\
%
%\hspace{1cm} {\bf  Step 3c}:\ update
% $z^{(n+1)}(i)={\bf U}_{m+1}(i,p(L_{m+1}+ 1) + L_{m+1}-n-1)  + \sum\limits_{j\in {\mathcal N}_i} s_{m+1}(i, j) * w^{(n)}(j).$\\
%
% \hspace{1cm} {\bf  Step 3d}:\ set $n=n+1$ and return to  {\bf Step   3b} if $n\le L_{m+1}-1$.

\hspace{1cm} {\bf  Step 3b}:\  Set ${\bf U}_{m}(i,p)=z^{(L_{m+1})}(i)$.\\

			\hspace{0.8 cm} {\bf end}\\
\hspace{0.5 cm} {\bf end}\\
\hspace{0.3 cm} Set  ${\bf W}(i, j)={\bf U}_{1}(i, j), 0\le j\le L_1$. %Set ${\bf W}={\bf U}_1$.
			\\
			\STATE {\bf  Step 4}:\  { Find the value  of the output signal ${\bf H}{\bf x}$ at vertex $i$.}\\
%$i$-th row of the matrix ${\bf U}_{d-1}$.}\\

\hspace{1cm} %\cc{I can not get this correctly. apply coefficients and then replace?}
{\bf  Step 4a}: Apply  Algorithm  \ref{singleshiftprocedure.algorithm}
with polynomial coefficients, entries of the graph shift and the value of input being   replaced by
 polynomial coefficients $h_{l}=1, 0\le l\le L_{1}$,  entries
$S_{1}(i,j),j\in N_i$ in the $i$-th row of the shift ${\bf S}_{1}$,
and the value $u^{(0)}(i)={\bf W}\big(i, L_{1}\big) $ of the $L_1$-column of the matrix
${\bf W}$.

\hspace{1cm} %\cc{I can not get this correctly. apply coefficients and then replace?}
{\bf  Step 4b}: Denote
the corresponding  output by $u^{(L_{1})}(i)$.

 %\hspace{0.5 cm}  {\bf Step 4a}:\    set $u^{(0)} ={\bf W}(i, L_1)$\\
%		\hspace{0.5cm} {\bf Step 4b}:\  send $u^{(n)}(i)$ to its adjacent vertices $j\in {\mathcal N}_i$
%and receive $u^{(n)}(j)$ from its  adjacent vertices $j\in {\mathcal N}_i$.\\
%
%\hspace{0.5cm} {\bf  Step 4c}:\ update
% $z^{(n+1)}(i)={\bf W}(i, L_{1}-n-1)  + \sum\limits_{j\in {\mathcal N}_i} s_{1}(i, j) * u^{(n)}(j).$\\
%
% \hspace{0.5cm} {\bf  Step 4d}:\ set $n=n+1$ and return to  {\bf Step   4a} if $n\le L_1-1$.
%
%%\hspace{1cm} {\bf  Step 4d}:\  set ${\bf U}_{m}(i,p)=z^{(L_{m+1})}(i)$.\\

			{\bf Output:}  The  value $ \tilde x(i)=u^{(L_1)}(i)$  is  the output signal ${\bf H}{\bf x}=(\tilde x(i))_{i\in V}$ at the vertex $i$.
		\end{algorithmic}
\end{algorithm}

	 Denote   the degree of the graph ${\mathcal G}$ by  $\deg{\mathcal G}$, and for two positive quantities $a$ and $b$, we denote $a=O(b)$ if $a\le Cb$ for some absolute constant $C$, which is always independent of the order $N$ of the graph ${\mathcal G}$ and it could be different at different occurrences.
 Recall from the definition of a graph shift on a graph ${\mathcal G}$ that
 the number of nonzero entries in  every row of a graph shift on the graph ${\mathcal G}$ is no more than  $\deg{\mathcal G}+1$.
  %of the graph ${\mathcal G}$.
To implement \eqref{algorithm.eqn0}, \eqref{algorithm.eqn2} and \eqref{algorithm.eqn4}
 in a central facility,  the operations of addition and multiplication %needed
are about  $2 N  (\deg {\mathcal G}+1)  \prod_{k=1}^{d}(L_k+1)  $,
$  2N (\deg {\mathcal G}+1) \sum_{m=1}^{d-2} \prod_{k=1}^{m+1}(L_k+1)$ and
$2 N  (\deg {\mathcal G}+1) (L_1+1)$ respectively,
and  memory required
are about $d(\deg {\mathcal G}+1)N+\prod_{k=1}^{d}(L_k+1)+ 2N+ N  \sum_{m=0}^{d-1} \prod_{k=1}^{m}(L_k+1)$
to store  the graph shifts ${\bf S}_1,\dots, {\bf S}_d$,
	the polynomial coefficients of the polynomial graph filter  ${\bf H}$,
	the original graph signal ${\bf x}$,  the output ${\bf H}{\bf x}$ of the filtering procedure and
matrices  ${\bf U}_m, 1\le m\le d-1$,  in \eqref{algorithm.eqn0}, \eqref{algorithm.eqn2} and \eqref{algorithm.eqn3}.
Hence for the implementation of the filter procedure
${\bf x}\longmapsto {\bf H}{\bf x}$ in a central facility via applying \eqref{algorithm.eqn0}, \eqref{algorithm.eqn2} and \eqref{algorithm.eqn4},
the total computational cost is about $O\big( N \deg {\mathcal G} + (N+L_d+1) \prod_{k=1}^{d-1}(L_k+1)\big)$
and
the memory requirement is about
$O\big(N (\deg {\mathcal G}+1) \prod_{k=1}^{d}(L_k+1)\big)$.

	%\begin{remark}
Shown in Algorithm \ref{distributed_Hx.algorithm} is the  implementation
of
\eqref{algorithm.eqn0}, \eqref{algorithm.eqn2} and \eqref{algorithm.eqn4}
 at the vertex level. Hence it is implementable in  a distributed network where each agent is equipped with a data processing subsystem  for limited data storage and computation power, and a communication subsystem for
	direct data exchange to its adjacent vertices.
Denote the cardinality of a set $E$ by $\# E$.
To implement Algorithm \ref{distributed_Hx.algorithm} in a distributed network, we see
that the data processing subsystem  at a vertex $i\in V$ performs
 about $O\big(   (\# {\mathcal N}_i+1)\sum_{m=0}^{d-1} \prod_{k=1}^{m+1}(L_k+1)\big )=O\big((\deg {\mathcal G}+1) \prod_{k=1}^{d}(L_k+1)\big)$  operations of addition and multiplication, %\cc{Big O stands for about, do we still want to write about $O(\sharp)$. }
and it  stores  data of size about $O\big( \prod_{k=1}^{d}(L_k+1)+ (\# {\mathcal N}_i+1)( d+2+\sum_{m=0}^{d-1} \prod_{k=1}^{m}(L_k+1))\big ) = O\big((\deg {\mathcal G}+L_d+1) \prod_{k=1}^{d-1}(L_k+1)\big)$, including
	 polynomial coefficients of the  filter  ${\bf H}$,
	the $i$-th row of  graph shifts ${\bf S}_1,\dots, {\bf S}_d$,
	and the $i$-th and its adjacent $j$-th components of the original graph signal ${\bf x}$,
the output ${\bf H}{\bf x}$ of the filtering procedure and
	the matrices ${\bf U}_m, 1\le m\le d-1$, where $j\in {\mathcal N}_i$.
Comparing the implementation of \eqref{algorithm.eqn0}, \eqref{algorithm.eqn2} and \eqref{algorithm.eqn4}  in a central facility, the total computational cost to implement
%Algorithm \ref{distributed_Hx.algorithm} in a  distributed network is almost the same, while the total memory is slightly larger, since the
%polynomial coefficients of the polynomial filter  ${\bf H}$ needs to be stored at every agent in a distributed  network and
%it is enough to store one copy of the coefficients  in a central facility.
 Algorithm \ref{distributed_Hx.algorithm} in a  distributed network is almost the same, while the total memory is slightly  large, since the
polynomial coefficients of the polynomial graph filter  ${\bf H}$ need to be stored at every agent in a distributed  network while
only one copy of the coefficients needs to be stored   in a central facility.
In addition to data processing in a central facility, the implementation of Algorithm \ref{distributed_Hx.algorithm} in a distributed network requires
that every agent $i\in V$ communicates with its adjacent agents $j\in {\mathcal N}_i$ with
the $j$-th components of the original graph signal ${\bf x}$, matrices ${\bf U}_m, 1\le m\le d-1$ and the output ${\bf H}{\bf x}$ of filtering procedure, which is  about $O\big(   \# {\mathcal N}_i \prod_{k=1}^d (L_k+1)\big )= O( (\deg {\mathcal G}+1)\prod_{k=1}^d (L_k+1))$  loops.
We observe that for the implementation of  the proposed  Algorithm \ref{distributed_Hx.algorithm} in a distributed network,  the computational cost, memory requirement
 and communication expense for the data processing and communication subsystems equipped at each agent is {\bf independent} on the size $N$ of the network.

	\section{Inverse filtering and iterative  approximation  algorithm}
\label{iterativeapproximation.section}

	%where ${\bf b}$ be a graph signal and   ${\bf H}=h({\bf S}_1,...,{\bf S}_d)$ is  a polynomial filter of  graph shifts
	%${\bf S}_1,...,{\bf S}_d$ satisfying \eqref{commutativityS}
	% where ${\bf H}$ is a polynomial of the graph shifts ${\bf S}_1,...,{\bf S}_d$     {\color{red} \cite{segarra17, narang13,shuman18}}.

	%
	%  In this scenario, the centralized processing manner is infeasible since there is no strong central node to govern all computation over network.  The distributed processing manner that only involving local interaction between agents is indispensable approach for SDNs. Many operators on graph can be naturally implemented distributively, such as the graph filtering with filters being the finite polynomials of graph shift matrix \cite{segarra17}. However, not all operators can be realized in a distributed manner, for instance, the operator in graph denoising involves an inverse filtering that is not polynomials with respect to the graph shift matrix \cite{shuman18}.

%Let ${\mathcal G}=(V, E)$ be a connected, undirected and unweighted graph and
Let  ${\bf H}$ be an invertible graph filter on the graph ${\mathcal G}$.
 %on the \cc{ a connected, undirected and unweighted graph ? as it has been mentioned in section 4. }graph ${\mathcal G}$.
	In  some applications,  such as  signal denoising,  inpainting,  smoothing, reconstructing and  semi-supervised learning
%\cite{shuman18, jiang19}, \cite{mario19}--\cite{Emirov19}, \cite{ Shi15, sihengTV15, siheng_inpaint15},
{\cite{siheng_inpaint15,sihengTV15,  Emirov19, mario19, Leus17,  jiang19, Lu18, Shi15,shuman18}, }
an inverse filtering procedure
%	\vspace{-.4em}
\begin{equation}\label{inverseprocedure}
	{\bf x}= {\bf H}^{-1} {\bf b}
%	\vspace{-.4em}
\end{equation}
is involved.
	 In this section, we select a graph filter  ${\bf G}$ which provides an approximation to the inverse
filter ${\bf H}^{-1}$,  propose an iterative approximation algorithm
with each iteration including filtering procedures associated with filters ${\bf H}$ and ${\bf G}$,
%\eqref{iterativedistributedalgorithm.eqn1} and \eqref{iterativedistributedalgorithm.eqn2},
and show that the proposed algorithm  \eqref{iterativedistributedalgorithm.eqn1} and \eqref{iterativedistributedalgorithm.eqn2}
converges exponentially. % to the output of the inverse filtering procedure \eqref{inverseprocedure}.
The challenge to apply  the iterative approximation  algorithm \eqref{iterativedistributedalgorithm.eqn1}
	and \eqref{iterativedistributedalgorithm.eqn2}
%	 to  implement the inverse filtering procedure  \eqref{inverseprocedure}
is how to select
   the  filter ${\bf G}$ to approximate
the inverse filter ${\bf H}^{-1}$ appropriately, which will be discussed in the next section when ${\bf H}$
is a polynomial filter of commutative graph shifts.

 %  The selection of the approximation filter ${\bf G}$ will be discussed in the next section when ${\bf H}$ is a polynomial filter of multiple graph shifts.

%Denote the identity matrix by ${\bf I}$  and
Denote the spectral radius of a matrix ${\bf A}$ by $\rho({\bf A})$.
Take a graph filter 	${\bf G}$ such that
	the spectral radius of ${\bf I}-{\bf H}{\bf G}$  is strictly less than 1, i.e.,
	\begin{equation}\label{sigmaHG.eq} \rho({\bf I}-{\bf H}{\bf G})<1.
	\end{equation}
By Gelfand's formula on spectral radius,
  the requirement \eqref{sigmaHG.eq}  can be reformulated as
	\begin{equation}\label{sigmaHG.eq2}
	\rho({\bf I}-{\bf H}{\bf G})=\lim_{n\to \infty} \|({\bf I}-{\bf  HG})^n\|_2^{1/n}<1,
	\end{equation}
  where $\|{\bf x}\|_2$ is Euclidean  norm of a vector ${\bf x}$ %=(x(i))_{i\in V}$
  and $\|{\bf A}\|_2=\sup_{\|{\bf x}\|_2=1} \|{\bf A}{\bf x}\|_2$ is the operator norm of a matrix ${\bf A}$.
	By \eqref{sigmaHG.eq2}, we can rewrite the inverse filtering procedure
\eqref{inverseprocedure}  as
	\begin{equation}\label{inverseprocedure.new}
	{\bf x}= {\bf G} \big({\bf I}-({\bf I}-{\bf HG})\big)^{-1}{\bf b}= {\bf G} \sum_{n=0}^\infty ({\bf I}-{\bf HG})^n {\bf b}
	\end{equation}
	%where the convergence follows from \eqref{HG.eq}.
by applying Neumann series to ${\bf I}-{\bf HG}$.	Based on the above expansion, %representation,
 we propose the following iterative algorithm to implement the inverse filtering procedure \eqref{inverseprocedure}:
	\begin{equation} \label{iterativedistributedalgorithm.eqn1}
	\left\{ \begin{array}{l}
	{\bf z}^{(m)}= {\bf G} {\bf e}^{(m-1)}, \\
	{\bf e}^{(m)}={\bf e}^{(m-1)}-{\bf H} {\bf z}^{(m)},\\
	{\bf x}^{(m)}={\bf x}^{(m-1)}+ {\bf z}^{(m)}, \ m\ge 1,
	\end{array} \right.
	\end{equation}
	with initials
	\begin{equation} \label{iterativedistributedalgorithm.eqn2}
	{\bf e}^{(0)}={\bf b} \ \  {\rm and} \ \ {\bf x}^{(0)}={\bf 0}.
	\end{equation}
	Due to the approximation property \eqref{sigmaHG.eq}  of the graph filter ${\bf G}$ to the inverse filter ${\bf H}^{-1}$, we call the above algorithm \eqref{iterativedistributedalgorithm.eqn1}
	and \eqref{iterativedistributedalgorithm.eqn2} as an {\em iterative approximation algorithm}.
   In the following theorem, we show that  the requirement \eqref{sigmaHG.eq} for the approximation filter is a sufficient and necessary condition for
the exponential convergence of 		the iterative approximation algorithm \eqref{iterativedistributedalgorithm.eqn1}
	and \eqref{iterativedistributedalgorithm.eqn2}. %, see Appendix \ref{convergence_ICPA.appendix} for the proof.

%	Therefore we conclude from \eqref{sigmaHG.eq2} and \eqref{xm.nondef} that
%
%	converges exponentially.  %the exponential convergence of

	\begin{theorem} \label{convergence_ICPA.thm} {\rm
		Let   ${\bf H}$ be an invertible graph filter and ${\bf G}$ be a graph filter.
		%$=h({\bf S}_1,...,{\bf S}_d)$ be  a polynomial filter of a graph shifts ${\bf S}_1,...,{\bf S}_d$.
		Then   ${\bf G}$ satisfies   \eqref{sigmaHG.eq}  if and only if
		for any graph signal ${\bf b}$, the sequence ${\bf x}^{(m)}, m\ge 1$, in the iterative approximation  algorithm \eqref{iterativedistributedalgorithm.eqn1}
		and \eqref{iterativedistributedalgorithm.eqn2} converges exponentially to
		% the   % optimal
		%solution
		${\bf H}^{-1} {\bf b}$.
	Furthermore, for any $r\in (\rho({\bf I}-{\bf H}{\bf G}), 1)$, there exists a positive constant $C$  such that
	\begin{equation}\label{desiredestimate}
		\|{\bf x}^{(m)}- {\bf H}^{-1} {\bf b} \|_2 \leq   C \|{\bf x}\|_2
		r^m, \ m\ge 1.
		%\varepsilon \|\tilde{\bf x}-{\bf x}^{(m-1)}\|_2.
	\end{equation}
	}\end{theorem}

\begin{proof} %First the sufficiency.
 $\Longrightarrow$:\
	Applying the first two equations in \eqref{iterativedistributedalgorithm.eqn1} gives
$${\bf e}^{(m)}=({\bf I}-{\bf H}{\bf G}){\bf e}^{(m-1)},\  m\ge 1.$$
 Applying the above expression repeatedly
and using the initial % ${\bf b}^{(0)}={\bf b}$
in  \eqref{iterativedistributedalgorithm.eqn2}
 yields
	\vspace{-0.4em}\begin{equation}\label{convergence_ICPA.thm.pf.eq0-1}
	{\bf e}^{(m)}  = %({\bf I}-{\bf H}{\bf G}){\bf b}^{(m-1)} =\cdots
({\bf I}-{\bf H}{\bf G})^{m}{\bf b}, \ m\ge 0.
	\vspace{-0.4em}\end{equation}
Combining \eqref{convergence_ICPA.thm.pf.eq0-1} and the first and third equations in \eqref{iterativedistributedalgorithm.eqn1} gives
$$ {\bf x}^{(m)}= {\bf x}^{(m-1)}+ {\bf  G} ({\bf I}-{\bf H}{\bf G})^{m-1}{\bf b},\  m\ge 1.$$
Applying the above expression for ${\bf x}^{(m)}, m\ge 1$, repeatedly
and using the initial % ${\bf b}^{(0)}={\bf b}$
in  \eqref{iterativedistributedalgorithm.eqn2}, we obtain
	\vspace{-0.6em}\begin{equation} \label{xm.def}
	{\bf x}^{(m)}=
	{\bf G} \sum_{n=0}^{m-1}
	({\bf I}-{\bf H}{\bf G})^n {\bf b},\  m\ge 1.  \vspace{-0.4em}\end{equation}
By \eqref{sigmaHG.eq2}, there exists a positive constant $C_0$ for any $r\in (\rho({\bf I}-{\bf H}{\bf G}), 1)$ such that
\begin{equation}  \label{xm.def22}
\|({\bf I}-{\bf H}{\bf G})^n\|_2\le C_0 r^n,\  n\ge 1.
\end{equation}
	Combining \eqref{sigmaHG.eq},
  \eqref{inverseprocedure.new} and \eqref{xm.def}, we obtain
	\vspace{-0.4em}\begin{equation} \label {xm.nondef}
	\|{\bf x}^{(m)}-{\bf x}\|_2
  =   \Big\|{\bf G} \sum_{n=m}^\infty ({\bf I}-{\bf H}{\bf G})^n {\bf b}\Big\|_2.
%\nonumber \\
\vspace{-0.4em}
\end{equation}
From \eqref{xm.def22} and \eqref{xm.nondef} it follows that
\begin{equation*} \label{xm.nondef2}
  \|{\bf x}^{(m)}-{\bf x}\|_2\le  \|{\bf G} \|_2  \|{\bf b}\|_2 \sum_{n=m}^\infty \|({\bf I}-{\bf H}{\bf G})^n\|_2
	 \le   C_0 \|{\bf G} \|_2 \|{\bf H}\|_2\|{\bf x}\|_2
\sum_{n=m}^\infty r^n\le
\frac{C_0\|{\bf G} \|_2 \|{\bf H}\|_2}{1-r}  r^m \|{\bf x}\|_2	\vspace{-0.4em}
\end{equation*}
%\begin{eqnarray*} \label{xm.nondef2}
%  & & \|{\bf x}^{(m)}-{\bf x}\|_2\le  \|{\bf G} \|_2  \|{\bf b}\|_2 \sum_{n=m}^\infty \|({\bf I}-{\bf H}{\bf G})^n\|_2
%	\nonumber\\
%	& \hskip-0.08in \le & \hskip-0.08in  C_0 \|{\bf G} \|_2 \|{\bf H}\|_2\|{\bf x}\|_2
%\sum_{n=m}^\infty r^n\le
%\frac{C_0\|{\bf G} \|_2 \|{\bf H}\|_2}{1-r}  r^m \|{\bf x}\|_2	\vspace{-0.4em}\end{eqnarray*}
 for all $m\ge 1$. This proves the exponential convergence of ${\bf x}^{(m)}, m\ge 0$ to ${\bf H}^{-1}{\bf b}$.

$\Longleftarrow$:\   Suppose on the contrary that \eqref{sigmaHG.eq} does not hold.  Then there exist
an eigenvalue  $\lambda$ of ${\bf I}-{\bf H}{\bf G}$  and an eigenvector ${\bf b}_0$
such that
\begin{equation}\label{Gb0.eq1} |\lambda|\ge 1 \ \ {\rm and}\  \ ({\bf I}-{\bf H}{\bf G}){\bf b}_0=\lambda {\bf b}_0.\end{equation}
 Then the sequence  ${\bf x}^{(m)}, m\ge 1$, in  the iterative approximation  algorithm \eqref{iterativedistributedalgorithm.eqn1}
		and \eqref{iterativedistributedalgorithm.eqn2} with ${\bf b}$ replaced by ${\bf b}_0$ becomes
\begin{equation*}\label{Gb0.eq2} \vspace{-0.4em}
{\bf x}^{(m)}=\Big(\sum_{n=0}^{m-1} \lambda^n\Big) {\bf G} {\bf b}_0=\left\{\begin{array}{ll}
\frac{\lambda^m-1}{\lambda-1} {\bf G} {\bf b}_0  & {\rm if} \ \lambda\ne 1\\
m {\bf G} {\bf b}_0 & {\rm if} \ \lambda=1\end{array}\right.
%\vspace{-0.4em}
\end{equation*}
 by \eqref{xm.def} and \eqref{Gb0.eq1}. Hence
 the sequence ${\bf x}^{(m)}, m\ge 1$, does not converge to the nonzero vector ${\bf H}^{-1} {\bf b}_0$, since it is identically zero if ${\bf G} {\bf b}_0= {\bf 0}$, and it diverges by the assumption that $|\lambda|\ge 1$ if ${\bf G} {\bf b}_0\ne {\bf 0}$.
 This contradicts to the exponential convergence assumption and completes the proof. % of the necessity.
\end{proof}

	%We postpone the proof of Theorem \ref{convergence_ICPA.thm} to Appendix \ref{convergence_ICPA.appendix}.
By  Theorem  \ref{convergence_ICPA.thm},
	the inverse filtering procedure \eqref{inverseprocedure}
	can be implemented by applying the iterative approximation  algorithm \eqref{iterativedistributedalgorithm.eqn1}
	and \eqref{iterativedistributedalgorithm.eqn2} with the  graph filter ${\bf G}$ being chosen so that \eqref{sigmaHG.eq} holds.

	\medskip

%	 Moreover, the corresponding iterative approximation algorithm  with ${\bf G}=\gamma {\bf I}$  coincides with the
%gradient descent method \eqref{gradientdescent.al} with zero initial, see Remark \ref{remark.grad}.

	We finish this section with two remarks on the comparison among
	the gradient descent method \cite{Shi15},
	the  autoregressive moving
	average (ARMA) method  \cite{Leus17}, and the  proposed
	iterative
	approximation algorithm \eqref{iterativedistributedalgorithm.eqn1} and \eqref{iterativedistributedalgorithm.eqn2}, cf. Remark \ref{iopa.re}.

	\begin{remark}\label{remark.grad}{\rm
			For a positive definite graph filter ${\bf H}$, % has its spectrum contained in  the positive axis,
			the inverse filtering procedure \eqref{inverseprocedure}
			can be implemented by the gradient descent method
			\begin{equation}\label{gradientdescent.al}
			{\bf x}^{(m)}= {\bf x}^{(m-1)}-\gamma ({\bf H} {\bf x}^{(m-1)}-{\bf b}),\ \  m\ge 1,
			\end{equation}
associated with  the unconstrained optimization problem having the objective function  $F({\bf x})={\bf x}^T {\bf H}{\bf x}-{\bf x}^T{\bf b}$,
where $\gamma$  is an appropriate step length and ${\bf x}^T$ is the transpose of a vector ${\bf x}$.
			The above iterative method
% is known			as FastIDIIR algorithm, which
 is shown in  \cite{Shi15} to be convergent when $0<\gamma<2/\alpha_2$
 and to have fastest convergence when $\gamma=2/(\alpha_1+\alpha_2)$,
 where $\alpha_1$ and $\alpha_2$ are the minimal and maximal eigenvalues
of the matrix ${\bf H}$, cf. Remark \ref{iopa.re}.
		By \eqref{gradientdescent.al}, we have that
			\begin{equation}\label{fastidiir}
			{\bf x}^{(m)}=\gamma \sum_{n=0}^{m-1} ({\bf I}-\gamma {\bf H})^n {\bf b}+ ({\bf I}-\gamma {\bf H})^m {\bf x}^{(0)},\  m\ge 1.
			\end{equation}
			By  \eqref {fastidiir} and \eqref{xm.def} in Theorem  \ref{convergence_ICPA.thm}, the sequence  ${\bf x}^{(m)}, m\ge 1$, in the gradient descent algorithm with zero initial {\bf coincides} with the
			sequence in  the iterative approximation algorithm  \eqref{iterativedistributedalgorithm.eqn1}
			and \eqref{iterativedistributedalgorithm.eqn2}
			with ${\bf G}=\gamma {\bf I}$, in which
			the requirement \eqref{sigmaHG.eq} is met  as the spectrum of ${\bf I}-{\bf HG}$ is contained in $[1-\gamma \alpha_2, 1-\gamma\alpha_1]\subset (-1, 1)$ whenever $0<\gamma<2/\alpha_2$.
%			It is observed in  \cite{Shi15} that the gradient descent algorithm
%			with step length $\gamma=2/(\alpha_1+\alpha_2)$
%			has fastest convergence,
%			where  respectively,
%where $\gamma_1=\min_{\pmb \lambda_i\in \Lambda} h(\pmb \lambda_i)>0$ and
%			$\alpha_2=\max_{\pmb \lambda_i\in \Lambda} h(\pmb \lambda_i)>0$,
%cf. Remark \ref{iopa.re}.
	}\end{remark}

	\begin{remark}\label{remark.arma}{\rm
			Let ${\bf S}$ be a graph shift and
 $h$ be a polynomial of order $L$ with its  distinct nonzero roots  $1/b_l$ %, 1\le k\le L$,
 satisfying
			\begin{equation}\label{arma.eq1}
|b_l| \|{\bf S}\|_2<1, \ 1\le l\le L.
			\end{equation}
%outside the disk
%			$B(0, \|{\bf S}\|_2)=\{z\in \C, |z|\le  \|{\bf S}\|_2\}$.
%Then  $h(t)= h(0) \prod_{l=1}^L (1-b_k t)$
%%for some distinct complex numbers $b_k, 1\le k\le L$
%%			Write
%				%\vspace{-.4em}\begin{equation*}
%and
Applying partial fraction decomposition to the rational function $1/h(t)$ gives
 $$(h(t))^{-1}= \sum_{k=1}^L {a_k}(1-b_k t)^{-1}$$ %	\vspace{-.4em}\end{equation*}
			for some coefficients $a_k, 1\le k\le L$. % satisfying
Then for the polynomial filter
			${\bf H}=h({\bf S})$,  we can decompose the inverse filter ${\bf H}^{-1}$
			into a family of elementary inverse filters  $({\bf I}- b_k {\bf S})^{-1}$,
			\begin{equation*}
			{\bf H}^{-1} =\sum_{k=1}^L  a_k ({\bf I}-b_k {\bf S})^{-1}.
			\end{equation*}
			Due to the above decomposition,
			the  inverse filtering procedure \eqref{inverseprocedure} can be implemented as
			follows,
			\begin{equation} \label{ARMA00}
			{\bf x}=\sum_{k=1}^L a_k ({\bf I}-b_k {\bf S})^{-1} {\bf b}=:\sum_{k=1}^L  a_k {\bf x}_k. 	\end{equation}
			The autoregressive moving
average (ARMA) method
has widely and popularly been known in the time series model  \cite{Leus17}.
The ARMA can also be applied for the inverse filtering procedure  \eqref{inverseprocedure}, where it  uses the  decomposition \eqref{ARMA00}
 with the  elementary inverse procedure
			${\bf x}_k=  ({\bf I}-b_k {\bf S})^{-1} {\bf b}$
			implemented by the following iterative approach,
\begin{equation*}\label{ARMA1}
			{\bf x}_{k}^{(m)}=b_k {\bf S}{\bf x}_{k}^{(m-1)}+{\bf b}, \ m\ge 1
\end{equation*}
			with initial ${\bf x}_{k}^{(0)}={\bf 0}$.
			We remark that the above approach is the same as
			the iterative  approximation algorithm
			\eqref{iterativedistributedalgorithm.eqn1}
			and \eqref{iterativedistributedalgorithm.eqn2}
			with ${\bf H}$ and ${\bf G}$ replaced by  ${\bf I}- b_k {\bf S}$  and  ${\bf I}$ respectively.
%			\vspace{-.5em}\begin{equation*} \label{**iterativedistributedalgorithm.eqn1}
%			{\bf x}_k^{(m)}={\bf x}_k^{(m-1)}+   {\bf b}_k^{(m-1)} \ {\rm and} \
%			{\bf b}_k^{(m)}= b_k {\bf S} {\bf b}_k^{(m-1)},
%			\ m\ge 1.
%			\vspace{-.5em}\end{equation*}
			Moreover, in the above selection of the graph filters $\bf H$ and $\bf G$,   the requirement \eqref{sigmaHG.eq}
			is met as it follows from \eqref {arma.eq1} that
\begin{equation}\label{arma.convergence}
\rho({\bf I}-{\bf H}{\bf G})\le \|{\bf I}- {\bf H}{\bf G}\|_2\le |b_k| \|{\bf S}\|_2<1\end{equation}
 for all $1\le k\le L$. Applying \eqref{arma.convergence},
% and
% following the argument to prove
% Theorem  \ref {convergence_ICPA.thm} in Appendix \ref{convergence_ICPA.appendix},
 we see that the  convergence rate to apply ARMA
 in the implementation of  the inverse filtering procedure is $(\max_{1\le k\le L} |b_k|) \rho({\bf S})<1$.
 %\le (\max_{1\le k\le L} |b_k|) \|{\bf S}\|_2<1$.
		}
	\end{remark}
	
	% By \eqref{iterativedistributedalgorithm.eqn1} and
	%\eqref{iterativedistributedalgorithm.eqn2}, we can prove by induction on $m$ that
	%\begin{equation}
	%{\bf b}^{(m)}=({\bf I}-{\bf H}{\bf G})^{m}{\bf b},
	%\end{equation}
	%  In , Isufi and his collaborators
	% to implement ${\bf x}=\varphi({\bf I}-\psi {\bf S} )^{-1}{\bf b}$, where polynomial is $h(\pmb \lambda)=(1-\psi \pmb \lambda)/\varphi$ and ${\bf x}_t$ converges to ${\bf x}$ when $\|\psi {\bf S}\|<1$. For polynomial filters with order $K$,  the reciprocal of polynomial $h(\pmb \lambda)$ can be decomposed as
	% $$\frac{1}{h(\pmb \lambda)}=\sum_{k=1}^K\frac{\varphi_k}{1-\psi_k\pmb \lambda},$$
	% for some complex coefficients $\varphi_k$ and $\psi_k$. For this case apply algorithm (\ref{ARMA1}) for each $g_k(\pmb \lambda)=\frac{\varphi_k}{1-\psi_k\pmb \lambda}$ to implement
	% $${\bf x}_k=\varphi_k({\bf I}-\psi_k {\bf S} )^{-1}{\bf b},$$
	% then
	%${\bf x}=\sum_{k=1}^K{\bf x}_k$  and it converges when $\|\psi_k {\bf S}\|<1$ for $1\le k\le K$. They named this method as ARMA$_K$.

	\section{Iterative  polynomial approximation  algorithms for inverse filtering}
\label{ipaa.section}

Let ${\bf S}_k=(S_k(i,j))_{i,j\in V}, 1\le k\le d$, be commutative graph shifts on a connected, undirected and unweighted graph ${\mathcal G}=(V, E)$ of order $N$,
$\Lambda$ be their joint spectrum \eqref{jointspectrum.def}, % of the shifts  ${\bf S}_1, \ldots, {\bf S}_d$,
and
${\bf H}=h({\bf S}_1, \ldots, {\bf S}_d)$ be  an invertible polynomial filter in \eqref{MultiShiftPolynomial}. % in Appendix \ref{commutative.section}.
For polynomial graph filters of a single   shift,
  there are several methods to implement the inverse filtering in a distributed network
 \cite{ Emirov19, mario19,  Leus17,  segarra17,Shi15,shuman18} approximately, see Remark \ref{Chebyshev.remark}.
 In this section, we propose two iterative algorithms to implement
  the inverse filtering associated with a polynomial graph filter of  commutative graph  shifts
  in a distributed network with limited data processing and communication requirement for its agents and also
  in a centralized facility with linear complexity.
 For the case that the joint spectrum  $\Lambda$ is fully known, we construct the polynomial interpolation  approximation  ${\bf G}_I$ and
  optimal polynomial approximations ${\widetilde {\bf G}}_L, L\ge 0$,
  to approximate the inverse filter ${\bf H}^{-1}$ in Subsection \ref{iopa.subsection},
and propose the iterative optimal polynomial approximation  % (IOPA)
  algorithm  \eqref{optimaliterativedistributedalgorithm.eqn1}
     to  implement the inverse filtering procedure   ${\bf b}\longmapsto {\bf H}^{-1} {\bf b}$,
    see    Theorem \ref{IOPAconvergence.thm}.
%  	The construction of the interpolating polynomial  $g_I$ in  \eqref{interpolating.gfunction}
%	and the optimal polynomial $g_L^*\in {\mathcal P}_L$ of degree $L$ in \eqref{minimalg*.def}
%	depends on the prior information of the  joint spectrum  $\Lambda$  of commutative graph shifts ${\bf S}_k, 1\le k\le d$.
	For  a graph ${\mathcal G}$ of large order, it is often computationally expensive to
	find the joint spectrum  $\Lambda$ %in \eqref{jointspectrum.def}
exactly.
	However,  the graph shifts ${\bf S}_k, 1\leq k\leq d$, in some engineering applications are  symmetric and their spectrum sets are
	known being contained in some intervals  \cite{chung1997, narang12, sakiyama19, narang13}. For instance,
	the normalized Laplacian matrix  %${\bf L}^{\rm sym}_{\cal G}$
	on a simple graph is symmetric and its spectrum is contained in $[0, 2]$.
	%, or spectrum of an adjacency matrix of an undirected circulant graphs is in $[-2,2]$.
	In Subsection \ref{icpa.subsection}, we consider the implementation of
 the inverse filtering procedure   ${\bf b}\longmapsto {\bf H}^{-1} {\bf b}$ when
 the joint  spectrum $\Lambda$ of
 commutative shifts  ${\bf S}_1,...,{\bf S}_d$  is contained in a cube.
Based on multivariate Chebyshev polynomial approximation to the function $h^{-1}$,
we introduce  Chebyshev polynomial filters   ${\bf G}_K, K\ge 0$,
to approximate the inverse filter ${\bf H}^{-1}$,
and propose
 the
   iterative Chebyshev polynomial approximation  % (ICPA)
   algorithm \eqref{Chebysheviterativedistributedalgorithm.eqn1}
   to  implement the inverse filtering procedure   ${\bf b}\longmapsto {\bf H}^{-1} {\bf b}$,
       see   Theorem \ref{icpa.thm}.
     In addition to the exponential convergence, the proposed  iterative optimal polynomial approximation algorithm
     and
   Chebyshev polynomial approximation algorithm
can be implemented at vertex level in a distributed network, see
Algorithms \ref{IOPA.algorithm}  and
  \ref{ICPA.algorithm}.

	\vspace{-0.6em}
	\subsection{Polynomial interpolation and optimal polynomial approximation}
	\label{iopa.subsection}

Let ${\bf U}$ be the unitary matrix in \eqref{upperdiagonalization} and denote its conjugate transpose by ${\bf U}^{\rm H}$.
	For  polynomial filters ${\bf H}=h({\bf S}_1, \ldots, {\bf S}_d)$ and ${\bf G}=g({\bf S}_1, \ldots, {\bf S}_d)$, %${\bf H}{\bf G}$ can be
	%upper-triangularized by ${\bf U}$  simultaneously in the sense that
	one may verify that $ {\bf U}({\bf I}-{\bf H}{\bf G}){\bf U}^{\rm H}$ is an  upper triangular matrix with
	diagonal entries
	$1-h({\pmb \lambda}_i)g({\pmb \lambda}_i),\ {\pmb \lambda}_i\in \Lambda$.
	Consequently, the  requirement \eqref{sigmaHG.eq} for the polynomial  graph filter ${\bf G}$ %=g({\bf S}_1,...,{\bf S}_d)$
becomes
	\begin{equation}\label{epsilon.def}
	\rho({\bf I}-{\bf G}{\bf H})=\sup_{{\pmb \lambda_i} \in \Lambda} \big|1-h({\pmb \lambda_i}) g({\pmb \lambda_i}) \big |<1.
\end{equation}
	A necessary condition for the existence of a  multivariate polynomial $g$ such that \eqref{epsilon.def} holds is that
\begin{equation}\label{h.condition}
	h(\pmb \lambda_i)\ne 0 {\rm \ for \  all \ }\  \pmb \lambda_i\in \Lambda,
\end{equation}
	or equivalently the  filter ${\bf H}$ is invertible. %=h({\bf S}_1,...,{\bf S}_d)$ is  represented by a nonsingular matrix.
	Conversely  if \eqref{h.condition} holds, $(\pmb \lambda_i, 1/h(\pmb \lambda_i)), 1\le i\le N$, can be interpolated by a polynomial $g_I$ %of $d$ variables
	of  degree  at most $N-1$ \cite{cheney2000course}, i.e.,
\begin{equation}\label{interpolating.gfunction}
	g_I(\pmb \lambda_i)= 1/h(\pmb \lambda_i), \ \pmb \lambda_i\in \Lambda.
\end{equation}
%where $N$ is the order of the graph ${\mathcal G}$.
%	For  the above polynomial $g_I$, we have that
%\vspace{-.4em}\begin{equation*}\epsilon=\max_{{\pmb \lambda}_i\in \Lambda} |1- h({\pmb \lambda}_i) g_I({\pmb \lambda}_i)|=0. \vspace{-.4em}\end{equation*}
	Take ${\bf G}_I= g_I({\bf S}_1, \ldots, {\bf S}_d)$.
	Then all eigenvalues of ${\bf I}-{\bf G}_I {\bf  H}$ are zero and
${\bf I}-{\bf G}_I {\bf  H}$ is similar to a strictly upper triangular matrix. Therefore $\rho({\bf I}-{\bf G}_I {\bf H})=0$
	and  the iterative approximation algorithm \eqref{iterativedistributedalgorithm.eqn1}
	and \eqref{iterativedistributedalgorithm.eqn2} converges in at most $N$ steps. %\qs{answer Reviewer 1, comment 12 ``is Remark 4.1 similar to a strictly upper triangular matrix?'' }
	
\begin{remark} {\rm We remark that the polynomial filter  ${\bf G}_I$ constructed above
is the inverse filter ${\bf H}^{-1}$ when all elements  $\pmb \lambda_i, 1\le i\le N$, in the joint spectrum  $\Lambda$
 in \eqref{jointspectrum.def} of graph shifts ${\bf S}_1, \ldots, {\bf S}_k$ are distinct.
	The above conclusion can be proved by following the argument used in the proof of Theorem  \ref{polynomialfilter.thm} in Appendix \ref{polynomialfilters.appendix}
	and the observation that the matrix
	${\bf I}-{\bf G}_I {\bf H}$ has all eigenvalues being zero and it commutes  with ${\bf S}_k, 1\le k\le d$.
However  in general,
 the above conclusion does not hold  without the distinct assumption on the joint spectrum $\Lambda$. For instance,
one may verify that for the polynomial filter ${\bf H}={\bf I}+{\bf A}$
 on the directed line graph of order $N$,
the identity matrix ${\bf I}$ can be chosen to be the polynomial filter ${\bf G}_I$ and it is not the same as the inverse filter ${\bf H}$, where the graph shift ${\bf A}$ is   the adjacent matrix associated with the directed line graph 
 and has all eigenvalues being zero.
% we do not know whether
%the polynomial filter ${\bf G}_I$ is the inverse filter ${\bf H}^{-1}$ .
%or not , or equivalently
% ${\bf I}-{\bf G}_I{\bf  H}$ is not necessarily the zero matrix, even
% all its  eigenvalues  are zero.
% However,
  }
	\end{remark}

	%by using Lagrange interpolation method we can find a polynomial $g$ of degree at most $d(|V|-1)$,
	%
	%\begin{equation}\label{LagrangeInterpolationG}
	%g({\bf x})=\sum_{i\in V}\frac{1}{h({\bf \lambda}_i)}\prod_{1\leq k\leq d}l_i^k({\bf x}),
	%\end{equation}
	%where
	%$$l_i^k({\bf x})=\prod_{j\neq i}\frac{x_k-\lambda_j^k}{\lambda_i^k-\lambda_j^k},\ i\in V,$$
	% ${\bf \lambda}_i=(\lambda_i^1,...,\lambda_i^d)\in \Lambda$.
	
%\medskip
		\begin{figure}[t]  %[h]
		\begin{center}
			\includegraphics[width=68mm, height=48mm]{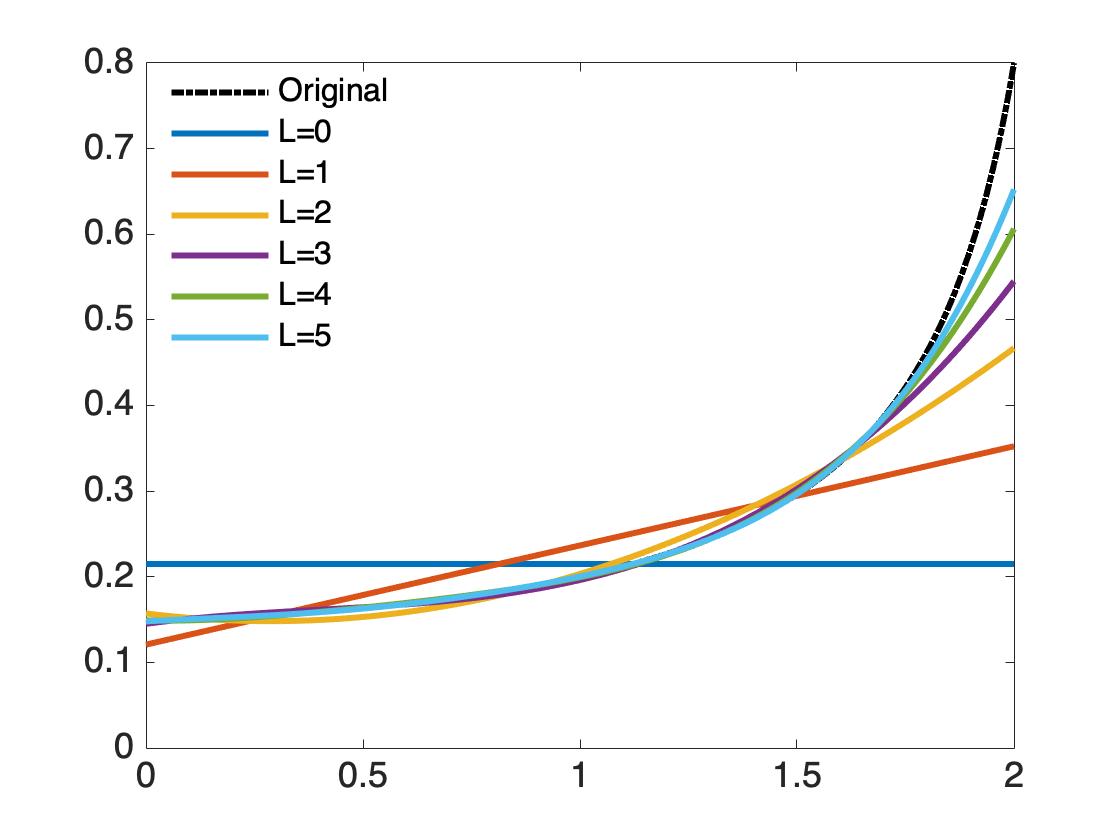}
			\includegraphics[width=68mm, height=48mm]{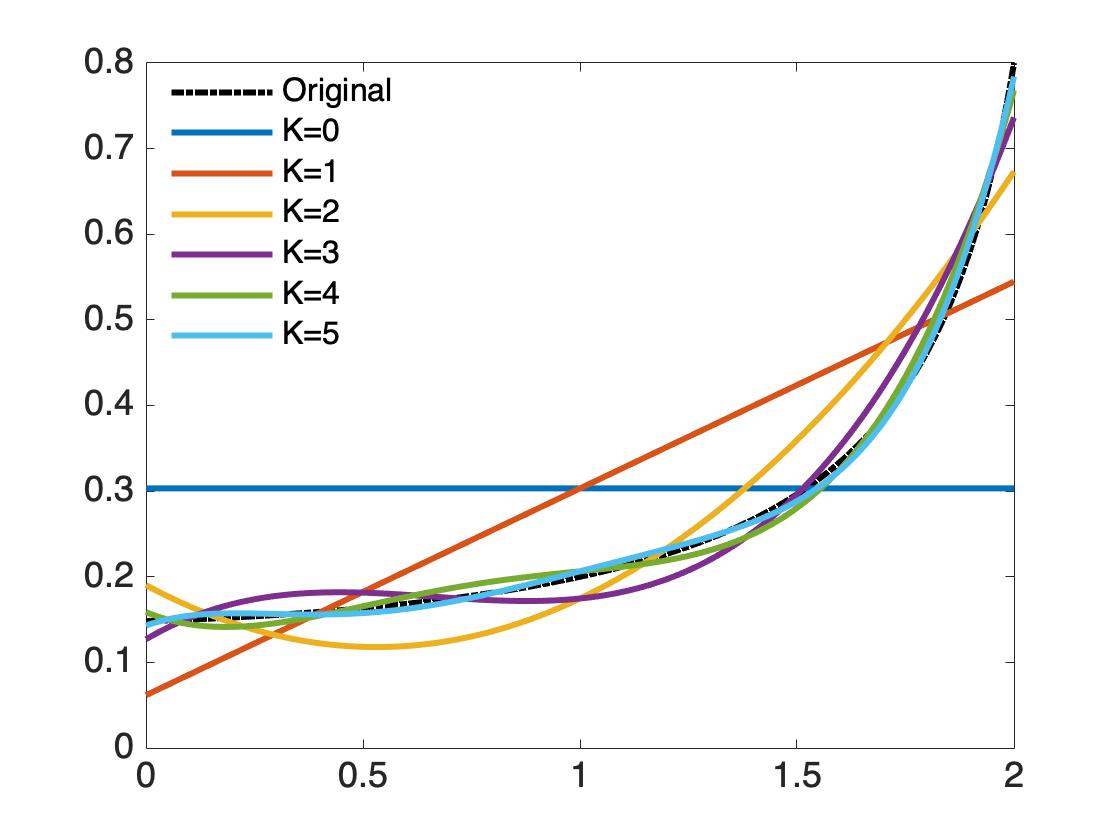}
			%\hskip0.2cm
			\caption{Plotted on the left are the original function $1/h_1$ on $[0, 2]$ (marked as ``Original") and its  optimal polynomial approximations ${\tilde g}_L, 0\le L\le 5$ (marked with different $L$),				while on the right are the original function $1/h_1$ on $[0, 2]$ and its  Chebyshev polynomial approximations  $g_K, 0\le K\le 5$ (marked with different $K$), where $ h_1(t)=(9/4-t)(3+t)$ is the polynomial in \eqref{h1filter.def},
the underlying graph  is the circulant graph ${\mathcal C}(1000, Q_0)$ in \eqref{circulant.edgedef} generated by $Q_0=\{1, 2, 5\}$ and the graph shift is the symmetric normalized Laplacian matrix on the circulant graph.
The approximation errors  $a_L$ in \eqref{optimal.condition} to measure the approximation property between
$\tilde g_L, 0\le L\le 5$ and $1/h_1$ are $0.4502$,
$0.1852$, $0.0612$, $0.0212$, $0.0072$, $0.0025$
%$0.4501$, $0.1850$, $0.0608$, $0.0210$, $0.0060$, $0.0023$
respectively, while
 approximation errors  $b_K$ in \eqref{chebyshevapproximation.con} to measure the approximation property between
$g_K, 0\le K\le 5$ and $1/h_1$
are $1.0463$,  $0.5837$, $0.2924$, $0.1467$, $0.0728$,
$0.0367$  respectively.
%$1.0463$, $0.5837$, $0.2924$, $0.1467$, $0.0728$, $0.0367$ respectively.
				This confirms  the observation numerically that optimal/Chebyshev polynomials with higher degrees  provide
				better  approximations to the function $1/h_1$ either on the spectrum of normalized Laplacian on the circulant graph  or on the interval $[0, 2]$ containing the spectrum.	}
			\label{approximation.fig}
		\end{center}
	\end{figure}

	For $L\ge 0$, denote the set of all polynomials of degree at most $L$  by ${\mathcal P}_L$.
	In  practice, we may not use
	the interpolation polynomial  $g_I$ in \eqref{interpolating.gfunction}, and hence the polynomial filter ${\bf G}=g_I({\bf S}_1, \ldots, {\bf S}_d)$
	in the iterative approximation algorithm \eqref{iterativedistributedalgorithm.eqn1}
	and \eqref{iterativedistributedalgorithm.eqn2},
	as  it is  of high degree in general.
	By  \eqref{desiredestimate}, the  convergence rate of the iterative approximation algorithm \eqref{iterativedistributedalgorithm.eqn1}
	and \eqref{iterativedistributedalgorithm.eqn2} depends on
	the spectral radius in \eqref{epsilon.def}.
Due to the above observation, we
select	${\tilde g}_L\in {\mathcal P}_L$ such that
\begin{equation}\label{minimalg*.def}
	{\tilde g}_L= \arg\!\min_{g\in {\mathcal P}_L} \sup_{\pmb \lambda_i\in \Lambda} |1- g(\pmb \lambda_i) h(\pmb \lambda_i)|,
\end{equation}
see Figure \ref{approximation.fig} for the approximation property of ${\tilde g}_L, L\ge 0$ to the reciprocal $1/h_1$ of the polynomial $ h_1(t)=(9/4-t)(3+t)$
in \eqref{h1filter.def}.
	For a multivariate polynomial $g\in {\mathcal P}_L$, we write
	$$ g ({\bf t})= \sum_{|{\bf k}|\le L} c_{\bf k}  {\bf t}^{\bf k},$$
where $|{\bf k}|=k_1+\cdots+k_d$ and
 ${\bf t}^{\bf k}=t_1^{k_1}\cdots t_d^{k_d} $
for  ${\bf t}=(t_1, \ldots, t_d)$  and ${\bf k}=(k_1, \ldots, k_d)$.
 Set ${\bf c}=(c_{\bf k})_{|{\bf k}|\le L}$.
 	Then for the case that all eigenvalues of ${\bf S}_k, 1\le k\le d$, are real, i.e.,  $\Lambda\subset \R^d$, the minimization problem \eqref{minimalg*.def} can be reformulated as a linear programming,
\begin{equation}
	%s_L^*=\arg
	\min \ s \ \ {\rm subject \ to} \  -(s-1) {\bf 1} \le {\bf P} {\bf c}\le (s+1){\bf 1},
	\end{equation}
	where  ${\bf P}=( h({\pmb \lambda}_i) {\pmb \lambda}_i^{\bf k})_{1\le i\le N, |{\bf k}|\le L}$,
${\bf 1}$ is the vector with all entries taking value 1, and we use standard componentwise ordering for real vectors.
% and for two vectors ${\bf u} and  ${\bf v}$  we say that
%${\bf u}\le {\bf v}$ if

Taking
polynomial filters
	\begin{equation}
\label{GLstar.def}
{\widetilde {\bf G}}_L={\tilde g}_L({\bf S}_1, \ldots, {\bf S}_d), \ L\ge 0,\end{equation}
to approximate the inverse filter ${\bf H}^{-1}$,  the iterative approximation algorithm \eqref{iterativedistributedalgorithm.eqn1}
	and \eqref{iterativedistributedalgorithm.eqn2}
	with the graph filter ${\bf G}$ replaced by ${\widetilde {\bf G}}_L$
becomes
	\begin{equation} \label{optimaliterativedistributedalgorithm.eqn1}
	\left\{ \begin{array}{l}
	{\bf z}^{(m)}= {\widetilde {\bf G}}_L {\bf e}^{(m-1)}, \\
	{\bf e}^{(m)}={\bf e}^{(m-1)}-{\bf H} {\bf z}^{(m)}, \\
	{\bf x}^{(m)}={\bf x}^{(m-1)}+ {\bf z}^{(m)}, \ m\ge 1,
	\end{array} \right.
	\end{equation}
	with initials  ${\bf e}^{(0)}$ and ${\bf x}^{(0)}$ given  in \eqref{iterativedistributedalgorithm.eqn2}.
We call the above iterative algorithm \eqref{optimaliterativedistributedalgorithm.eqn1}
  by the {\em iterative optimal polynomial approximation algorithm}, or IOPA in abbreviation.

     \begin{algorithm}[t]
\caption{The IOPA algorithm  to implement the inverse filtering procedure ${\bf b}\longmapsto {\bf H}^{-1}{\bf b}$
%for a polynomial filter
  at a vertex $i\in V$. }
\label{IOPA.algorithm}
\begin{algorithmic}  %[1]

\STATE {\bf Inputs}: Polynomial coefficients of ${\bf H}$ and ${\widetilde {\bf G}}_L$,  entries $S_k(i,j), j\in {\mathcal N}_i$ in the $i$-th row of the shift ${\bf S}_k, 1\le k\le d$,
the  value $b(i)$  of the input signal ${\bf b}=(b(i))_{i\in V}$ at the vertex $i$, and number $M$ of iteration.

%\STATE {\bf Operation}: Evaluate $m_k=\mu(B(k, r))$, compute ${\bf F}_k= {\bf H}_{0,k}^T{\bf H}_{0,k}+ {\bf H}_{1,k}^T{\bf H}_{1,k}$,
% find its inverse  $({\bf F}_k)^{-1}$, and then compute $ {\bf G}^L_{l; k}:=({\bf F}_k)^{-1} {\bf H}_{l,k}^T, l=0, 1$.

%  , and a local approximation
%${\bf G}_k=(\tilde f_k(i,j))_{i,j\in B(k, 2r)}$  to the matrix ${\bf H}$
% and compute ${\bf F}_k= {\bf H}_{0,k}^T{\bf H}_{0,k}+ {\bf H}_{1,k}^T{\bf H}_{1,k}$ and
%$({\bf F}_k)^{-1}=(\tilde f_k(i,j))_{i,j\in B(k, 2r)}$

\STATE {\bf Initialization}:  Initial $e^{(0)}(i)=b(i)$, $x^{(0)}(i)=0$ and $n=0$.

\STATE{\bf Iteration}: \\
\hspace{0.3cm}
{\bf For} $m=1, 2, \ldots, M$

\hspace{0.5cm}  {\bf Step 1}: \ Use Algorithm \ref{singleshiftprocedure.algorithm} for $d=1$ and Algorithm \ref{distributed_Hx.algorithm}
for $d\ge 2$ to implement the filtering procedure ${\bf e}^{(m-1)}\longmapsto {\bf z}^{(m)}={\widetilde {\bf G}}_L {\bf e}^{(m-1)}$ at the vertex $i$, and the output is the $i$-th entry  $z^{(m)}(i)$ of the vector ${\bf z}^{(m)}$.

\hspace{0.5cm} {\bf Step 2}: \ Use Algorithm \ref{singleshiftprocedure.algorithm} for $d=1$ and Algorithm \ref{distributed_Hx.algorithm} for $d\ge 2$
to implement the filtering procedure ${\bf z}^{(m)}\longmapsto {\bf w}^{(m)}={\bf H} {\bf z}^{(m)}$ at the vertex $i$, and the output is the $i$-th entries $w^{(m)}(i)$ of the vector ${\bf w}^{(m)}$.

\hspace{0.5cm} {\bf Step 3}:  Update $i$-th entries of ${\bf e}^{(m)}$ and ${\bf x}^{(m)}$ by
$e^{(m)}(i)=e^{(m-1)}(i)-{w}^{(m)}(i)$ and
$x^{(m)}(i)=x^{(m-1)}(i)+ z^{(m)}(i)$ respectively.

\hspace{0.3cm} {\bf end}

\STATE {\bf Output}: The approximated  value $ x(i)\approx x^{(M)}(i)$  is  the output signal ${\bf H}^{-1}{\bf b}=( x(i))_{i\in V}$ at the vertex $i$.
\end{algorithmic}  %\vspace{-.03in}
\end{algorithm}

  Presented in Algorithm
  \ref{IOPA.algorithm} is the implementation of IOPA algorithm at the vertex level in a distributed network.
 In each iteration of  Algorithm
  \ref{IOPA.algorithm},  each vertex/agent  of the distributed network
 needs about $O((L+1)^{d-1}+ \prod_{k=1}^{d-1} (L_k+1))$ steps
containing data exchanging among adjacent vertices and weighted sum  %inear combination
of values at adjacent vertices in each iteration.
The memory requirement for each vertex is about
$O\big( (\deg {\mathcal G}+L_k+1) \prod_{k=1}^{d-1}(L_k+1)+ (\det {\mathcal G})+L+1)(L+1)^{d-1})\big)$.
The total   operations of addition and
multiplication in each iteration   to  implement the inverse filtering procedure   ${\bf b}\longmapsto {\bf H}^{-1} {\bf b}$ via Algorithm   \ref{IOPA.algorithm}
  in a distributed network  and   procedure \eqref{optimaliterativedistributedalgorithm.eqn1}
 in a central facility are almost
the same, which are both about $O\big(N (\deg {\mathcal G}+1) (\prod_{k=1}^{d}(L_k+1)+(L+1)^{d})\big)$.

 By \eqref{minimalg*.def}, we have
 	\begin{equation}\label{optimal.conditionold}
	\rho({\bf I}-{\widetilde {\bf G}}_L{\bf H})= \sup_{\pmb \lambda_i\in \Lambda} |1- {\tilde g}_L(\pmb \lambda_i) h(\pmb \lambda_i)|
\end{equation}
and
 $\rho({\bf I}-{\widetilde {\bf G}}_L{\bf H}), 0\le L\le N-1$,  is a  nonnegative decreasing sequence with
 the last term $\rho({\bf I}-{\widetilde {\bf G}}_{N-1}{\bf H})$ being the same as $\rho({\bf I}-{\bf G}_I{\bf H})= 0$ by \eqref{interpolating.gfunction},
 i.e.,
\begin{equation}\label{monotone.eq00}
0 =  \rho({\bf I}-{\bf G}_I{\bf H})= \rho({\bf I}-{\widetilde {\bf G}}_{N-1}{\bf H})\le
\rho({\bf I}-{\widetilde {\bf G}}_{L+1}{\bf H}) \le  \rho({\bf I}-{\widetilde {\bf G}}_L{\bf H})\le
\rho({\bf I}-{\widetilde {\bf G}}_0{\bf H}), \ 0\le L\le N-1.
\end{equation}
%\begin{eqnarray}\label{monotone.eq00} \hskip-0.10in
%0& \hskip-0.08in = &  \hskip-0.08in \rho({\bf I}-{\bf G}_I{\bf H})= \rho({\bf I}-{\widetilde {\bf G}}_{N-1}{\bf H})\le
%\rho({\bf I}-{\widetilde {\bf G}}_{L+1}{\bf H})\nonumber\\
% \hskip-0.10in & \hskip-0.08in \le & \hskip-0.08in \rho({\bf I}-{\widetilde {\bf G}}_L{\bf H})\le
%\rho({\bf I}-{\widetilde {\bf G}}_0{\bf H}), \ 0\le L\le N-1.
%\end{eqnarray}
%for $0\le L\le N-1$.
This implies that the polynomial filters ${\bf G}_L$ with larger $L$ provide better approximation to the inverse filter ${\bf H}^{-1}$
and hence the corresponding  IOPA algorithm  \eqref{optimaliterativedistributedalgorithm.eqn1} has faster convergence.
 In the following theorem, we show that
 the IOPA algorithm  \eqref{optimaliterativedistributedalgorithm.eqn1}
	converges exponentially when $L$ is appropriately chosen, see Subsection \ref{circulant.demon.subsection} for the numerical demonstration.

	\begin{theorem}
		\label{IOPAconvergence.thm}
{\rm
		%Let   ${\bf b}$ be a graph signal,
		 Let
		${\bf S}_1,...,{\bf S}_d$ be  commutative graph shifts, % satisfying \eqref{commutativityS}.
		${\bf H}=h({\bf S}_1, \ldots, {\bf S}_d)$ be an invertible polynomial graph filter  for some multivariate polynomial $h$,
		and
		let degree $L\ge 0$ be so chosen that
\begin{equation}\label{optimal.condition}
	a_L:= \sup_{\pmb \lambda_i\in \Lambda} |1- {\tilde g}_L(\pmb \lambda_i) h(\pmb \lambda_i)|<1.
\end{equation}
		Then for any graph signal $\bf b$,  the sequence
		${\bf x}^{(m)}, m\ge 1$, in the IOPA  algorithm  \eqref{optimaliterativedistributedalgorithm.eqn1}
		converges exponentially to
		% the   % optimal
		%solution
		${\bf H}^{-1} {\bf b}$.  Moreover, for any $r\in (a_L, 1)$, there exists a positive constant $C$  such that
\eqref{desiredestimate} holds.
%		\vspace{-.4em}\begin{equation*}\label{iopadesiredestimate}
%		\|{\bf x}^{(m)}- {\bf H}^{-1} {\bf b} \|_2 \leq   C \|{\bf x}\|_2
%		r^m, \ m\ge 1.
%		%\varepsilon \|\tilde{\bf x}-{\bf x}^{(m-1)}\|_2.
%			\vspace{-.2em}\end{equation*}
		%where $r\in (\epsilon, 1)$ and
	}
\end{theorem}

\begin{proof} The conclusion follows from \eqref{optimal.conditionold}, \eqref{optimal.condition} and Theorem \ref{convergence_ICPA.thm}
with ${\bf G}$ replaced by ${\widetilde {\mathbf G}}_L$. % and the observation that $a_L:=\rho({\bf I}-{\widetilde {\bf G}}_L{\bf H})$.
\end{proof}

%	 In the following  theorem, we show that
% the IOPA algorithm  \eqref{optimaliterativedistributedalgorithm.eqn1}
%	converges exponentially when $L$ is so chosen  that
%	\vspace{-.4em}\begin{equation}\label{optimal.condition}
%	a_L:=\rho({\bf I}-{\widetilde {\bf G}}_L{\bf H}):= \sup_{\pmb \lambda_i\in \Lambda} |1- {\tilde g}_L(\pmb \lambda_i) h(\pmb \lambda_i)|<1.\vspace{-.4em}\end{equation}

\smallskip

Let $L_0$  be the minimal nonnegative integer so that  $a_{L_0}<1$.
  By  \eqref{monotone.eq00} and  Theorem  \ref{IOPAconvergence.thm},
	the inverse filtering procedure \eqref{inverseprocedure}
	can be implemented by applying the IOPA  algorithm  \eqref{optimaliterativedistributedalgorithm.eqn1}
with   $L\ge L_0$ and
 the IOPA  algorithm  \eqref{optimaliterativedistributedalgorithm.eqn1}
 converges faster  when the higher degree  $L$ of the  optimal polynomial ${\tilde g}_L$ is selected, see Subsection
 \ref{circulant.demon.subsection} for the numerical demonstration.
 However, the implementation of  IOPA  algorithm  \eqref{optimaliterativedistributedalgorithm.eqn1} with larger $L$ at every agent/vertex in a distributed network
 has higher  computational cost in each iteration
and requires more memory for each agent/vertex, and also it takes higher computational cost to solve the
the  minimization problem
\eqref{minimalg*.def} for larger $L$.

	%The corresponding iterative polynomial approximation algorithm \eqref{iterativedistributedalgorithm.eqn1}
	%  and \eqref{iterativedistributedalgorithm.eqn2} converges to the solution in  a finite step.
	%  % by  \eqref{cm.eqn} and \eqref{epsilon.zero}.
	
We finish this subsection with a remark on the
IOPA algorithm \eqref{optimaliterativedistributedalgorithm.eqn1} and
			the gradient descent method
			\eqref{gradientdescent.al}.
	
	\begin{remark}\label{iopa.re}{\rm
 For  the case that the graph filter ${\bf H}$ has its spectrum contained in  $[\alpha_1, \alpha_2]$,
			the solution of the minimization problem  \eqref{minimalg*.def} with $L=0$
			is given by
$			\tilde g_0= 2/(\alpha_1+\alpha_2)$, %,\end{equation*}
			where $\alpha_1=\min_{\pmb \lambda_i\in \Lambda} h(\pmb \lambda_i)$ and
			$\alpha_2=\max_{\pmb \lambda_i\in \Lambda} h(\pmb \lambda_i)$ are the minimal and maximal eigenvalues of ${\bf H}$ respectively.
			Therefore,			to implement  the inverse filtering procedure \eqref{inverseprocedure},
			the gradient descent method
			\eqref{gradientdescent.al}
			with zero initial and optimal step length $\gamma=2/(\alpha_1+\alpha_2)$ is the {\bf same} as the
			proposed  IOPA algorithm \eqref{optimaliterativedistributedalgorithm.eqn1} with $L=0$, cf. Remark \ref{remark.grad}.
By \eqref{monotone.eq00}, we see that the  IOPA algorithm with $L\ge 1$ has faster convergence than the gradient descent method does,
 at the cost of heavier computational cost at each iteration,
 see Table \ref{ComparisonOneShift} and Figure \ref{TotalTime.fig} in Subsection \ref{circulant.demon.subsection}
%   see Table \ref{ComparisonOneShift} in Section
% \ref{circulant.demon.subsection}
 for numerical demonstrations.
}\end{remark}

	\subsection{Chebyshev polynomial approximation}
	\label{icpa.subsection}

	In this subsection, we assume that   commutative graph shifts ${\bf S}_1,...,{\bf S}_d$
have  their joint spectrum $\Lambda$ contained in the cubic
$[{\pmb\mu}, {\pmb \nu}]=[\mu_1, \nu_1]\times \cdots\times [\mu_d, \nu_d]$,
	\begin{equation}\label{jointspectral.interval}
	{\pmb \lambda}_i\in [{\pmb \mu}, {\pmb \nu}] \ {\rm for \ all} \ {\pmb \lambda}_i\in \Lambda,
	\end{equation}
	%Let ${\bf S}$  be a graph shift with its spectrum contained in $[a, b]$ and
%	Let  ${\bf S}_1,...,{\bf S}_d$ be symmetric commutative graph shifts %satisfying \eqref{commutativityS}
%	such that
	and  $h$ be a multivariate polynomial satisfying
	\begin{equation}\label{h.cond1}
	h({\bf t})\ne 0 \ \ {\rm for \ all}\ \ {\bf t}\in [{\pmb \mu}, {\pmb \nu}].
	\end{equation}
	Define  Chebyshev polynomials $T_k, k\ge 0$,  by
	\begin{equation*} \label{shifted_Cheby.eqn}
	T_k(s)=\left\{ \begin{array}{lc} 1 & {\rm if} \hskip0.1cm k=0,
\\ s & {\rm if} \hskip0.1cm k=1, \\
%	2s^2-1 & {\rm if} \hskip0.1cm k=2, \\
	2s {T}_{k-1}(s)-{T}_{k-2}(s) & {\rm if} \hskip0.1cm k \geq 2, \end{array} \right.
	\end{equation*}
	and %{\color{red} (Normalization should be with 1/2, $T_{0}(s)=\frac{1}{2}$)}
shifted multivariate Chebyshev polynomials $\bar T_{\bf k}, {\bf k}=(k_1, \ldots, k_d)\in \Z_+^d$, on $[{\pmb \mu}, {\pmb \nu}]$ by
	$$\Bar{T}_{{\bf k}}({\bf t})=\prod_{i=1}^{d} T_{k_i}\Big(\frac{2t_i-\mu_i-\nu_i}{\nu_i-\mu_i}\Big),
	\ \  {\bf t}=(t_1,...,t_d)\in [{\pmb \mu}, {\pmb \nu}].$$
	By \eqref{h.cond1},
	$1/h$ is an analytic function on $[{\pmb \mu}, {\pmb \nu}]$,
	and  hence it  has Fourier expansion % represented by  linear combinations
	in term of  shifted Chebyshev polynomials $\Bar{T}_{{\bf k}}, {\bf k}\in \Z_{+}^d$,
	\begin{equation*}\label{h.expansion}
	\frac{1}{h({\bf t})}=\sum_{{\bf k}\in \Z_+^d}c_{{\bf k}}\Bar{T}_{{\bf k}}({\bf t}), \ {\bf t}\in [{\pmb \mu}, {\pmb \nu}],
	\end{equation*}
	where
	\begin{equation*}\label{MultiVarPol-ChebCoeff}
	c_{\bf k} %=\int_{{\bf t}\in [{\pmb \mu}, {\pmb \nu}]} f({\bf t})\bar  T_{\bf k}({\bf t}) d{\bf t}, \ {\bf k}\in \Z_+^d.
	=\frac{2^{d-p({\bf k})}}{\pi^d} %\frac{2^{d}}{\pi^d}
	\int_{[0,\pi]^d}\frac{  \bar T_{\bf k}(t_1( {\pmb \theta}), \ldots, t_d({\pmb \theta}))}
	%\prod_{i=1}^{d} \cos(k_i\theta_i)}
	{h(t_1( {\pmb \theta}), \ldots, t_d({\pmb \theta}))}d \boldsymbol{\theta},\  {\bf k}\in \Z_+^d,
	\end{equation*}
%	{\color{red} (it should be $\frac{2^{d-l}}{\pi^d}$ instead of $\left(\frac{2}{\pi}\right)^d$, where $l$ is a number of nonzero entries in ${\bf k}$)}
$p({\bf k})$ is the number of zero components in ${\bf k}\in \Z^d_+$, and
	$t_i({\pmb \theta})= \frac{\nu_i+\mu_i}{2}+\frac{\nu_i-\mu_i}{2}\cos(\theta_i), 1\le i\le d,
	$
	for  $\boldsymbol{\theta}=(\theta_1,...,\theta_d)$.
	Define  partial sum of the  % Fourier
expansion \eqref{h.expansion} by
	\begin{equation}\label{MultiVariablePol-TrucatedChebEq}
	g_K({\bf t})=\sum_{|{\bf k}|\leq K}c_{{\bf k}}\Bar{T}_{{\bf k}}({\bf t}),
	\end{equation}
where $|{\bf k}|=\sum_{i=1}^d k_i$ for  ${\bf k}=(k_1,...,k_d)^T\in \mathbb{Z}_+^d$.
Due to the analytic property of the polynomial $h$,  the partial sum $g_K, K\ge 0$,  converges to  $1/h$  exponentially \cite{phillips03},
	\begin{equation}\label{exponentialdecay}
	b_K:=\sup_{{\bf t}\in [{\pmb \mu}, {\pmb \nu}]}|1-h({\bf t})g_K({\bf t})|\leq Cr_0^K, \ K\geq 0,
	\end{equation}
	for some  positive constants $C\in (0, \infty)$ and $r_0\in (0, 1)$, see
 Figure \ref{approximation.fig} for the approximation property of $g_K, K\ge 0$ to the reciprocal $1/h_1$ of the polynomial $ h_1(t)=(9/4-t)(3+t)$
in \eqref{h1filter.def}.

Set \begin{equation}\label{GK.def000}
{\bf G}_K=g_K({\bf S}_1, \ldots, {\bf S}_d),\ K\ge 0,\end{equation}
and call the iterative approximation algorithm \eqref{iterativedistributedalgorithm.eqn1}
	and \eqref{iterativedistributedalgorithm.eqn2}
	with the graph filter ${\bf G}$ replaced by ${\bf G}_K$ by the {\em iterative Chebyshev polynomial approximation algorithm}, or ICPA in abbreviation,
	\begin{equation} \label{Chebysheviterativedistributedalgorithm.eqn1}
	\left\{ \begin{array}{l}
	{\bf z}^{(m)}= {\bf G}_K {\bf e}^{(m-1)},  \\
	{\bf e}^{(m)}={\bf e}^{(m-1)}-{\bf H} {\bf z}^{(m)}, \\
	{\bf x}^{(m)}={\bf x}^{(m-1)}+ {\bf z}^{(m)}, \ m\ge 1,
	\end{array} \right.
	\end{equation}
	with initials  ${\bf  e}^{(0)}$ and ${\bf x}^{(0)}$ given in \eqref{iterativedistributedalgorithm.eqn2}.
In the following theorem, we show that the  ICPA algorithm \eqref{Chebysheviterativedistributedalgorithm.eqn1} converges exponentially,
	 when the degree $K$ is  so chosen that \eqref{chebyshevapproximation.con} holds, see Subsection  \ref{circulant.demon.subsection} for the demonstration. %, see  Appendix \ref{icpa.thm.appendix} for the proof.

	\begin{theorem}\label{icpa.thm}
	{\rm 	Let  ${\bf S}_1,...,{\bf S}_d$ be  commutative graph shifts, % satisfying \eqref{commutativityS}.
		${\bf H}$ be a polynomial graph filter  of the graph shifts,
		and
		let degree $K\ge 0$ of Chebyshev polynomial approximation be so chosen that
\begin{equation}\label{chebyshevapproximation.con}
b_K:=\sup_{{\bf t}\in [{\pmb \mu}, {\pmb \nu}]}|1-h({\bf t})g_K({\bf t})| <1.
\end{equation}
	Then for any graph signal $\bf b$,
		${\bf x}^{(m)}, m\ge 0$, in the  ICPA  algorithm  \eqref{Chebysheviterativedistributedalgorithm.eqn1}
		converges exponentially to
		% the   % optimal
		%solution
		${\bf H}^{-1} {\bf b}$. Moreover for any $r\in (b_K, 1)$, there exists a positive constant $C$  such that
%\eqref{desiredestimate} holds,
\begin{equation}\label{iopadesiredestimate}
		\|{\bf x}^{(m)}- {\bf H}^{-1} {\bf b} \|_2 \leq   C \|{\bf x}\|_2
		r^m, \ m\ge 1.
		%\varepsilon \|\tilde{\bf x}-{\bf x}^{(m-1)}\|_2.
\end{equation}
}	\end{theorem}
	
\begin{proof}
Following the argument used in \eqref{optimal.conditionold}, one may verify  that
\begin{equation}\label{icpa.thm.pf.eq1}\rho({\bf I}-{\bf G}_K {\bf H})=\sup_{\pmb \lambda_i\in \Lambda} |1- g_K(\pmb \lambda_i) h(\pmb \lambda_i)|
\le b_K,
\end{equation}
where the inequality holds by \eqref{jointspectral.interval} and
the definition \eqref{exponentialdecay} of $b_K, K\ge 0$.
Then the desired conclusion follows from
\eqref{icpa.thm.pf.eq1} and Theorem \ref{convergence_ICPA.thm}
with ${\bf G}$ replaced by ${\mathbf G}_K$. % and the observation that $\mu_L:=\rho({\bf I}-{\widetilde {\bf G}}_L{\bf H})$.
\end{proof}

\begin{remark}\label{icpaconvergencerate.remark}  {\rm
We remark that the convergence conclusion \eqref{iopadesiredestimate} in Theorem \ref{icpa.thm}  can be
improved as
\begin{equation}\label{iopadesiredestimate.rem}
\|{\bf x}^{(m)}-{\bf H}^{-1}{\bf b}\|_2\le \frac{\|{\bf H}\|_2\|{\bf G}_K\|_2}{1-b_K} (b_K)^m \|{\bf H}^{-1}{\bf b}\|_2, \ m\ge 1.
\end{equation}
provided that  the commutative graph shifts  ${\bf S}_1, \ldots, {\bf S}_d$ are symmetric.
 Under the assumption that ${\bf S}_1, \ldots, {\bf S}_d$ are symmetric,
  there exists a unitary matrix $\bf U$ such that
          they can be diagonalized simultaneously and hence
	$ {\bf U}^T({\bf I}-{\bf H}{\bf G}_K) {\bf U}$ is a diagonal matrix with diagonal entries $1- h({\pmb \lambda}_i) g_K({\pmb \lambda}_i), 1\le i\le N$, where ${\pmb \lambda_1}, \ldots, {\pmb \lambda}_N\in \Lambda$. Therefore
\begin{equation}\label{v20.eq0}
\|{\bf I}-{\bf G}_K {\bf H}\|_2 = \rho({\bf I}-{\bf G}_K {\bf H})
=\sup_{1\le i\le N} |1- h({\pmb \lambda}_i) g_K({\pmb \lambda}_i)|\le b_K
\end{equation}
	where the last inequality follows from \eqref{jointspectral.interval} and
\eqref	{exponentialdecay}.
The desired exponential convergence can be obtained by applying
the similar argument used in
Theorem \ref{convergence_ICPA.thm} with \eqref{sigmaHG.eq} and 	\eqref{sigmaHG.eq2} replaced by
\eqref{chebyshevapproximation.con} and \eqref{v20.eq0}.
}\end{remark}

\begin{remark} \label{ICPAalgorithm.remark}
{\em   We remark that each iteration in the ICPA algorithm \eqref{Chebysheviterativedistributedalgorithm.eqn1}
can be implemented at vertex level, % at the vertex level,
see Algorithm \ref{ICPA.algorithm}.
%to the distributed  implementation at the vertex level. % in a distributed network.
     \begin{algorithm}[t]
\caption{The ICPA algorithm  to implement the inverse filtering procedure ${\bf b}\longmapsto {\bf H}^{-1}{\bf b}$
  at a vertex $i\in V$. }
\label{ICPA.algorithm}
\begin{algorithmic}  %[1]

\STATE {\bf Inputs}: Polynomial coefficients of polynomial filters ${\bf H}$ and ${\bf G}_K$,  entries $S_k(i,j), j\in {\mathcal N}_i$ in the $i$-th row of the shifts ${\bf S}_k, 1\le k\le d$,
the  value $b(i)$  of the input signal ${\bf b}=(b(i))_{i\in V}$ at the vertex $i$, and number $M$ of iteration.

%\STATE {\bf Operation}: Evaluate $m_k=\mu(B(k, r))$, compute ${\bf F}_k= {\bf H}_{0,k}^T{\bf H}_{0,k}+ {\bf H}_{1,k}^T{\bf H}_{1,k}$,
% find its inverse  $({\bf F}_k)^{-1}$, and then compute $ {\bf G}^L_{l; k}:=({\bf F}_k)^{-1} {\bf H}_{l,k}^T, l=0, 1$.

%  , and a local approximation
%${\bf G}_k=(\tilde f_k(i,j))_{i,j\in B(k, 2r)}$  to the matrix ${\bf H}$
% and compute ${\bf F}_k= {\bf H}_{0,k}^T{\bf H}_{0,k}+ {\bf H}_{1,k}^T{\bf H}_{1,k}$ and
%$({\bf F}_k)^{-1}=(\tilde f_k(i,j))_{i,j\in B(k, 2r)}$

\STATE {\bf Initialization}:  Initial $e^{(0)}(i)=b(i)$, $x^{(0)}(i)=0$ and $n=0$.

\STATE{\bf Iteration}:  Use the iteration in Algorithm \ref{IOPA.algorithm} except replacing  ${\widetilde {\bf G}}_L$  by
 ${\bf G}_K$ in
\eqref{GK.def000}, and the output is
$ x^{(M)}(i)$.

\STATE {\bf Output}: The approximated  value $ x(i)\approx x^{(M)}(i)$  is  the output signal ${\bf H}^{-1}{\bf b}=( x(i))_{i\in V}$ at the vertex $i$.
\end{algorithmic}  %\vspace{-.03in}
\end{algorithm}
In  each iteration of the ICPA algorithm \eqref{Chebysheviterativedistributedalgorithm.eqn1},  every agent in a distributed network (vertex of the graph)
needs about $O((K+1)^{d-1}+ \prod_{k=1}^{d-1} (L_k+1))$ steps with each step
containing data exchanging among adjacent vertices and  weighted linear combination
of values at adjacent vertices.
The memory requirement for each agent is about
$O\big( (\deg {\mathcal G}+L_d+1) \prod_{k=1}^{d-1}(L_k+1)+ (\deg {\mathcal G}+K+1)(K+1)^{d-1}\big)$.
The total  operations of addition and
multiplication to implement each iteration of Algorithm \ref{ICPA.algorithm}
in a distributed network  and
to implement \eqref{Chebysheviterativedistributedalgorithm.eqn1} in a central facility are almost
the same, which are both about $O\big(N (\deg {\mathcal G}+1) (\prod_{k=1}^{d}(L_k+1)+(K+1)^{d})\big)$.
}	
\end{remark}

\begin{remark}\label{Chebyshev.remark}
{\rm  By \eqref {exponentialdecay},   an inverse filtering procedure \eqref{inverseprocedure}
	can be approximately implemented by the filter procedure
	${\bf G}_K {\bf x}$ with large $K$, i.e.,
%	\begin{equation}
${\bf H}^{-1} {\bf x}\approx {\bf G}_K {\bf x}$ for large $K$.  %\end{equation}
	The above implementation of the inverse filtering has been discussed in \cite{Emirov19, shuman18} for the case that % where
	${\bf H}$ is a polynomial graph filter of {\bf one}  shift, %  ${\bf S}$.
 and it is known as
	the  Chebyshev polynomial approximation algorithm (CPA).
%{\color{blue} Mention our sampta paper for iterative Chebyshev polynomial approximation algorithm with one shift. }
We remark that in the single graph shift setting,  the approximation ${\bf G}_K {\bf x}$
	in the   CPA is the same as the {\bf first} term ${\bf x}^{(1)}$ in
the  ICPA  algorithm  \eqref{Chebysheviterativedistributedalgorithm.eqn1}.	
	To  implement the inverse filtering with high accuracy, the CPA requires Chebyshev polynomial approximation
	of {\bf high} degree, which  means more integrals involved in coefficient calculations.
	On the other hand, we can select Chebyshev polynomial approximation
	of lower degree in  the ICPA  algorithm  \eqref{Chebysheviterativedistributedalgorithm.eqn1}
to reach the same accuracy with few iterations. By Theorem \ref{icpa.thm},
the  ICPA  algorithm  \eqref{Chebysheviterativedistributedalgorithm.eqn1}
has exponential convergent rate $b_K$, which has limit zero as $K\to \infty$. This indicates that the ICPA algorithm  converges faster for large $K$, however for each agent in a distributed network, its data processing system need more memory  to store data  and time to  process data,
 and its communication system costs more for larger $K$ too.
 Our simulation in the next section confirms  the above observation, see Table \ref{ComparisonOneShift} and Figure \ref{TotalTime.fig} in Subsection \ref{circulant.demon.subsection}.}
%		 Hence
%	\vspace{-.3em}\begin{equation}\label{chebyshevapproximation.conff}
%	b_K<1
%	\vspace{-.3em}\end{equation}
%	for large $K$.
 %and for a set $E$ we use $\# E$ to denote its cardinality.
\end{remark}
	% {\color{red} Trade-off for higher order approximation and more iteration}
	
	\section{Numerical simulations}\label{Numeric.section}

%		In the first part of Section
%\ref{Numeric.section},
% we	demonstrate the implementation of the   proposed  iterative   algorithms for inverse filtering
%on a  circulant graph,
%	and compare their performances with the  gradient descent method with zero initial \cite{Shi15}
%	and the autoregressive moving average %(ARMA)
%algorithm   \cite{Leus17}.
%In the second and third parts of Section
%\ref{Numeric.section}, we apply the proposed iterative algorithms to denoise time-varying signals governed by some differential equations and
% a US hourly
%temperature data set % on  August 1st,  2010
%respectively.

%\begin{figure}[t]   %[h]
%		\begin{center}
%			%\includegraphics[width=38mm, height=38mm]{CirculantGraph}
%			%\includegraphics[width=38mm, height=38mm]{CirculantGraph}\\
%			%\includegraphics[width=38mm, height=38mm]{CirculantGraph}
%			%\includegraphics[width=38mm, height=38mm]{CirculantGraph}
%			%\\
%			\includegraphics[width=55mm, height=49mm]{CirculantGraph}
%			%\hskip0.2cm
%			\caption{ The circulant graph with 50 nodes and generating set $Q=\{1,2,5\}$, where edges in  red/green/blue  are also edges of the
% cycle graphs $\mathcal{C}_1$,
%  $\mathcal{C}_2$ and $\mathcal{C}_5$ generated by $\{1\}, \{2\}, \{5\}$ respectively.
%			}
%			\label{CirculantGraph}
%			\vspace{-2em}
%		\end{center}
%	\end{figure}

In this section, we
	demonstrate
	the iterative optimal polynomial approximation (IOPA)  algorithm \eqref{optimaliterativedistributedalgorithm.eqn1}
	and
	the iterative  Chebyshev polynomial approximation (ICPA)  algorithm  \eqref{Chebysheviterativedistributedalgorithm.eqn1}
	to implement an  inverse  filtering procedure,
	and compare their performances with the gradient descent method \eqref{gradientdescent.al} with zero initial  \cite{Shi15},
	and the autoregressive moving average (ARMA) algorithm  \eqref{ARMA00} and \eqref{arma.convergence}  \cite{Leus17}.
	
	Let $N\ge 1$ and $Q=\{q_1, \ldots, q_M\}$ be a set of integers ordered so that $1\le q_1<\ldots<q_M<  N/2$.
 The {\em circulant} graph ${\mathcal C}(N, Q)$  generated by $Q$
has the vertex set  $V_N=\{0, 1, \ldots, N-1\}$ and the edge set
\begin{equation}\label{circulant.edgedef}
E_N(Q)=\{(i,i\pm q\ {\rm mod}\ N),\  i\in V_N, q\in Q\}, \end{equation}
where  $a=b\ {\rm mod }\ N$ if $(a-b)/N$ is an integer.
% \cite{ ekambaram13, vnekambaram13, dragotti19a, dragotti19,  valsesia19}.
Shown in Figure \ref{CirculantGraph} is a circulant graph with $N=50$ and  $Q=\{1,2,5\}$.
   Circulant graphs are widely used in image processing
%\cite{ ekambaram13, vnekambaram13}--\cite{valsesia19}.
 {	\cite{ekambaram13, vnekambaram13, dragotti19a,dragotti19,valsesia19}}.   In Subsection \ref{circulant.demon.subsection}, we demonstrate the performance of the proposed IOPA
and ICPA  algorithms
on the implementation of the inverse filtering on circulant graphs.

Graph signal denoising  is one of  the most popular applications in graph filtering %\cite{sandryhaila14}--\cite{Waheed18}, \cite{jiang19, mario19, segarra17, Coutino17, Leus17}
{\cite{cheng_SDS16,mario19,Coutino17, Leus17, jiang19,sandryhaila14, segarra17, Waheed18,moura14}}. and in some cases, it can be recasted as an inverse filtering procedure. In Subsection  \ref{denoising.subsection},
 we consider denoising  noisy sampling  % \cc{delete the sampling?}
 data
	\begin{equation}
	{\bf b}_{i}= {\bf x}(t_i)+\pmb \eta_i, \ 1\le i\le M,
	\end{equation}
	of some time-varying graph signal
	${\bf x}(t)$  on  random geometric graphs, which is
		governed by a  differential equation
	\begin{equation}\label{de.eq00}
	{\bf x}^{\prime\prime}(t)= {\bf P} {\bf x}(t),
	\end{equation}
	where  ${\pmb \eta}_i, 1\le i\le M$, are  noises with noise level
$\eta=\max_{1\le i\le M}\|{\pmb \eta}_i\|_\infty$,
 the sampling procedure is taken uniformly at $t_i=t_1+(i-1)\delta, 1\le i\le M$, with uniform sampling gap $\delta>0$,
	%at  $t_1<t_2<\ldots< t_M$ with maximal gap $\delta=\max_{1\le i\le M-1} |t_{i+1}-t_i|$,
		 and
	${\bf P}$ is a graph  filter  with small geodesic-width.

Finally in Subsection  \ref{denoisingweather.subsection}, we  apply
the proposed IOPA and ICPA algorithms to denoise  the  hourly temperature  dataset collected at $218$ locations in the United States.

	\subsection{Iterative approximation algorithms on circulant graphs}
\label{circulant.demon.subsection}

\begin{table*}[t]
		\renewcommand\arraystretch{1.2}
		\centering
		\caption{ 			Average   relative iteration error
			over 1000 trials for the ARMA method,  GD0 algorithm,  and
			 IOPA  and  ICPA algorithms with different degrees
to implement the inverse filtering  ${\bf b}\longmapsto {\bf H}_1^{-1} {\bf b}$ on the circulant graph ${\mathcal C}(1000, Q_0)$. % with $N=1000$. %(done) {\color{red} Nazar: order (ARMA, GD0, ICPA0, ICPA1, ICPA1.... Delete the column IOPA0}
		}
		%\label{CirculantGraphICPA.Table}			
       \begin{tabular} {|c|c|c|c|c|c|c|c|c|c|c|c|}			
			\hline
			\backslashbox{Alg.}{AE}{m}& 1 & 2 & 3 & 4 & 5 & 7 & 9 &  11 & 14 &  17 & 20     \\
			\hline
			ARMA & .3259 & .2583 & .1423 & .1098 & .0718 & .0381 & .0207 & .0113 & .0047 & .0019 & .0008 \\
			 \hline
			GD0   &  .2350 & .0856 & .0349 & .0147 & .0063 & .0012 & .0002 & .0000 & .0000 & .0000 & .0000 \\
			 \hline
			ICPA0 & .5686	&	.4318 & .3752 & .3521 & .3441 & .3460 & .3577 & .3743 & .4061 & .4451 & .4913 \\
			 \hline
			ICPA1 & .4494 & .2191 & .1103 & .0566 & .0295 & .0082 & .0024 & .0007 & .0001 & .0000 & .0000 \\
			 \hline
			ICPA2 & .1860 & .0412 & .0098 & .0024 & .0006 & .0000 & .0000 & .0000 & .0000 & .0000 & .0000 \\
			 \hline
			IOPA1 & .1545 & .0266 & .0047 & .0008 & .0002 & .0000 & .0000 & .0000 & .0000 & .0000 & .0000 \\
			 \hline
			ICPA3 & .0979 & .0113 & .0014 & .0002 & .0000 & .0000 & .0000 & .0000 & .0000 & .0000 & .0000 \\
			 \hline
			ICPA4 & .0499 & .0030 & .0002 & .0000 & .0000 & .0000 & .0000 & .0000 & .0000 & .0000 & .0000 \\
			 \hline
			IOPA2 & .0365 & .0019 & .0001 & .0000 & .0000 & .0000 & .0000 & .0000 & .0000 & .0000 & .0000 \\
			 \hline
			ICPA5 & .0225 & .0007 & .0000 & .0000 & .0000 & .0000 & .0000 & .0000 & .0000 & .0000 & .0000 \\
			 \hline
			IOPA3 & .0167 & .0003 & .0000 & .0000 & .0000 & .0000 & .0000 & .0000 & .0000 & .0000 & .0000 \\
			 \hline
			IOPA4 & .0044 & .0000 & .0000 & .0000 & .0000 & .0000 & .0000 & .0000 & .0000 & .0000 & .0000 \\
			 \hline
			IOPA5 & .0019 & .0000 & .0000 & .0000 & .0000 & .0000 & .0000 & .0000 & .0000 & .0000 & .0000 \\
			 \hline
		\end{tabular}
		\label{ComparisonOneShift}
	\end{table*}

	\begin{figure}[t]   %[h]
		\begin{center}
			%\includegraphics[width=38mm, height=38mm]{CirculantGraph}
			%\includegraphics[width=38mm, height=38mm]{CirculantGraph}\\
			%\includegraphics[width=38mm, height=38mm]{CirculantGraph}
			%\includegraphics[width=38mm, height=38mm]{CirculantGraph}
			%\\
			\includegraphics[width=55mm, height=49mm]{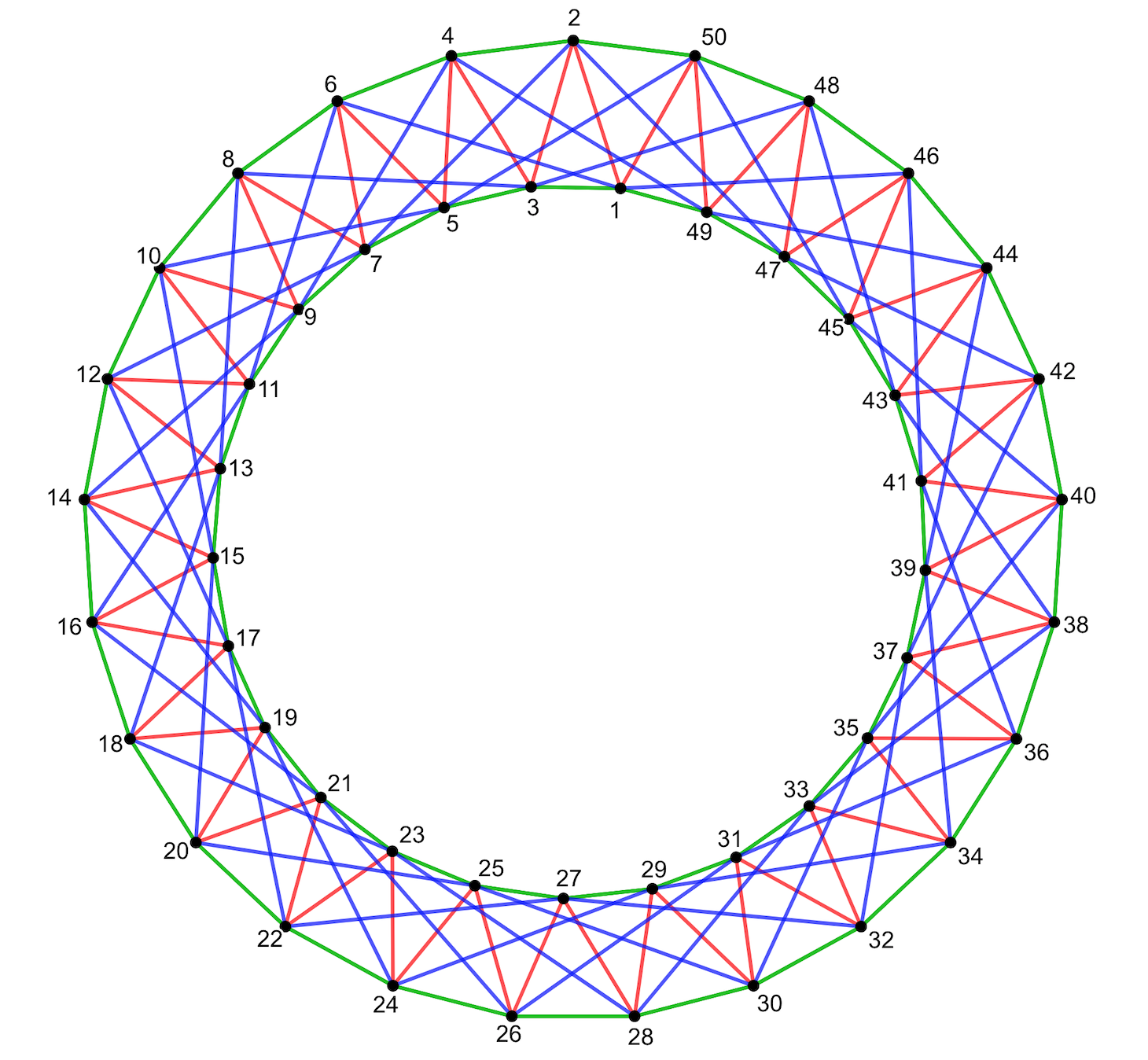}
			%\hskip0.2cm
			\caption{ The circulant graph with 50 nodes and generating set $Q_0=\{1,2,5\}$, where edges in  red/green/blue  are also edges of the
 circulant graphs $\mathcal{C}_1$,
  $\mathcal{C}_2$ and $\mathcal{C}_5$ generated by $\{1\}, \{2\}, \{5\}$ respectively.
			}
			\label{CirculantGraph}
			\vspace{-2em}
		\end{center}
	\end{figure}
% \cite{ ekambaram13, vnekambaram13, dragotti19a, dragotti19,  valsesia19}.
In this section, we consider
 the circulant graph $%{\mathcal C}:=
 {\mathcal C}(N, Q_0)$ generated by $Q_0=\{1,2,5\}$,
 the input graph  signal  ${\bf x}$ with entries randomly  selected in the interval $[-1, 1]$, and  the graph signal ${\bf b}={\bf H}_1{\bf x}$ as the observation, where \begin{equation}\label{h1filter.def}
 h_1(t)=(9/4-t)(3+t)\end{equation}
  and
  $${\bf H}_1=h_1({\mathbf L}^{\rm sym}_{{\mathcal C}(N, Q_0)})$$
   is a polynomial graph filter of
the symmetric normalized Laplacian ${\bf L}^{\rm sym}_{{\mathcal C}(N, Q_0)}$
	on  ${\mathcal C}(N, Q_0)$, see Remark 	\ref{circulantgraph.appendix} for commutative graph shifts on circulant graphs.
	We implement the  inverse filtering ${\bf b}\longmapsto {\bf H}_1^{-1}{\bf b}$ through the  IOPA algorithm \eqref{optimaliterativedistributedalgorithm.eqn1} and ICPA algorithm  \eqref{Chebysheviterativedistributedalgorithm.eqn1}  on the circulant graph ${\mathcal C}(N, Q_0)$.
%For $N=50$, the \st{numerical} evaluations for $\mu_L^*, \ 0\le L\le  5$  in \eqref{optimal.condition}
%are   $0.4501$, $0.1850$, $0.0608$, $0.0210, 0.0060, 0.0023$,
%	and for
%	$b_K,\ 0\le K\le 5,$ in  \eqref{chebyshevapproximation.con}
%	are  $1.0463$, $0.5837, 0.2880, 0.1431, 0.0719, 0.0367$ respectively.
By Theorems \ref{IOPAconvergence.thm}  and \ref{icpa.thm}, the IOPA algorithm with $L\ge 0$ and  the ICPA algorithm  with  $K\ge 1$ converge, and we denote those algorithms by IOPA$L$  and ICPA$K$ for abbreviation. Notice that the filter ${\bf H}_1$ is positive definite, and \vspace{-.4em}\begin{equation*}\frac{1}{h_1(t)}= \frac{4/21}{9/4-t} +\frac{4/21}{3+t} \vspace{-.4em}\end{equation*}
 meets the requirement \eqref{arma.eq1} for the ARMA. For the circulant graph $\mathcal C(N, Q_0)$ with $N=1000$, we also implement the inverse filtering  ${\bf b}\longmapsto {\bf H}_1^{-1}{\bf b}$ by the gradient descent method with zero initial, GD0 in abbreviation, with the optimal step length $\gamma= 2/(6.7500+2.5588)$, and the ARMA method, where $2.5588$ and  $6.7500$ are the minimal and maximal eigenvalues for ${\bf H}_1$ respectively.

% with the optimal step length $\gamma= 2/(6.75+2.56)$, and the ARMA method, where $2.56$ and  $6.75$ are the minimal and maximal eigenvalues for ${\bf H}_1$ respectively.
% \cc{we define the ${\mathcal C}:={\mathcal C}(N, Q_0)$, then we could just use $\mathcal C$ with $N=?$ or not use the notation ${\mathcal C}$. Now we mix the use of two.  }
	
	Set the relative  iteration error
	\begin{equation*} E(m, {\bf x})= {\|{\bf x}^{(m)}-{\bf x}\|_2}/{\|{\bf x}\|_2}, \  m\ge 1,
	\end{equation*}
	where ${\bf x}^{(m)}, m\ge 1$,
	are the output at $m$-th iteration. Shown in Table \ref{ComparisonOneShift} are the comparisons of the   ARMA algorithm, the GD$0$ algorithm,  and
			 IOPA$L$  and  ICPA$K$ algorithms regard to the average of the relative iteration error for implementing the inverse filtering on the circulant graph ${\mathcal C}(1000, Q_0)$  over 1000 trials, where $0\le L, K\le 5$.
This confirms  that  exponential convergence
	and  applicability of the inverse filtering procedure  ${\bf b}\longmapsto {\bf H}_1^{-1}{\bf b}$ of IOPA$L$, $0\le L\le 5$
	and ICPA$K$, $1\le K\le 5$
on the circulant graph ${\mathcal C}(1000, Q_0)$.
 The  average
 exponential convergence rates of
IOPA$L$, $0\le L\le 5$,  % in   Table \ref{ComparisonOneShift}
over 1000 trials are $ 0.4401,  0.1820,  0.0593 ,  0.0208, 0.0067, 0.0023$ respectively,
%$ 0.4369,  0.1820,  0.0593 ,  0.0207, 0.0066, 0.0023$ respectively,
%which is close to the theoretical bound $a_L= 0.4502$, $0.1852$, $0.0612$, $0.0212$, $0.0072$, $0.0025$ for $0\le L\le 5$, see
%\eqref{optimal.condition} in Theorem  \ref{IOPAconvergence.thm}.  Similarly,
and the average
 exponential convergence rates
 of  ICPA$K$, $1\le K\le 5$, are
$0.5485$, $0.2804$, $0.1459$, $0.0685$, $0.0334$ respectively.
%, while
%$0.5392, 0.2743, 0.1456, 0.0685,    0.0334$ respectively, while
%  The average exponential convergence rates of
%IOPA$L$, $0\le L\le 5$  in Table \ref{ComparisonOneShift}  are $ 0.4169,  0.1770,  0.0541 ,  0.0200, 0.0052, 0.0021$ respectively, and them of  ICPA$K$, $1\le L\le 5$, are
%$0.5268, 0.2449, 0.1231, 0.0634,    0.0286$ respectively.
It is observed that average
 exponential convergence rates
 of  IOPA$L$, $0\le L\le 5$ and ICPA$K$, $1\le K\le 5$,
 are close to their theoretical bounds $a_L, 0\le L\le 5$ in \eqref{optimal.condition} % in Theorem  \ref{IOPAconvergence.thm}
 and $b_K, 1\le K\le 5$
 in  \eqref{chebyshevapproximation.con} %of Theorem \ref{icpa.thm}
  respectively, which are listed in the caption of  Figure
\ref{approximation.fig}.
   % for   ICPA$K$, $1\le L\le 5$,
%	 are   $0.5837$, $0.2924$, $0.1467$, $0.0728$, $0.0367$ respectively.
By  %the black  line in Figure \ref{approximation2.fig} and %which is also confirmed numerically by the data listed in
	the  third %column
	row in Table \ref{ComparisonOneShift}, we see that
	the  ICPA0 %  algorithm  \eqref{Chebysheviterativedistributedalgorithm.eqn1}
%	with $K=0$
 does not yield the desired inverse filtering result. The reason for the divergence is that  the theoretical bound $b_0=1.0463$
in  \eqref{chebyshevapproximation.con} is  strictly larger than one.
%From Table \ref{ComparisonOneShift}, we also observe that
%the average exponential convergence rates of GD0 algorithm and  ARMA method over 1000 trials   % in Table \ref{ComparisonOneShift}
%are $0.4401$ and $0.7431$, while the theoretical convergence rates in
%%are $0.4369$ and $0.7386$, while the theoretical convergence rates in
%Remarks  \ref{remark.grad}  and \ref{remark.arma} are $0.4502$ and $0.7584$ respectively.
%The average exponential convergence rates of GD0 algorithm and  ARMA method    in Table \ref{ComparisonOneShift} are $0.4169$ and $0.7326$, while the theoretical convergence rates in
%Remarks  \ref{remark.grad}  and \ref{remark.arma} are $0.4501$ and $0.7582$ respectively.
 	From Table
	\ref{ComparisonOneShift}, % and Figure \ref{approximation.fig},
we  observe that
	the IOPA$L$ algorithms  with higher
	degree $L$ (resp. the ICPA$K$  with higher degree $K$) have  faster convergence,
and the IOPA$L$ algorithm outperforms the ICPA$K$ algorithm when the same degree $L=K$ is selected.
 Comparing with the ARMA algorithm and the GD0 algorithm, we
	observe
that the proposed  IOPA$L$ algorithms   with $L\ge 1$
	and ICPA$K$ algorithms with $K\ge 2$ have  faster convergence, while
	the GD$0$=IOPA$0$ algorithm outperforms the ICPA$K$ when $K=1$ and the ARMA has slowest convergence. %  and  the ICPA$K$ with $K=0$ does not converge.

	\begin{figure}[t]  %[h]
		\begin{center}
			\includegraphics[width=68mm, height=45mm]{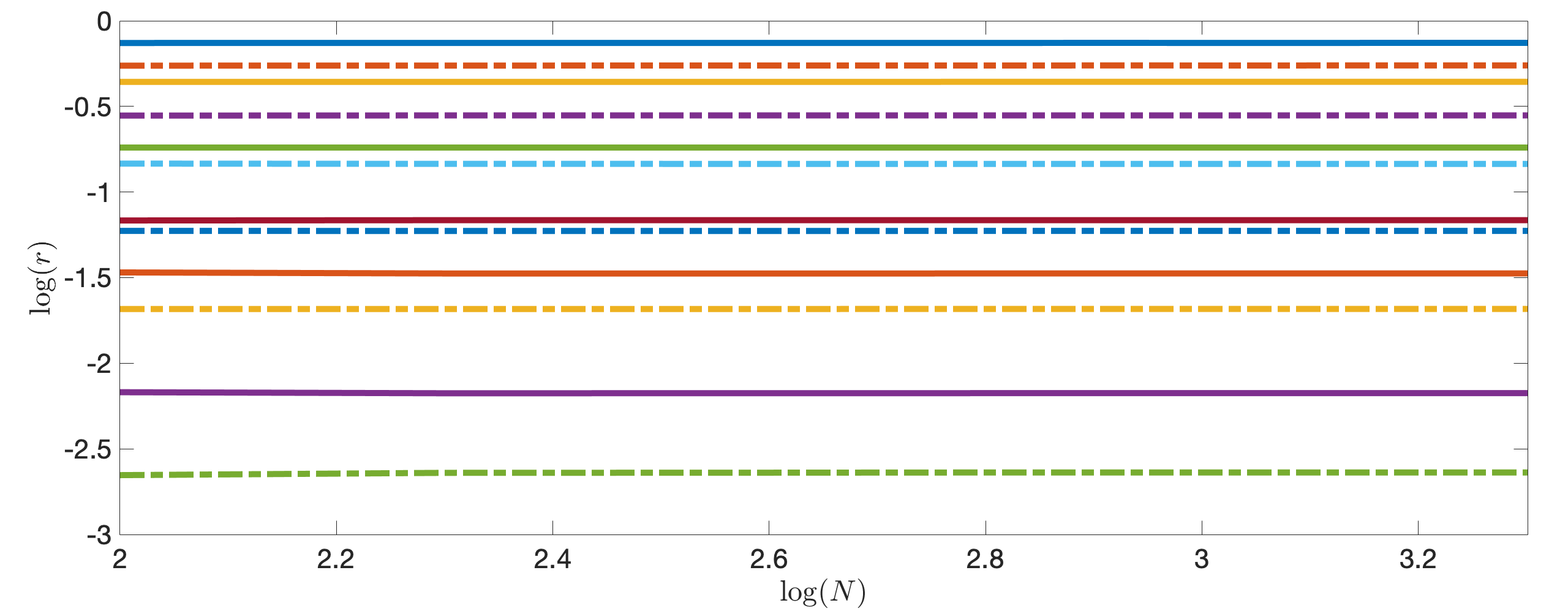}%{Conv_Plot.png} %{TimePlot_log_new.png}
			%\hskip0.2cm
			\caption{ %The average exponential convergence rate $r$ in the logarithmic scale over 1000 trials
%with respect to the
%order $100\le N\le 2000$ of the circulant graph ${\mathcal C}(N, Q_0)$ in the logarithmic scale. \cc{
Plotted from top to  bottom are the average exponential convergence rate $r$ in the logarithmic scale over 1000 trials by ARMA, ICPA1, GD0, ICPA2, IOPA1, ICPA3, ICPA4, IOPA2, ICPA5, IOPA3, IOPA4, IOPA5 to implement the inverse filtering  ${\bf b}\longmapsto {\bf H}_1^{-1}{\bf b}$ on circulant graphs  ${\mathcal C}(N, Q_0)$ with $100\le N\le 2000$, respectively.
%for the ARMA algorithm, the GD$0$,			the IOPA$L$  and  ICPA$K$ algorithms, $1\le L, K\le 5$
  % is almost independent on $N\ge 100$,
	}
			\label{ConvPlot.fig}
		\end{center}
		\vspace{-2em}
	\end{figure}

		\begin{figure}[t]  %[h]
		\begin{center}
			\includegraphics[width=68mm, height=45mm]{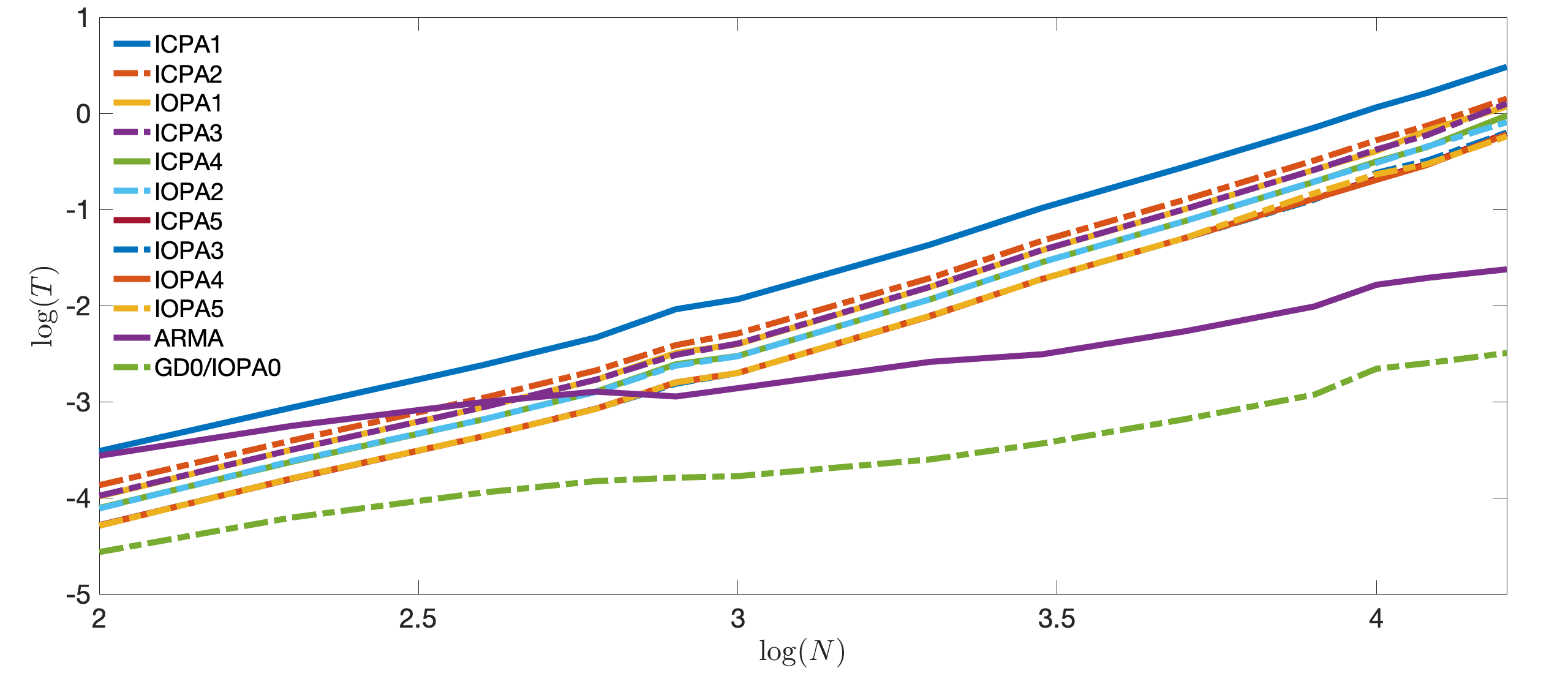}%{Fig1_jiang_new.png}%{Fig1_jiang.png} %{TimePlot_log_new.png}
			%\hskip0.2cm
			\caption{ %\cc{Plotted are the average of total running time $T$ for using the IOPA$L$, ICPA$K$, GD0 and ARMA to  implement inverse filtering and reach the same  relative error level  on circulant graph ${\mathcal C}(N, Q_0)$ with $100\le N\le 16000$ over 1000 trials, where $1\le K, L\le 5$. }
 Plotted are  the average of total running time $T$ in the logarithmic scale  for
 the  GD0, ARMA and  the   IOPA$L$ and ICPA$K$ algorithms with $1\le L, K\le 5$ to  implement the inverse filtering
 on circulant graphs ${\mathcal C}(N, Q_0)$ with $100\le N\le 16000$.
 % the  relative iteration
%error to  reach the relative error less than $ 10^{-3}$ %to implement the IOPA algorithm,  the ICPA algorithm, the GD0
%				algorithm and the ARMA algorithm.
% in the logarithmic scale $\log T$
%with respect to the logarithmic scale $\log N$ of the order $N$ of the circulant graph  ${\mathcal C}(N, Q_0)$ for $100\le N\le 16000$
%, where
%$N=50,100,200,400,600,800,1000,1200,1400,1600,1800,2000$.				
			}
			\label{TotalTime.fig}
		\end{center}
		\vspace{-2em}
	\end{figure}

We also apply ARMA, GD0, and IOPA$L$  and ICPA$K$ with $1\le L, K\le 5$
to implement inverse filtering procedure associated with ${\bf H}_1=h_1({\mathbf L}^{\rm sym}_{{\mathcal C}(N, Q_0)})$
on the circulant graph ${\mathcal C}(N, Q_0)$ with  $h_1$  in \eqref{h1filter.def} and $N\ge 100$.
  All experiments were performed on MATLAB R2017b, running on a DELL T7910 workstation with two Intel Core E5-2630 v4 CPUs (2.20 GHz) and 32GB memory.
 From the simulations, we observe that
the exponential convergence rate  $r$ for the proposed algorithms is almost independent on $N\ge 100$, see Figure  \ref{ConvPlot.fig},
and the number of iterations
 to ensure the relative iteration error  $E(m, {\bf x})\le 10^{-3}$
  are $ 20,  8,  11,   5,   4 ,   4,   3,    3,  2,   2,   2,   2$  for ARMA,  GD0,  ICPA1, ICPA2, IOPA1, ICPA3,   ICPA4,  IOPA2,  ICPA5, IOPA3,  IOPA4,
 IOPA5 respectively.
  Shown in Figure  \ref{TotalTime.fig} is the average running time  $T$ in the logarithmic scale
 over 1000 trials, where  the running time $T$ is measured in seconds
 to ensure the relative iteration error  $E(m, {\bf x})\le 10^{-3}$.
%This  indicates that % to arrive  the same accuracy level,
%  GD0 (IOPA0) are the fastest iterative algorithm to reach the desired accuracy,   ICPA5, IOPA3, IOPA4 and IOPA5 follow
%   for  $100\le N\le 650$, and ARMA for  $N\ge 650$,
%the next are ICPA4 and IOPA2,  IOPA1 and ICPA3 have comparable running time.
From our simulations, we see that
 there is a complicated
trade-off between the convergence rate and the running time to apply our proposed algorithms, ARMA and GD0 for the implementation of an inverse filtering procedure.

	\subsection{Denoising time-varying signals}
\label{denoising.subsection}

	In this section, we consider denoising  noisy sampling  % \cc{delete the sampling?}
 data
	of some time-varying graph signal
	${\bf x}(t)$  on  random geometric graphs, which is
		governed by the  differential equation
\eqref{de.eq00}.
	Discretizing the differential equation \eqref{de.eq00} gives
	\begin{equation}\label{de.eq01}
	\delta^{-2} \big({\bf x}( t_{i+1})+{\bf x}(t_{i-1})-2 {\bf x}(t_i)\big)\approx   {\bf P} {\bf x}(t_i),
	\end{equation}
	where $i=1, \ldots, M$.
	Applying the trivial extension ${\bf x}(t_0)={\bf x}(t_1)$ and ${\bf x}(t_{M+1})= {\bf x}(t_{M})$ around the boundary, we can
	reformulate  \eqref{de.eq01} in a recurrence relation,
	\begin{equation}\label{dis.eq-1}
	{\bf x}( t_{i})\approx (2{\bf I}+\delta^2 {\bf P}) {\bf x}(t_{i-1})-{\bf x}(t_{i-2}), 2\le i\le M,
	\end{equation}
with 	$ {\bf x}(t_0)={\bf x}(t_1)$.
Let ${\mathcal T}=(T, F)$  be the line graph with the vertex set $T=\{t_1, \cdots, t_{M}\}$
	and edge set $F=\{(t_1, t_2), %(t_1, t_0),
	%(t_2, t_3),  %(t_2, t_1),
	\ldots, (t_{M-1}, t_{M})\}\cup \{(t_{M}, t_{M-1}),  %(t_2, t_1),
 \ldots, (t_2, t_1) %(t_{M}, t_{M-1}) %(t_2, t_{2\pm 1}), \ldots,
	\}$.
Denote
Kronecker product of  two matrices ${\bf A}$ and ${\bf B}$
by ${\bf A}\otimes{\bf B}$, and
 the Laplacian  matrix of the line graph ${\mathcal T}$ with vertices $\{t_1, \ldots, t_M\}$ by ${\bf L}_{\mathcal T}$.
Then we can reformulate the recurrence relation \eqref{dis.eq-1}
	in the matrix form
	\begin{equation} \label {de.eq02}
	(\delta^{-2} {\bf L}_{\mathcal T}\otimes {\bf I}+   {\bf I}\otimes {\bf P}){\bf X}\approx {\bf 0},
	\end{equation}
	where
${\bf X}$ is the vectorization of discrete time signals
${\bf x}(t_1), \ldots, {\bf x}(t_M)$.
	In most of applications \cite{ekambaram13,  Grassi18,  Loukas16,  Waheed18},  the time-varying signal ${\bf x}(t)$ at every moment  $t$ has certain smoothness  in the vertex domain,
	which is usually described by
	\begin{equation}\label{timevaryingobservation2}
	({\bf x}(t_i))^T {\bf L}_{\mathcal G}^{\rm  sym} {\bf x}(t_i)\approx 0,\ 1\le i\le M,
	\end{equation}
where ${\bf L}_{\mathcal G}^{\rm  sym}$ is the symmetric normalized Laplacian on the connected, undirected and unweighted
 graph ${\mathcal G}=(V, E)$. Based on the observations \eqref{de.eq02} and \eqref{timevaryingobservation2}, we propose
	the following Tikhonov
regularization approach
	%\vspace{-.5em}
	\begin{equation}\label{ExtendedBasicEnergyModel}
	\widehat {\bf X} := {\rm arg}\min_{\bf Y}  \|{\bf Y}-{{\bf B}}\|_2^2+\alpha {\bf Y}^T ( {\bf I}\otimes {{\bf L}^{\rm sym}_{\mathcal{G}}}
	) {\bf Y} +\beta  {\bf Y}^T (\delta^{-2} {\bf L}_{\mathcal T}\otimes {\bf I} +  {\bf I}\otimes {\bf P})  {\bf Y},
%	\vspace{-.6em}
	\end{equation}
%	\vspace{-.5em}\begin{eqnarray}\label{ExtendedBasicEnergyModel}
%	\widehat {\bf X}&\hskip-0.08in :=& \hskip-0.08in {\rm arg}\min_{\bf Y}  \|{\bf Y}-{{\bf B}}\|_2^2+\alpha {\bf Y}^T ( {\bf I}\otimes {{\bf L}^{\rm sym}_{\mathcal{G}}}
%	) {\bf Y}\nonumber\\
%	& & \quad +\beta  {\bf Y}^T (\delta^{-2} {\bf L}_{\mathcal T}\otimes {\bf I} +  {\bf I}\otimes {\bf P})  {\bf Y},
%	\vspace{-.6em}\end{eqnarray}
	where ${\bf B}$ is  the vectorization of the observed noisy  data
	${\bf b}_1, \ldots, {\bf b}_M$,   and $\alpha, \beta$ are penalty constants in the vertex  and ``temporal" domains to be appropriately chosen \cite{Kurokawa17}.

	Set
	\begin{equation*}
	{\bf D}_{\alpha, \beta}= {\bf I}+ \alpha  {\bf I}\otimes {{\bf L}^{\rm sym}_{\mathcal{G}}}
	+\beta   ( \delta^{-2}{\bf L}_{\mathcal T}\otimes {\bf I} + {\bf I}\otimes {\bf P}), \ \alpha, \beta\ge 0.
	\end{equation*}
	%{\color{red} \st{(Size of each ${\bf I}$ is different here, do we need to mention it?)}}
The minimization problem \eqref{ExtendedBasicEnergyModel}
	has an explicit solution
	\begin{equation}\label{denoisingprocedure}
	\widehat {\bf X}= ({\bf D}_{\alpha, \beta})^{-1} {\bf B},
	\end{equation}
	when ${\bf I}+\alpha {\bf L}^{\rm sym}_{\mathcal G}+\beta  {\bf P}$ is positive definite.
%	\begin{equation} {\bf I}+\alpha {\bf L}^{\rm sym}_{\mathcal G}+\beta  {\bf P}\succ {\bf 0}.
%	\end{equation}
Set
${\bf S}_1=
	{\bf I}\otimes {{\bf L}^{\rm sym}_{\mathcal{G}}}$ and ${\bf S}_2=\frac{1}{2}{\bf L}_{\mathcal T}\otimes {\bf I}$.
As shown in Proposition \ref{cartesianproduct.pr} of Appendix \ref{productgraph.appendix},
   ${\bf S}_1$ and ${\bf S}_2$ are commutative graph shifts on the
	Cartesian product graph ${\mathcal T} \times {\mathcal G}$. Moreover one may obtain from
\eqref{kroneckerproductrule}  and \eqref{jointspectrum.def} that  ${\bf S}_1$ and ${\bf S}_2$ have their joint spectrum
 contained in $[0, 2]^2$.
	Therefore for the case that
	${\bf P}=
	p({\bf L}^{\rm sym}_{\mathcal G})$ for some polynomial $p$,
	${\bf D}_{\alpha, \beta}=h_{\alpha, \beta}({\bf S}_1, {\bf S}_2)$ is a polynomial graph filter of commutative graph filters ${\bf S}_1$ and ${\mathbf S}_2$,
% of    shifts	${\bf S}_1$ and ${\bf S}_2$,
where $h_{\alpha, \beta}(t_1, t_2)= 1+ \alpha t_1+\beta p(t_1)+ 2\beta \delta^{-2} t_2$.
Moreover, one may verify that ${\bf D}_{\alpha, \beta}$ is positive definite if
$$h_{\alpha, \beta}(t_1, t_2)>0, \ \ 0\le t_1, t_2\le 2,$$ %\vspace{-.4em}\end{equation*}
%The above requirement on $h_{\alpha, \beta}$ can be reformulated as $1+ \alpha t_1+\beta p(t_1)>0$ for all $0\le t_1\le 2$,
which is  satisfied if
$1+\beta p(t_1)> 0$ for all $0\le t_1\le 2$.
% \cc{I did not get the idea of the following sentence. I thought we only need to mention``
%We use
%	the IOPA  algorithm \eqref{optimaliterativedistributedalgorithm.eqn1}
%	and
%	the ICPA  algorithm  \eqref{Chebysheviterativedistributedalgorithm.eqn1}
%to implement the denoising procedure \eqref{denoisingprocedure} as  ${\bf D}_{\alpha, \beta}$ is a polynomial of commutative graph shifts.
%''}
Hence
	we may use
	the IOPA  algorithm \eqref{optimaliterativedistributedalgorithm.eqn1}
	and
	the ICPA  algorithm  \eqref{Chebysheviterativedistributedalgorithm.eqn1}
with the polynomial filter ${\bf H}$ being replaced by ${\bf D}_{\alpha, \beta}$
	to implement the denoising procedure \eqref{denoisingprocedure}.
%, \ 0\le t_1, t_2\le 2.\vspace{-.4em}\end{equation*}.
By the  exponential convergence  of the proposed algorithms, we may use
their outputs   at $m$-th iteration with large $m$ as  denoised time-varying signals.

\begin{figure}[t]  %[t]
		\begin{center}
		\includegraphics[width=60mm, height=50mm]{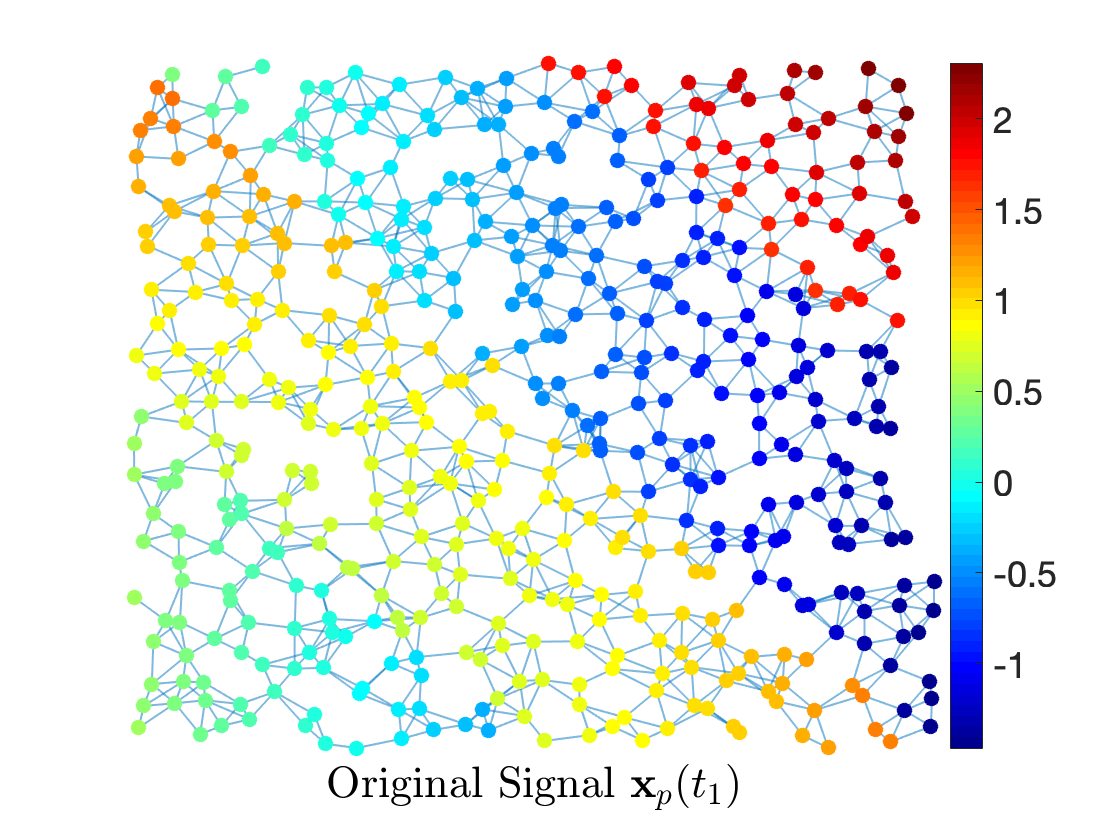}
	\includegraphics[width=60mm, height=50mm]{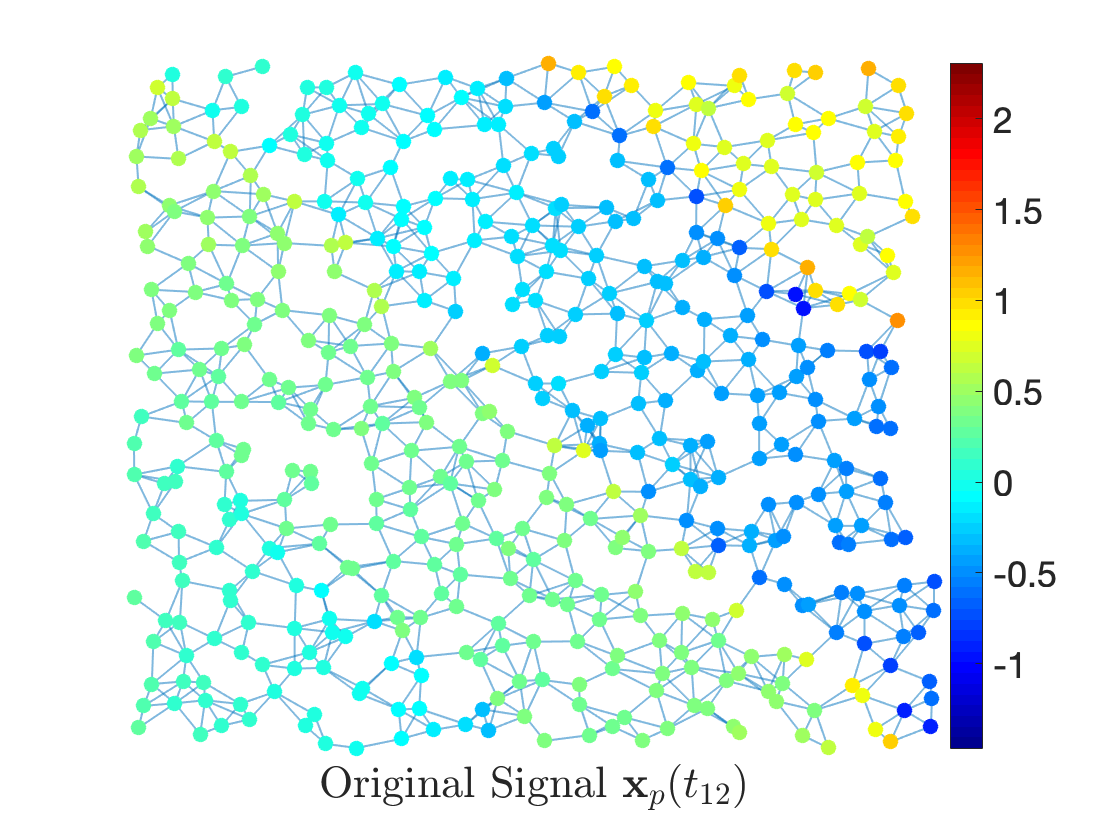}
			\caption{
Presented on the left and right  are the first snapshot ${\bf x}_p(t_1)$ and the middle snapshot ${\bf x}_p(t_{12})$
				of a time-varying signal ${\bf x}_p(t_m), 1\le m\le 24$, on  the random geometric graph ${\mathcal G}_{512}$ respectively.
The qualities
$({\bf x}_p(t_m))^T {\bf L}_{{\mathcal G}_{512}}^{\rm  sym} {\bf x}_p(t_{m})$
to measure smoothness of ${\bf x}_p(t_m)$ in the vertex domain
 are $84.1992$ and $42.4746$ for $m=1, 12$  respectively,
% to measure their  smoothness  of those snapshots are
%??? and ??? respectively,
 which indicates that the graph signal ${\bf x}(t_{12})$ in the last snapshot is smoother than the
initial graph signal ${\bf x}(t_1)$.
 }			\label{RGG_512}
		\end{center}
		\vspace{-1em}
	\end{figure}

Let ${\mathcal G}_{512}$ be the random geometric graph reproduced by the GSPToolbox, which has
 512 vertices randomly deployed in the region $[0, 1]^2$ and
  an edge existing between two vertices if their physical
distance is not larger than
$\sqrt{2/512}=1/16$  \cite{jiang19, Nathanael2014}.
Denote the symmetric normalized Laplacian matrix on ${\mathcal  G }_{512}$ by ${\bf L}^{\rm sym}_{{\mathcal G}_{512}}$ and
 the coordinates of a vertex $i$ in ${\mathcal G}_{512}$ by $(i_x, i_y)$.
%{\color{blue} In our simulations of this section, we will
%denoise the time-varying signal ${\bf x}(t)$
%		governed by the differential equation \eqref{de.eq00}, where
%				the governing filter  is given by ${\bf P} =-{\bf I}+ {\bf L}^{\rm sym}_{{\mathcal G}_{512}}/2$ and
%				the initial graph signal ${\bf x}_p(t_1)$ is a blockwise polynomial consisting of four strips and imposing  $(0.5-2 i_x)$ on the first and third diagonal strips and $(0.5 + i_x^2 + i_y^2)$ on the second and fourth strips respectively  \cite{jiang19}.
%Shown in Figure \ref{RGG_512} are two snapshots of the time-varying graph signal in \eqref{de.eq00},
%where $M=24, \delta=0.1$, $t_{m}= (m-1)\delta$ for $1\le m\le M$.}
For the simulations in this section, the time-varying signal ${\bf x}(t_m), 1\le m
\le M$, is given in  \eqref{de.eq01}, where  $M=24, \delta=0.1$, % $t_{m}= (m-1)\delta$ for $1\le m\le M$,
				the governing filter  is given by ${\bf P} =-{\bf I}+ {\bf L}^{\rm sym}_{{\mathcal G}_{512}}/2$, and
				the initial graph signal ${\bf x}_p(t_1)$ is a blockwise polynomial consisting of four strips and imposing  $(0.5-2 i_x)$ on the first and third diagonal strips and $(0.5 + i_x^2 + i_y^2)$ on the second and fourth strips respectively  \cite{jiang19}. Shown in Figure \ref{RGG_512} are two snapshots of the above time-varying graph signal.
% It is observe that the  snapshots ${\bf x}_p(t_m), 1\le m\le M$, have different smoothness at different moments.

 Appropriate selection of the penalty constants $\alpha, \beta$ in the vertex
and temporal domains is crucial to have a satisfactory denoising performance.
In the simulations,
we let noise entries of  $\pmb\eta_i, 1\le i\le 24$ in \eqref{de.eq00}, be  i.i.d. variables  uniformly selected  in the range $[-\eta,\eta]$, and
we take %penalty constants $\alpha$ and $\beta$ on model fitting of time-varying signals
%in the vertex
%and temporal domains by
\begin{equation} \label{alpha0.def01}
\alpha =  \frac{ {\mathbb E} \|{\bf B}-{\bf X}\|_2^2}
%(NM \eta^2}
{  {\mathbb E} \big({\bf B}^T ( {\bf I}\otimes {\bf L}^{\rm sym}_{{\mathcal G}_{512}}) {\bf B}\big)}
 =
\frac{MN \eta^2/3}
{  {\bf X}^T ( {\bf I}\otimes {\bf L}^{\rm sym}_{{\mathcal G}_{512}}) {\bf X}+  MN\eta^2/3}
 \approx
\frac{\eta^2}{0.2306+\eta^2},
\end{equation}
%\vspace{-.4em}
%\begin{eqnarray} \label{alpha0.def01}
%\alpha =  \frac{ {\mathbb E} \|{\bf B}-{\bf X}\|_2^2}
%%(NM \eta^2}
%{  {\mathbb E} \big({\bf B}^T ( {\bf I}\otimes {\bf L}^{\rm sym}_{{\mathcal G}_{512}}) {\bf B}\big)}
%&\hskip-0.08in = & \hskip-0.08in
%\frac{MN \eta^2/3}
%%(NM \eta^2}
%{  {\bf X}^T ( {\bf I}\otimes {\bf L}^{\rm sym}_{{\mathcal G}_{512}}) {\bf X}+  MN\eta^2/3}
%\nonumber\\
%&  \hskip-0.13in \approx & \hskip-0.13in
%\frac{\eta^2}{0.2306+\eta^2},
%\vspace{-.6em}\end{eqnarray}
and
  \begin{eqnarray}\label{beta0.def01}
  \beta&  \hskip-0.13in = & \hskip-0.13in \frac{ {\mathbb E} \|{\bf B}-{\bf X}\|_2^2}
%(NM \eta^2}
{ 2 {\mathbb E} \big({\bf B}^T
(\delta^{-2} {\bf L}_{\mathcal T}\otimes {\bf I}+   {\bf I}\otimes {\bf P}) {\bf B}\big)}\approx 0.0026
% \nonumber\\
%&  \hskip-0.13in= & \hskip-0.13in
%\frac{NM}{2{\rm tr} (\delta^{-2} {\bf L}_{\mathcal T}\otimes {\bf I}+   {\bf I}\otimes {\bf P})}
%=\frac{2M\delta^2}{4 M-4 -M\delta^2}
%\approx 0.0026
\end{eqnarray}
to balance
the  fidelity term and the regularization terms on the vertex and temporal domains
 in  the Tikhonov
regularization approach
\eqref{ExtendedBasicEnergyModel}.
%where ${\rm tr}({\bf A})$ is the trace of   a matrix $\bf A$.

We use the IOPA  algorithm \eqref{optimaliterativedistributedalgorithm.eqn1} with $L=1$,
	the ICPA  algorithm  \eqref{Chebysheviterativedistributedalgorithm.eqn1} with $K=1$
	and the gradient descent method  \eqref{gradientdescent.al} with zero initial
to implement the inverse filter procedure  ${\bf B}\longmapsto \widehat {\bf X}={\bf D}_{\alpha, \beta}^{-1} {\bf B}$, denoted  by
IOPA1$(\alpha, \beta)$,  ICPA1$(\alpha, \beta)$ and GD0$(\alpha, \beta)$ respectively.
Let $\widehat {\bf X}^{(m)}, m\ge 1$, be the  outputs of either the IOPA1$(\alpha, \beta)$ algorithm, or the  ICPA1$(\alpha, \beta)$ algorithm, or the  GD0$(\alpha, \beta)$ method
 at $m$-th iteration.
To measure the denoising performance of our approaches, we define
the input
 signal-to-noise ratio
 \begin{equation*}{\rm ISNR}=-20 \log_{10} %\frac
{\| {\bf B}-{\bf X}\|_2}/{\|{\bf X}\|_2}, \end{equation*}
and the output signal-to-noise ratio
 \begin{equation*}{\rm SNR}(m)=-20 \log_{10} %\frac
{\|\widehat {\bf X}^{(m)}-{\bf X}\|_2}/{\|{\bf X}\|_2},\  m\ge 1,\end{equation*}
and
 \begin{equation*}{\rm SNR}(\infty)=-20 \log_{10} %\frac
{\|\widehat {\bf X}-{\bf X}\|_2}/{\|{\bf X}\|_2}.
\end{equation*}
%%%%%%%%%%%%%%%%%%%%%%%%%%%%%%%%%%%%%%%%
%%%%%%%%%%%%%%%%%%%%%%%%%%%%%%%%%%% Table (Nazar) %%%%%%%%%%%%%%%%%%%%%%%%%%%%%%%%%%%%%%%%%%%%%%
\begin{table}[ht]
		\centering
		\caption{
			The average   of the signal-to-noise ratio
${\rm SNR}(m), m=1, 2, 4, 6,  \infty$ for  the noise level $\eta=3/4, 1/2, 1/4, 1/8$ over 1000 trials, where
penalty constants $\alpha$ and $\beta$ are given in \eqref{alpha0.def01} and \eqref{beta0.def01} respectively.		}
		\begin{tabular} {|c|c|c|c|c|c|}
		 \cline{1-6}
			  \backslashbox{Alg.}{SNR} {m}  & 1 & 2 & 4  & 6  %&8
 &$\infty $    \\
		 \hline
		 \hline
%			 \multirow{4}{*}{$\eta =\frac{1}{8}, \alpha=0.0271,\beta=1.2385*10^{-5}$}
     \multicolumn{6}{c} {$\eta$=3/4, ISNR=  3.3755} %\vline
     \\
\hline\hline

 IOPA1($\alpha$, 0)  & 6.5777 &	6.8047	&6.7927  &	6.7926 &	6.7926\\
IOPA1(0, $\beta$)& 6.0597 &	6.0907  &	 6.0735  &	6.0735 &	6.0735\\
IOPA1($\alpha$, $\beta$)& 7.4797 &	8.5330  & 	8.4942 &		8.4931 & 8.4930\\

\hline
			 ICPA1($\alpha$, 0) &6.4581  &   6.8169&     6.7928 &    6.7926  &   6.7926\\
            ICPA1(0, $\beta$)&6.0433 &    6.0899&     6.0735   &  6.0735  &  6.0735\\
			 ICPA1($\alpha$, $\beta$)&    7.4036 &     8.4602 &    8.4924   &  8.4930 &    8.4930\\
 \cline{1-6}
			 GD0($\alpha$, 0)& 4.9399  &  6.7283   &  6.8062 &   6.7943   &  6.7926\\
           GD0(0, $\beta$)&     5.0027&     6.3873   &  6.1225   &  6.0787   &  6.0735\\
			 GD0($\alpha$, $\beta$)&     4.1778  &   6.9998   &  8.3432 &    8.4750  &   8.4930\\
		\hline	
\hline
		    \multicolumn{6}{c} {$\eta$=1/2, ISNR=6.8975
} \\
\hline\hline
			 IOPA1($\alpha$, 0)&   9.2211&    9.3576    & 9.3544 &    9.3544    & 9.3544\\
            IOPA1(0, $\beta$)&    9.4981 &    9.6116 &    9.5949 &    9.5949 &    9.5949\\
  		    IOPA1($\alpha$, $\beta$)&  10.0425  &  11.0678   & 11.0624   & 11.0620 &   11.0620\\
			
 \cline{1-6}
			 ICPA1($\alpha$, 0)&   9.1525  &   9.3617  &   9.3544   &  9.3544  &   9.3544\\
           ICPA1(0, $\beta$)&    9.5037   &  9.6110  &   9.5949&     9.5949  &   9.5949\\
			 ICPA1($\alpha$, $\beta$)&   9.7218 &   11.0092   & 11.0613 &   11.0620  & 11.0620\\
 \cline{1-6}
			 GD0($\alpha$, 0)&         7.1610   &  9.2163  &   9.3568   &  9.3546&     9.3544\\
            GD0(0, $\beta$)&  6.8746  &   9.5953 &    9.6392  &   9.6000   &  9.5949\\
			 GD0($\alpha$, $\beta$)&  5.3263   &  8.9866  &  10.8804 &  11.0423   & 11.0620\\
\hline\hline
		    \multicolumn{6}{c} {$\eta$=1/4, ISNR= 12.9164} %\vline
 \\  \hline\hline
			 IOPA1($\alpha$, 0)& 13.8837   & 13.9053  &  13.9053 &   13.9053  &  13.9053\\
            IOPA1(0, $\beta$)&      15.0923   & 15.6251  &  15.6109  &  15.6108  &  15.6108\\
			 IOPA1($\alpha$, $\beta$)& 14.6334&    15.9121  &  15.9192  &  15.9192  &  15.9192\\
 \cline{1-6}
			 ICPA1($\alpha$, 0)& 13.8693 &   13.9055 &   13.9053   & 13.9053 &   13.9053\\
            ICPA1(0, $\beta$) & 15.2045   & 15.6255 &   15.6109  &  15.6108  &  15.6108\\
			 ICPA1($\alpha$, $\beta$)&  14.1329  &  15.8756  &  15.9190 &   15.9192  &  15.9192\\
 \cline{1-6}
			 GD0($\alpha$, 0)&  12.2195  &  13.8694 &   13.9052  &  13.9053 &   13.9053\\
            GD0(0, $\beta$) &   8.5703  &  14.2275&    15.6302 &   15.6153   & 15.6108\\
			 GD0($\alpha$, $\beta$)&   7.2800 &   12.7687  &  15.7309 &   15.9044  &  15.9192\\
 \hline\hline
 		    \multicolumn{6}{c} {$\eta$=1/8, ISNR= 18.9370} %\vline
 \\  \hline\hline
			 IOPA1($\alpha$, 0)& 19.2287   & 19.2299  &  19.2299 &   19.2299  &  19.2299\\
            IOPA1(0, $\beta$)& 19.7355   & 21.6233 &   21.6187  &  21.6187  &  21.6187\\
			 IOPA1($\alpha$, $\beta$)& 19.1231&    21.6012  &  21.6335  &  21.6336  &  21.6336\\
 \cline{1-6}
			 ICPA1($\alpha$, 0)& 19.2275 &   19.2299 &   19.2299   & 19.2299 &   19.2299\\
            ICPA1(0, $\beta$) & 20.1427   & 21.6275 &   21.6187  &  21.6187  &  21.6187\\
			 ICPA1($\alpha$, $\beta$)&  18.9606  &  21.5903  &  21.6335 &   21.6336  &  21.6336\\
 \cline{1-6}
			 GD0($\alpha$, 0)&  18.5134  &  19.2279 &   19.2299  &  19.2299 &   19.2299\\
            GD0(0, $\beta$) &   9.1228  &  17.0384 &    21.5398 &   21.6206   & 21.6187\\
			 GD0($\alpha$, $\beta$)&   8.6071 &   16.0831  &  21.3347 &   21.6140  &  21.6336\\
 \hline\hline
		\end{tabular}
		\label{ProductDenoising.table}
		%\vspace{-1em}
	\end{table}
Presented in Table \ref{ProductDenoising.table}
	are
the average  over 1000 trials of
${\rm ISNR}$ and  ${\rm SNR}(m),
m=1, 2, 4, 6, \infty$.
%for  penalty constant  pairs $(\alpha, \beta)$  being either
%$(\alpha_0, 0)$, or  $( 0, \beta_0)$ or $(\alpha_0, \beta_0)$
%for the IOPA$(\alpha, \beta)$,  IOPA$(\alpha, \beta)$ and GD$(\alpha, \beta)$ algorithms
%at different iteration steps,
%	where  $N=512$ and ${\bf X}$ is  the vectorization of the  time-varying signal  ${\bf x}_p(t_i), 1\le i\le 24$ in Figure
%\ref{RGG_512}.
From Table \ref{ProductDenoising.table}, we observe
that the denoising procedure ${\bf B} \longmapsto \widehat {\bf X}={\bf D}_{\alpha, \beta}^{-1}{\bf B}$  via
Tikhonov regularization \eqref{ExtendedBasicEnergyModel}  on  the temporal-vertex domain
can improve
the  signal-to-noise ratio in the range from 2dBs to 5dBs, depending on the noise level $\eta$. Also we see that
  the denoising procedure ${\bf B} \longmapsto \widehat {\bf X}^{(m)}$  via the output   of  the $m$-th iteration in
IOPA1($\alpha, \beta$) algorithm with $m\ge 2$,  the
GD0$(\alpha, \beta)$ method  and the ICPA1$(\alpha, \beta)$ algorithm with  $m\ge 4$ have similar denoising performance.
 Due to  the correlation of time-varying signals across the joint temporal-vertex domains,
 it is expected that
 the  Tikhonov regularization \eqref{ExtendedBasicEnergyModel} on  the temporal-vertex domain
has  better denoising performance than
  Tikhonov regularization either only on  the vertex domain (i.e., $\beta=0$ in \eqref{ExtendedBasicEnergyModel})
  or only on the temporal domain
(i.e., $\alpha=0$ in \eqref{ExtendedBasicEnergyModel}) do.
The above  performance  expectation is confirmed in Table \ref{ProductDenoising.table}.  We remark that denoising approach via
the  Tikhonov regularization  on  the temporal-vertex domain is an inverse filtering procedure of
a polynomial graph filter of {\bf two} shifts, while
the one either on  the vertex domain  or on the temporal domain only is an inverse filtering procedure of
a polynomial graph filter of {\bf one} shift.

	\vspace{-0.2em}	
\subsection{Denoising  an hourly temperature dataset}
\label{denoisingweather.subsection}

%For a vector ${\bf x}=(x(i))_{i\in V}$ we denote its maximum norm  by $\|{\bf x}\|_\infty$ respectively,
%	
%
%
%time-varying signal: two shifts vs single shifts.
%
%real-world problem
%
%{\color{red} Do more large $N$ to demonstrate the convergence rate is independent on the size of graph $N$}
%	

	%\begin{figure}[t]  %[t]
			\begin{figure*}[t]  %[t]
\begin{center}		
%		\begin{subfigure}[b]{0.49\linewidth}        %% or \columnwidth
        \includegraphics[width=65mm, height=42mm]{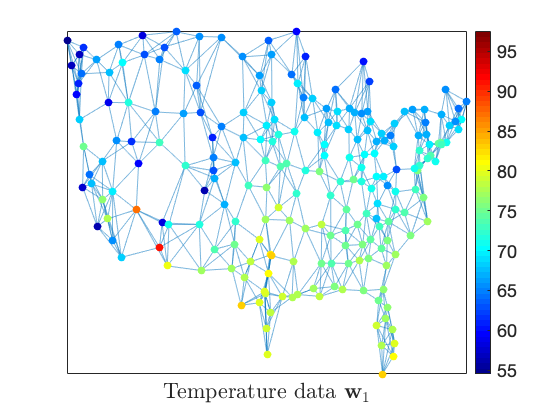}
 %   \end{subfigure}
  %  \begin{subfigure}[b]{0.49\linewidth}        %% or \columnwidth
        \includegraphics[width=65mm, height=42mm]{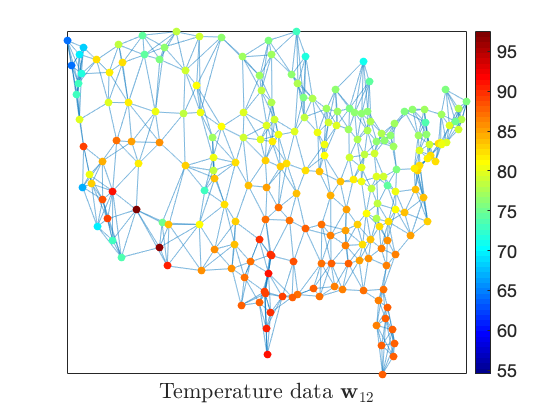}
%    \end{subfigure}
			\caption{Presented on the left and right sides are the temperature data ${\bf w}_1$ and ${\bf w}_{12}$,
  where ${\bf w}_i$, $1\le i\le 24$, are the hourly temperature of 218 locations in the United States on August 1st, 2010.
			%For the above time-varying signal, its average energy $(\sum_{i=1}^{24}\|{\bf w}(i))\|_2^2)/24= {\bf W}^T {\bf W}/24$ is $268.81$, where ${\bf X}$ is the vectorization of the time-varying signal ${\bf w}(i), 1\le i\le 24$.
			}		
			\label{UShourly_temp}
		\end{center}
			\vspace{-1em}
	\end{figure*}

In the section, we consider denoising  the  hourly temperature  dataset collected at $218$ locations in the United States on  August 1st, 2010,  measured in Fahrenheit \cite{cheng2021, zeng17}.
The above real-world dataset is of size  $218 \times 24$, and it can be modelled as a time-varying signal ${\bf w}(i), 1\le i\le 24$,
on the product graph ${\mathcal C}\times {\mathcal W}$,  where
${\mathcal C}:={\mathcal C}(24, \{1\})$ is the circulant graph  with 24 vertices and  generator $\{1\}$, and
 ${\mathcal W}$ is the undirected graph with  $218$ locations as vertices and
 edges  constructed by the  5 nearest neighboring algorithm, see Figure \ref{UShourly_temp} for two snapshots of the dataset.
Unlike in the  simulation in the last subsection, the above time-varying signal ${\bf w}(t_i), 1\le i\le 24$, is  not necessarily to be governed by a different equation of the form \eqref{de.eq00}.

Given noisy temperature data
\begin{equation*}
\widetilde {\bf w}_i={\bf w}_i+{\pmb \eta}_i, \  i=1, \ldots, 24,
\end{equation*}
we propose
	the following denoising approach,
	\begin{equation}\label{Weatherminimization}
	\widehat {\bf W} := {\rm arg}\min_{\bf Z}  \|{\bf Z}-{\widetilde {\bf W}}\|_2^2+\tilde\alpha {\bf Z}^T ( {\bf I}\otimes {\bf L}^{\rm sym}_{\mathcal{W}}
	) {\bf Z} +\tilde \beta  {\bf  Z}^T  ({\bf L}_{\mathcal C}^{\rm sym}\otimes {\bf I}) {\bf Z},
	\end{equation}
%	\vspace{-.3em}\begin{eqnarray}\label{Weatherminimization}
%	\widehat {\bf W}&\hskip-0.08in :=&\hskip-0.08in  {\rm arg}\min_{\bf Z}  \|{\bf Z}-{\widetilde {\bf W}}\|_2^2+\tilde\alpha {\bf Z}^T ( {\bf I}\otimes {\bf L}^{\rm sym}_{\mathcal{W}}
%	) {\bf Z}\nonumber\\
%	& & \qquad\quad +\tilde \beta  {\bf  Z}^T  ({\bf L}_{\mathcal C}^{\rm sym}\otimes {\bf I}) {\bf Z},
%	\vspace{-.6em}\end{eqnarray}
	where $\widetilde {\bf W}$ is  the vectorization of the  noisy temperature data
	$\widetilde {\bf w}_1, \ldots, \widetilde {\bf w}_{24}$  with  noises ${\pmb \eta}_i, 1\le i\le 24$ in \eqref{de.eq00}
having their components randomly selected in $[-\eta, \eta]$ in a uniform distribution, ${\bf L}^{\rm sym}_{\mathcal W}$ and
${\bf L}_{\mathcal C}^{\rm sym}$ are normalized Laplacian matrices on the graph ${\mathcal W}$ and ${\mathcal C}$ respectively,
 and $\tilde \alpha, \tilde\beta\ge 0$ are penalty constants in the vertex and temporal domains to be appropriately selected.

%on observation noises and penalties  on the vertex and temporal domains.

  Set
 $\tilde {\bf S}_1={\bf I}\otimes {\bf L}^{\rm sym}_{\mathcal{W}}$, $\tilde {\bf S}_2={\bf L}_{\mathcal C}^{\rm sym}\otimes {\bf I}
 $
 and
 ${\bf F}_{\tilde\alpha, \tilde \beta}={\bf I}+\tilde\alpha \tilde {\bf S}_1+ \tilde\beta \tilde {\bf S}_2, \
\tilde \alpha, \tilde \beta\ge 0$.
%
% $$f_{\tilde \alpha,\tilde \beta}(t_1, t_2)= 1+\tilde \alpha t_1+\tilde \beta t_2,\ 0\le t_1, t_2\le 2.$$
 One may verify that %  $\tilde {\bf S}_1$ and $\tilde {\bf S}_2$ are commutative graph shifts on the
%	Cartesian product graphs  ${\mathcal C}\times{\mathcal G}$
%	 with their joint spectrum
% eigenvalues being contained in $[0, 2]^2$, and
  the explicit solution of the minimization problem  \eqref{Weatherminimization} is given by
$\widehat {\bf W}= ({\bf F}_{\tilde\alpha, \tilde\beta})^{-1} {\widetilde
{\bf W}}$, and
 the proposed approach to  denoise  the temperature dataset becomes an inverse filtering procedure \eqref{inverseprocedure}
with ${\bf H}$   and ${\bf b}$ replaced by ${\bf F}_{\tilde\alpha,\tilde \beta}$ and $\widetilde {\bf W}$ respectively.
	In absence of notation, we still denote  the IOPA  algorithm \eqref{optimaliterativedistributedalgorithm.eqn1} with $L=1$,
	the ICPA  algorithm  \eqref{Chebysheviterativedistributedalgorithm.eqn1} with $K=1$
	and the gradient descent method  \eqref{gradientdescent.al} with initial zero
to implement the inverse filter procedure  $\widetilde {\bf W}\longmapsto {\bf F}_{\tilde\alpha, \tilde\beta}^{-1} \widetilde {\bf W}$ by
IOPA1$(\tilde\alpha, \tilde\beta)$,  ICPA1$(\tilde \alpha, \tilde\beta)$ and GD0$(\tilde\alpha, \tilde\beta)$ respectively.

 In our simulations, we
 % let
 %noise entries of  ${\pmb \eta}_i, 1\le i\le 24$ be  i.i.d. variables  uniformly selected  in the range $[-\eta,\eta]$, and we
  take %the penalty  constants  $\tilde \alpha$ and $\tilde \beta$ as
%\vspace{-0.6em}\begin{equation*}
%\tilde\alpha  =  \frac{\mathbb{E}\|{\bf Z}-{\widetilde {\bf W}}\|_2^2}
%{\mathbb{E} \big({\widetilde {\bf W}^T} \tilde {\bf S}_1 {\widetilde{\bf W}})}
% =
%\frac{ 4096\eta^2}
%{ {\bf W}^T \tilde {\bf S}_1 {\bf W}+ 4096\eta^2}
%,
%\vspace{-0.4em}\end{equation*}
\begin{equation*}
\tilde\alpha  =  \frac{\mathbb{E}\|{\bf Z}-{\widetilde {\bf W}}\|_2^2}
{\mathbb{E} \big({\widetilde {\bf W}^T} \tilde {\bf S}_1 {\widetilde{\bf W}})}
 =
\frac{ 1744\eta^2}
{ {\bf W}^T \tilde {\bf S}_1 {\bf W}+ 1744\eta^2}
\end{equation*}
and
% \vspace{-0.6em} \begin{equation*} \tilde\beta =   \frac{\mathbb{E}\|{\bf Z}-{\widetilde {\bf W}}\|_2^2}
%  { \mathbb{E} (\widetilde {\bf W}^T \tilde {\bf S}_2 \widetilde{\bf W})}= \frac{4096\eta^2}{  {\bf W}^T \tilde {\bf S}_2 {\bf W}+ 4096\eta^2}
% \vspace{-0.6em} \end{equation*}
 \begin{equation*} \tilde\beta =   \frac{\mathbb{E}\|{\bf Z}-{\widetilde {\bf W}}\|_2^2}
  { \mathbb{E} (\widetilde {\bf W}^T \tilde {\bf S}_2 \widetilde{\bf W})}= \frac{1744\eta^2}{  {\bf W}^T \tilde {\bf S}_2 {\bf W}+ 1744\eta^2}
 \end{equation*}
  to balance three terms in  the regularization approach \eqref{Weatherminimization}.
Presented in Table \ref{denoising_UShourlytemperature.Table}
	are
the average  over 1000 trials of the input signal-to-noise ratio
${\rm ISNR}$ and
the output signal-to-noise ratio %{\color{red} (need to replace $\widetilde {\bf W}$ with ${\bf W}$)}
	\vspace{-0.5em}$${\rm SNR}(m)=-20 \log_{10} \frac{\|\widehat{\bf W}^{(m)}-{\bf W}\|_2}{\|{\bf W}\|_2},\  m\ge 1,\vspace{-0.6em}$$
	which are used to measure the denoising performance of the IOPA1$(\tilde\alpha, \tilde\beta)$,  ICPA1$(\tilde\alpha, \tilde\beta)$ and GD0$(\tilde\alpha, \tilde\beta)$ at the $m$th iteration,  where
$\widehat {\bf W}^{(\infty)}:=\widehat {\bf W}$ and
$\widehat {\bf W}^{(m)}, m\ge 1$, are  outputs of
the IOPA1$(\tilde\alpha, \tilde\beta)$ algorithm, or the  ICPA1$(\tilde\alpha, \tilde\beta)$, or the  GD0$(\tilde\alpha, \tilde\beta)$  at $m$-th iteration.
%Shown in Table \ref{denoising_UShourlytemperature.Table}
%	are our denoising performance of the IOPA$(\tilde\alpha, \tilde\beta)$,  IOPA$(\tilde\alpha, \tilde\beta)$ and GD$(\tilde\alpha, \tilde\beta)$
%at different iteration steps,
%	where   $\widehat {\bf W}$ is  the vectorization of the  time-varying signals  ${\bf w}(t_i), 1\le i\le 24$,
%observation noises ${\pmb \eta}_i, 1\le i\le M$ in \eqref{de.eq00}
%have their components randomly selected in $[-\eta, \eta]$ in a uniform distribution,
%$\widetilde {\bf W}^{(m)}, m\ge 1$ are  outputs of
%the IOPA$(\tilde\alpha, \tilde\beta)$ algorithm, or the  IOPA$(\alpha, \beta)$, or the  GD$(\alpha, \beta)$  at $m$-th iteration,
%the average  over 200 trial of the signal-to-noise ratio
%$${\rm SNR}(m)=-20 \log_{10} \frac{\|\widetilde{\bf W}^{(m)}-\widetilde {\bf W}\|_2}{\|\widetilde {\bf W}\|_2},\  m\ge 1,$$
%
%  {\color{red} Jiang: I am not sure whether the above selection of $\alpha, \beta$ are in good shape or not, but theoretically it is.
% For the simulation, please do
%the noise level at $\eta=1/2, 1/4, 1/8$, then explain your simulation results. 			
%			 is the Tikhonov denoising performance with different weighted factors and different noise levels.
%}
% the output signal-to-noise ratio ${\rm SNR}(m),
%m=1, 2, 4, 6, \infty$, where noise entries of  $\eta_i, 1\le i\le 24$ are  i.i.d. variables  uniformly selected  in the range $[-\eta,\eta]$ with $\eta=5, 10, 20, 35$.
 From Table \ref{denoising_UShourlytemperature.Table}, we see that
   the Tikhonov regularization on the
temporal-vertex domain has better  performance on denoising the hourly temperature dataset than
the Tikhonov regularization only either on the vertex domain  (i.e. $\tilde\beta=0$) or
on the temporal domain   (i.e. $\tilde \alpha=0$)
do.  Also we observe that  the temporal correlation has larger influence than the vertex correlation for small noise corruption $\eta\le 10$, while
the influence of the vertex correlation is more significant than the temporal correlation for the moderate and  larger noise corruption.
%}

%  with the penalty constants $\tilde{\alpha}, \tilde{\beta}$ both in vertex and temporal domain is  better than the denoising performance with  penalty constants only on vertex domain $\tilde{\beta}=0$ or only on the temporal domain $\tilde{\alpha}=0$ for moderate and large noise corruption $\eta\ge 10$. {\color{red} From Table \ref{denoising_UShourlytemperature.Table}, the denoising performance with the nonzero penalty constants $\tilde{\alpha}, \tilde{\beta}$ is better than  that of other parameters scenarios, owing to the  exploitation of  the joint correlation over the vertex and temporal domains. Therefore, the proposed filter with multi-shift and the corresponding algorithms for the inverse filtering have the ability to perform denoising of the real-world datasets. }

%{\color{blue} the 2nd observation is wrong} \st{2) for $m\ge 1$, applying the output of $m$-th iteration of  the IOPA algorithm to denoise the weather dataset has better denoising performance
%that using the output of $m$-th iteration of the gradient descent method  and the ICPA algorithm;}

%%%%%%%%%%%%%%%%%%%%%%%%%%%%%%%%% Table (Nazar) %%%%%%%%%%%%%%%%%%%%%%%%%

		\begin{table}[ht]
			\centering
			\caption{	The average  over 1000 trials of the signal-to-noise ratio
${\rm SNR}(m), m=1, 2, 4, 6, \infty$  denoise
 the US hourly temperature  dataset collected at $218$ locations on  August 1st, 2010, where  $\eta=35, 20, 10$.
  } \label{denoising_UShourlytemperature.Table}
			\begin{tabular} {|c|c|c|c|c|c|}
		 \cline{1-6}
			  \backslashbox{Alg.}{SNR} {m}  & 1 & 2 & 4&6  &$\infty $    \\
		 \hline
		 \hline
%			 \multirow{4}{*}{$\eta =\frac{1}{8}, \alpha=0.0271,\beta=1.2385*10^{-5}$}
     \multicolumn{6}{c} {$\eta$=35, ISNR=     11.5496} %\vline
     \\
\hline\hline
			 IOPA1($\tilde{\alpha}$, 0)&  14.8906 &   16.2623  &  16.2499  &  16.2497 &   16.2497 \\
            IOPA1(0, $\tilde{\beta}$)&  13.3792  &  15.7143 &   15.6925  &  15.6911   & 15.6911\\
     		 IOPA1($\tilde{\alpha}$, $\tilde{\beta}$)& 11.2985  &  18.1294&    19.0536  &  19.0491&    19.0487\\
\hline
			 ICPA1($\tilde{\alpha}$, 0)&   14.2783 &   16.3118&    16.2509 &   16.2498&    16.2497\\
           ICPA1(0, $\tilde{\beta}$)&     14.0451   & 15.7475   & 15.6925  &  15.6911 &   15.6911\\
			 ICPA1($\tilde{\alpha}$, $\tilde{\beta}$)&   9.8634  &  16.9294&    19.0281   & 19.0486&    19.0487\\
 \cline{1-6}
			 GD0($\tilde{\alpha}$, 0)&  7.2407 &   13.2001 &   16.1692   & 16.2523  &  16.2497\\
            GD0(0, $\tilde{\beta}$)&   5.7453  &  10.8805  &  15.3374   & 15.7069 &    15.6911\\
			 GD0($\tilde{\alpha}$, $\tilde{\beta}$)&    3.9579    & 7.8606  &  14.4865 &   17.9663 &   19.0487\\
		\hline	
\hline
		    \multicolumn{6}{c} {$\eta$=20, ISNR=   16.4086 } \\
\hline\hline
			 IOPA1($\tilde{\alpha}$, 0)&   18.3271&    20.2473   & 20.2470 &   20.2470   & 20.2470\\
            IOPA1(0, $\tilde{\beta}$)&   15.4936  &  20.4129   & 20.5195&    20.5183   & 20.5183\\
			 IOPA1($\tilde{\alpha}$, $\tilde{\beta}$)&  12.3927  &  21.0773 &   22.8075  &  22.8097  &  22.8095\\
 \cline{1-6}
			 ICPA1($\tilde{\alpha}$, 0)&  17.5792  &  20.2654  &  20.2474  &  20.2470   & 20.2470\\
            ICPA1(0, $\tilde{\beta}$)& 16.73029  &  20.5223  &  20.5196  &  20.5183  &  20.5183\\
			 ICPA1($\tilde{\alpha}$, $\tilde{\beta}$)&   10.7460  &  19.4217  &  22.7759 &   22.8092  &  22.8095\\
 \cline{1-6}
			 GD0($\tilde{\alpha}$, 0)&    8.4637&   15.7834 &    20.1310&    20.2470  &  20.2470 \\
            GD0(0, $\tilde{\beta}$)&  5.9817 &   11.7217  &  19.1824 &    20.4607  &  20.5183\\
			 GD0($\tilde{\alpha}$, $\tilde{\beta}$)&     4.2594  &  8.4753  &  16.1761 &   21.0514&    22.8095\\
\hline\hline
 \multicolumn{6}{c} {$\eta$=10, ISNR=22.4320} %\vline
 \\  \hline\hline
			 IOPA1($\tilde{\alpha}$, 0)& 23.3572  &  24.5564 &   24.5565&    24.5565  &  24.5565\\
            IOPA1(0, $\tilde{\beta}$)&   16.9511  &  25.9123 &   26.4291 &   26.4284 &    26.4284 \\
		     IOPA1($\tilde{\alpha}$, $\tilde{\beta}$)&    14.2863 &   24.9125   & 26.9961  &  26.9990  &  26.9990\\
 \cline{1-6}
			 ICPA1($\tilde{\alpha}$, 0)&     22.5720  &  24.5572 &   24.5565&    24.5565  &  24.5565\\
            ICPA1(0, $\tilde{\beta}$) &     18.6319  &  26.2493  &  26.4294 &   26.4285 &   26.4284\\
			 ICPA1($\tilde{\alpha}$, $\tilde{\beta}$)& 12.7428  &  23.3488 &   26.9816   & 26.9989   & 26.9990\\
 \cline{1-6}
			 GD0($\tilde{\alpha}$, 0)&  11.7089 &   21.2276 &   24.5387 &   24.5566   & 24.5565\\
            GD0(0, $\tilde{\beta}$) &      6.2342  &  12.3916 &   22.7545 &   26.1414 &   26.4284\\
			 GD0($\tilde{\alpha}$, $\tilde{\beta}$)&  4.9806  &  9.9239  &  19.2003  &  25.2121 &   26.9990\\
	 \hline	\hline
%\multicolumn{6}{c} {$\eta$=5, ISNR=28.4515} %\vline
% \\  \hline\hline
%			 IOPA1($\tilde{\alpha}$, 0)& 29.0542  &  29.2064 &   29.2064 &    29.2064  &  29.2064\\
%            IOPA1(0, $\tilde{\beta}$)&   18.7842  &  31.1069 &   32.0763 &   32.0763 &    32.0763 \\
%		     IOPA1($\tilde{\alpha}$, $\tilde{\beta}$)&    17.3155 &   30.1455   & 31.9043  &  31.9053  &  31.9053\\
% \cline{1-6}
%			 ICPA1($\tilde{\alpha}$, 0)&     28.8442  &  29.2065 &   29.2064&    29.2064  &  29.2064\\
%            ICPA1(0, $\tilde{\beta}$) &     20.4974  &  31.6507  &  32.0766 &   32.0763 &   32.0763\\
%			 ICPA1($\tilde{\alpha}$, $\tilde{\beta}$)& 16.5458  &  29.5321 &   31.9029   & 31.9053   & 31.9053\\
% \cline{1-6}
%			 GD0($\tilde{\alpha}$, 0)&  18.4949 &   28.6438 &   29.2064 &   29.2064   & 29.2064\\
%            GD0(0, $\tilde{\beta}$) &      6.9379  &  13.8403 &   26.4300 &   31.6320 &   32.0763\\
%			 GD0($\tilde{\alpha}$, $\tilde{\beta}$)&  6.2628  &  12.4936  &  24.2288  &  30.8303 &   31.9053\\
%	 \hline		
	
		\end{tabular}
			\vspace{-1.5em}
	\end{table}

		\vspace{-1em}
\section{Conclusions and further works}

Polynomial graph filters of multiple shifts are preferable for denoising
and extracting  features  along different dimensions/directions
 for
multidimensional graph signals, such as video or
time-varying signals.
A necessary condition is derived in this paper for a graph filter to be a polynomial  of multiple  shifts
and the necessary condition is shown to be sufficient if the elements in the joint spectrum of multiple shifts are distinct.
%, see Theorem \ref{polynomialfilter.thm}.
The design methodology of  polynomial  filters of multiple graph shifts and their inverses
%with certain
%spectral characteristic. % in engineering applications.
with
specific features and physical interpretation for engineering applications will be discussed in our future works.

Some  Tikhonov regularization approaches on the temporal-vertex domain to denoise a time-varying signal
 can be reformulated as an inverse filtering procedure for a polynomial graph filter of two shifts which represent the features on the temporal  and vertex domain respectively.
%To implement an inverse filtering directly,
% a centralized implementation may suffer from high
%computational burden  as the inverse graph filter usually has
%full geodesic-width.
Two exponentially convergent  iterative  algorithms
are introduced for
the inverse filtering procedure of a polynomial graph filter, and  each iteration of the proposed algorithms can be implemented
%in a centralized facility with linear complexity and also
in  a distributed
network,  where each vertex is
equipped with systems for limited data storage, computation
power and data exchanging facility to its adjacent vertices.
The proposed iterative algorithms are demonstrated to implement the
inverse filtering procedure  effectively and to have satisfactory performance
on denoising multidimensional graph signals.
\appendices
\setcounter{equation}{0}
\setcounter{theorem}{0}
\setcounter{section}{0}
\setcounter{subsection}{0}
\renewcommand{\thesection}{Appendix \Alph{section}}
\renewcommand{\thesubsection}{A.\arabic{subsection}}
\renewcommand{\theequation}{\Alph{section}.\arabic{equation}}
\renewcommand{\thetheorem}{\Alph{section}.\arabic{theorem}}

\section{Commutative graph  shifts}\label{commutative.section}

 Graph shifts
 are building blocks of a polynomial filter and the concept of commutative graph shifts  ${\mathbf S}_1, \ldots, {\mathbf S}_d$
 is similar to  the  one-order delay $z_1^{-1}, \ldots, z_d^{-1}$ in classical multi-dimensional  signal processing.
 %\cc{can we put this sentence at beginning?}
   In Appendices \ref{circulantgraph.appendix} and  \ref{productgraph.appendix}, we introduce two illustrative  families of commutative graph shifts
 on circulant/Cayley graphs and product graphs respectively, see also Subsections \ref{denoising.subsection} and \ref{denoisingweather.subsection} for
 %two families of
 commutative graph shifts with specific features.
 %The concept of commutative graph shifts  ${\bf S}_1, \dots,  {\bf S}_d$
%   plays a similar role  in graph signal processing
%   as the  one-order delay $z_1^{-1}, \ldots, z_d^{-1}$ in classical multi-dimensional signal processing,
%  and
 For commutative  graph shifts  ${\mathbf S}_1, \ldots, {\mathbf S}_d$, we define
 their joint spectrum \eqref{jointspectrum.def} in Appendix \ref{jointspectrum.appendix}, which     %are commutative
 is crucial
   for us to  develop the IOPA and ICPA algorithms  % for inverse filtering
     in Section  \ref{ipaa.section}.

%Let  ${\bf S}_1,...,{\bf S}_d$ be commutative graph shifts.
%
%Commutativity of graph shifts ${\bf S}_1,...,{\bf S}_d$
% are building blocks to design
%	polynomial graph  filters  ${\bf H}$
%%	\vspace{-0.6em}\begin{equation}\label{MultiShiftPolynomial}
%%	{\bf H}=h({\bf S}_1, \ldots, {\bf S}_d)=\sum_{ l_1=0}^{L_1} \cdots \sum_{ l_d=0}^{L_d}  h_{l_1,\dots,l_d}{\bf S}_1^{l_1}\cdots {\bf S}_d^{l_d}
%%	\vspace{-0.6em}\end{equation}
%with certain spectral characteristic.  	
%In  this  appendix,  we give a sufficient condition on graph filters  being polynomial of multiple graph shifts.
	
	Commutativity of multiple graph shifts %${\bf S}_1,...,{\bf S}_d$
 are  essential to design
	polynomial graph  filters  % ${\bf H}$
%	\vspace{-0.6em}\begin{equation}\label{MultiShiftPolynomial}
%	{\bf H}=h({\bf S}_1, \ldots, {\bf S}_d)=\sum_{ l_1=0}^{L_1} \cdots \sum_{ l_d=0}^{L_d}  h_{l_1,\dots,l_d}{\bf S}_1^{l_1}\cdots {\bf S}_d^{l_d}
%	\vspace{-0.6em}\end{equation}
with certain spectral characteristic. If a graph filter ${\bf H}$ is a polynomial of  commutative multiple graph  shifts ${\bf S}_1,...,{\bf S}_d$, then
		it  commutes with  ${\bf S}_k, 1\le k\le d$, i.e., commutators $[{\bf H}, {\bf S}_k]:={\bf H}{\bf S}_k-{\bf S}_k {\bf H}$
 between ${\bf H}$ and ${\bf S}_k, 1\le k\le d$ are always the zero matrix,
		\vspace{-0.6em}\begin{equation}\label{polynomialfilter.thm.eq1} %\label{CommutativityHS}
		[{\bf H}, {\bf S}_k]={\bf 0},\ 1\leq k\leq d.
		\vspace{-0.6em}\end{equation}
The above necessary condition for a graph filter ${\bf H}$ to be   a polynomial of  ${\bf S}_1,...,{\bf S}_d$
 is not sufficient in general. For instance, one may verify that
 any filter ${\bf H}$ satisfies
\eqref{polynomialfilter.thm.eq1} with $d=1$ and ${\bf S}_1={\bf I}$, while  ${\bf H}$ is not necessarily a polynomial  $h({\bf I})= h(1){\bf I}$ of the identity matrix ${\bf I}$.
	For $d=1$, it is shown in \cite[Theorem 1]{aliaksei13}
		that any filter satisfying \eqref{polynomialfilter.thm.eq1}
		is a polynomial filter if the graph shift has distinct  eigenvalues.
In  Theorem \ref{polynomialfilter.thm} of  Appendix \ref{polynomialfilters.appendix},   we show that the  necessary condition  \eqref{polynomialfilter.thm.eq1} is also sufficient
 under the  additional assumption that elements in the joint spectrum of multiple graph shifts ${\mathbf S}_1, \ldots, {\mathbf S}_d$ are distinct. 	
% this  appendix,  we give a sufficient condition on graph filters  being polynomial of multiple graph shifts.

%In the following theorem,  we show that the  necessary condition  \eqref{polynomialfilter.thm.eq1} is also sufficient
% under the  additional assumption that
%		the joint eigenvalues $\pmb \lambda_i, 1\le i\le N$, in  the
%		joint spectrum  $\Lambda$ in \eqref{jointspectrum.def} are distinct.

Let ${\mathcal A}$ be a Banach algebra of graph filters with its norm denoted by $\|\cdot\|_{\mathcal A}$.  Our representative examples
 are the algebra of graph filters with Frobenius norm $\|\cdot\|_F$,   operator algebras ${\mathcal B}(\ell^p), 1\le p\le \infty$, on the space $\ell^p$ of all $p$-summable graph signals,
Gr\"{o}chenig-Schur algebras, Wiener algebra, Beurling algebras,  Jaffard algebras and  Baskakov-Gohberg-Sj\"{o}strand algebras, see   % the survey papers
 \cite{grochenig2010, krishtal2011, shinsun2013, shinsun2019,  shinsun2020} for historical remarks and various applications. %ecent advances.
Denote the set of polynomials of  commutative graph shifts ${\bf S}_1,...,{\bf S}_d$  by ${\mathcal P}:={\mathcal P}({\bf S}_1, \ldots, {\bf S}_d)$. Under the assumption that
${\mathbf S}_k\in {\mathcal A}, 1\le k\le d$, one may verify that all  polynomials of  ${\bf S}_1,...,{\bf S}_d$ reside in the Banach algebra ${\mathcal A}$ too, i.e.,  ${\mathcal P}\subset {\mathcal A}$. For any filter ${\mathbf H}\in {\mathcal A}$, define
its distance to the polynomial set ${\mathcal P}$ of graph shifts ${\bf S}_1,...,{\bf S}_d$  by
\begin{equation}\label{distance.def}
{\rm dist}({\bf H}, {\mathcal P})=\inf_{{\bf P}\in {\mathcal P}} \|{\bf H}-{\bf P}\|_{\mathcal A}.\end{equation}
Under the assumption that the  elements of joint spectrum of multiple graph shifts ${\mathbf S}_1, \ldots, {\mathbf S}_d$ are distinct, we obtain from Theorem \ref{polynomialfilter.thm} that
${\rm dist}({\bf H}, {\mathcal P})=0$ for any filter ${\bf H}\in {\mathcal A}$ satisfying \eqref{polynomialfilter.thm.eq1}.
 In Theorem \ref{distance.thm} of Appendix \ref{distance.appendix}, we establish some quantitative estimates
 to the distance  ${\rm dist}({\bf H}, {\mathcal P}), {\mathbf H}\in {\mathcal A}$, in terms of
 norms of  commutators  $[{\bf H}, {\bf S}_k], 1\le k\le d$, on an unweighted and undirected finite graph.

  % \ref{iterativeapproximation.section}
% to implement the inverse filtering procedure
%${\bf b}\longmapsto {\bf H}^{-1}{\bf b}$ associated with a polynomial filter ${\bf H}$ of
%${\bf S}_1,...,{\bf S}_d$.
%The crucial advantage is that  we can define  of commutative graph shifts.

%		In this Appendix, we   introduce  two illustrative  families of commutative graph shifts
%%	${\bf S}_1,\dots , {\bf S}_d $
%on circulant graphs and Cartesian product graphs. %on an undirected graph ${\mathcal G}$.
%		{\bf Connection to engineering literatures is quite weak.}
%	\cite{aliaksei14}

% another is the  graph signals  residing on  circulant graphs but  with multiple features. For example, the  time-varying graph signal, images, and Netflix movie ratings. Product graph has been used to represent the graph on which multi-dimensional graph signal resides, such as time-varying graph signals are residing on a graph being the product of a sensor network graph and a time period graph, see Fig \ref{Time_varying_graph},  and images are residing  on a graph  being the product of two line graphs \cite{aliaksei14}, see Fig \ref{ImageGraph}.
%	
	%\vspace{-1em}
	\subsection{Commutative graph shifts on circulant graphs and  Cayley graphs}\label{circulantgraph.appendix}

%	\subsection{Commutative shifts on circulant graphs and product graphs}

%\begin{figure}[t]   %[h]
%		\begin{center}
%			%\includegraphics[width=38mm, height=38mm]{CirculantGraph}
%			%\includegraphics[width=38mm, height=38mm]{CirculantGraph}\\
%			%\includegraphics[width=38mm, height=38mm]{CirculantGraph}
%			%\includegraphics[width=38mm, height=38mm]{CirculantGraph}
%			%\\
%			\includegraphics[width=55mm, height=48mm]{CirculantGraph}
%			%\hskip0.2cm
%			\caption{ The circulant graph with 50 nodes and generating set $S=\{1,2,5\}$, where edges in  red/green/blue  are also edges of the
% circulant graphs with 50 nodes and generating sets $\{1\}, \{2\}, \{5\}$ respectively.
%			}
%			\label{CirculantGraph}
%			\vspace{-2em}
%		\end{center}
%	\end{figure}

Let  ${\mathcal C}(N, Q)=(V_N, E_N(Q))$ be the circulant graph of order $N$ generated by
$Q=\{q_1, \ldots, q_M\}$, where $1\le q_1<\ldots<q_M< N/2$, see \eqref{circulant.edgedef}  %{cir.graph}
and Figure \ref{CirculantGraph}.
% is the  circulant graph  ${\mathcal C}(N, Q)$ with $N=50$ nodes and generating set $Q=\{1,2,5\}$.
 Observe that
$$E_N(Q)=\cup_{1\le k\le d}\big \{(i,i\pm q_k\ {\rm mod}\ N), i\in V_N\big\}.$$
%=\cup_{1\le k\le d} E_N(S_k),$$
%where $S_k=\{s_k\},  1\le k\le d$.
Then the circulant graph ${\mathcal C}(N, Q)$ can be  decomposed into a family of circulant graphs
	${\mathcal C}(N, Q_k)$  generated by $Q_k=\{q_k\},  1\le k\le d$,
	and
the symmetric normalized
 Laplacian matrix ${\bf L}_{{\mathcal C}(N,Q)}^{\rm sym}$ on ${\mathcal C}(N, Q)$ is the  average of symmetric normalized
 Laplacian matrices ${\bf L}_{{\mathcal C}(N,Q_k)}^{\rm sym}$ on ${\mathcal C}(N, Q_k), 1\le k\le d$, i.e.,
	\vspace{-.6em}\begin{equation*}
{\bf L}_{{\mathcal C}(N, Q)}^{\rm sym} = \frac{1}{d}
 \sum_{k=1}^d {\bf L}_{{\mathcal C}(N,Q_k)}^{\rm sym},
	\vspace{-.4em}\end{equation*}
where $Q_k=\{q_k\}, 1\le k\le d$.
%For the case that $Q$ contains only one element, we call the  circulant graph  ${\mathcal C}(N, Q)$ by a cycle graph.
In the following proposition, we establish the  commutativity of ${\bf L}_{{\mathcal C}(N,Q_k)}^{\rm sym}, 1\le k\le d$.

% normalized Laplacian matrices
%${\bf L}_{k}^{\rm sym}, %$ on the cycle graphs ${\mathcal C}(S_k),
% 1\le k\le d$,
%are commutative graph shifts on the circulant graph ${\mathcal C}(S)$.
% and	they are commutative.

	\begin{proposition}\label{circulantgraphshift.pr}
		{\rm The symmetric normalized Laplacian matrices ${\bf L}_{{\mathcal C}(N,Q_k)}^{\rm sym}$
		of the circulant graphs  ${\mathcal C}(N, Q_k), 1\le k\le d$,
		are commutative  graph shifts on the circulant graph ${\mathcal C}(N, Q)$.
%		\vspace{-.4em}\begin{equation*}\label{cyclecommutative.eq}
%{\bf L}_{{\mathcal C}(N,Q_{k'})}^{\rm sym} {\bf L}_{{\mathcal C}(N,Q_k)}^{\rm sym}
%= {\bf L}_{{\mathcal C}(N,Q_k)}^{\rm sym} {\bf L}_{{\mathcal C}(N,Q_{k'})}^{\rm sym}
%%{\bf L}^{\rm sym}_k{\bf L}^{\rm sym}_{k'}={\bf L}^{\rm sym}_{k'}{\bf L}^{\rm sym}_k,
%\ 1\leq k,k'\leq d.\vspace{-.4em}\end{equation*}
}	\end{proposition}

	\begin{proof}
Clearly ${\bf L}_{{\mathcal C}(N,Q_k)}^{\rm sym}, 1\le k\le d$, are graph shifts
on the circulant graph ${\mathcal C}(N, Q)$.
Define
$${\bf B}=(b({i-j\ {\rm mod}\  N}))_{1\le i,j\le N},$$
 where
$b(0)=\cdots=b({N-2})=0$ and  $b({N-1})=1$.  Then one may verify that
		\vspace{-.4em}\begin{equation*} {\bf L}_{{\mathcal C}(N,Q_k)}^{\rm sym}= {\bf I}-\frac{1}{2}({\bf B}^{q_k}+{\bf B}^{-q_k})=
-\frac{1}{2}{\bf B}^{-q_k} ({\bf B}^{q_k}-{\bf I})^2,
 \vspace{-.4em}\end{equation*}
		where $1\le k\le d$.
%$${\bf A}=\begin{pmatrix} 0 & 1 & \cdots & 0 & 0\\
%0 & 0 & 1 &  & 0\\
%\vdots & 0 & 0 & \ddots& 0\\
%0 &  & \ddots & \ddots & 1\\
%1 & 0 & \cdots & 0 & 0 \end{pmatrix}$$,
Therefore for $1\le k, k'\le d$,
%\begin{eqnarray*}
% {\bf L}_{{\mathcal C}(N,Q_{k'})}^{\rm sym} {\bf L}_{{\mathcal C}(N,Q_k)}^{\rm sym}
%& \hskip-0.08in = & \hskip-.08in \frac{1}{4} {\bf B}^{-q_k-q_{k'}} ({\bf B}^{q_k}-{\bf I})^2  ({\bf B}^{q_{k'}}-{\bf I})^2
% \\
%& \hskip-0.08in = & \hskip-.08in  {\bf L}_{{\mathcal C}(N,Q_k)}^{\rm sym} {\bf L}_{{\mathcal C}(N,Q_{k'})}^{\rm sym}.
%\end{eqnarray*}
\begin{equation*}
 {\bf L}_{{\mathcal C}(N,Q_{k'})}^{\rm sym} {\bf L}_{{\mathcal C}(N,Q_k)}^{\rm sym}
=  \frac{1}{4} {\bf B}^{-q_k-q_{k'}} ({\bf B}^{q_k}-{\bf I})^2  ({\bf B}^{q_{k'}}-{\bf I})^2
=  {\bf L}_{{\mathcal C}(N,Q_k)}^{\rm sym} {\bf L}_{{\mathcal C}(N,Q_{k'})}^{\rm sym}.
\end{equation*}
% the adjacent matrix associated with  the cycle graph ${\mathcal C}(N, \{1\})$.
%which is .
		This completes the proof. %\eqref{cyclecommutative.eq} follows.
	\end{proof}

	Following the proof of  Proposition \ref{circulantgraphshift.pr}, we have that
the adjacent matrices  $2({\bf I}-  {\bf L}_{{\mathcal C}(N,Q_{k})}^{\rm sym})={\bf B}^{q_k}+{\bf B}^{-q_k} $
		of the circulant graphs  ${\mathcal C}(N, Q_k), 1\le k\le d$, (and their linear combinations)
		are commutative  graph shifts on the circulant graph ${\mathcal C}(N, Q)$.

\bigskip
Connected circulant graphs are regular undirected Cayley graphs of finite cyclic groups.
In general, for an Abelian group  $G$ generated by a finite set $S$ of non-identity elements
 and  a color assignment   $c_s$ to
 each element $s\in S$,  the Cayley graph
 ${\mathcal G}$ is defined to have elements in the
 group $G$ as its vertices
and  directed  edges of color $c_s$ between  vertices $g$ to $gs$ in $G$.
For the case that  the generator $S$ is symmetric (i.e.,  $S^{-1}= S$) and the same color is assigned for
any element in the generator $S$ and its inverse (i.e., $c_s=c_{s^{-1}}, s\in S$), one may verify that  the Cayley graph ${\mathcal G}$ is a regular undirected graph,
it can be decomposed into a family of  regular subgraphs  ${\mathcal G}_s, s\in S_1$ with the same colored edges,
$$ {\mathcal G}=\cup_{s\in S_1} {\mathcal G}_s,$$
and the adjacent matrix of the Cayley graph is the
summation of  the adjacent matrices associated with the regular subgraphs ${\mathcal G}_s, s\in S_1$,
where the subset $S_1\subset S$ is chosen so that different colors are assigned for distinct elements in $S_1$ and
 all colors $c_s, s\in S$ are represented in $S_1$.
Furthermore, similar to the commutativity for symmetric normalized Laplacian matrices  in Proposition \ref{circulantgraphshift.pr}, we can show that
the adjacent matrices, Laplacian matrices, symmetric normalized Laplacian matrices
 associated with the regular subgraphs ${\mathcal G}_s, s\in S_1$,
 are commutative graph shifts of the Cayley graph ${\mathcal G}$. 

	\vspace{-1em}
	\subsection{Commutative graph shifts on Cartesian product graphs} %and time-varying graph signals}
\label{productgraph.appendix}

	Let ${\mathcal G}_1=(V_1, E_1)$ and ${\mathcal G}_2=(V_2, E_2)$ be two finite  graphs with adjacency matrices ${\bf A}_1$ and ${\bf A}_2$.
	Their Cartesian product graph  ${\mathcal G}_1\times {\mathcal G}_2$ %=(V_1\times V_2, E_1 \otimes E_2)$
  has  vertex set $V_1\times V_2$ and  adjacency matrix  %${\bf C}$
	given by
	%\begin{equation}\label{cartesian.prod}
	${\bf A}={\bf A}_1\otimes {\bf I}_{\# V_2}+ {\bf I}_{\# V_1}\otimes {\bf A} _2$
	%\end{equation}
\cite{ Grassi18, Loukas16}. Shown in  Figure \ref{Time_varying_graph} is an illustrative example of product graphs
and
 the time-varying graph signal ${\bf X}$ can be considered as a  signal on the  Cartesian product graph  ${\mathcal T} \times {\mathcal G}$, see Subsection \ref{denoising.subsection}.
 	\begin{figure}[t]  %[h]
		\begin{center}
			\includegraphics[width=88mm, height=42mm]{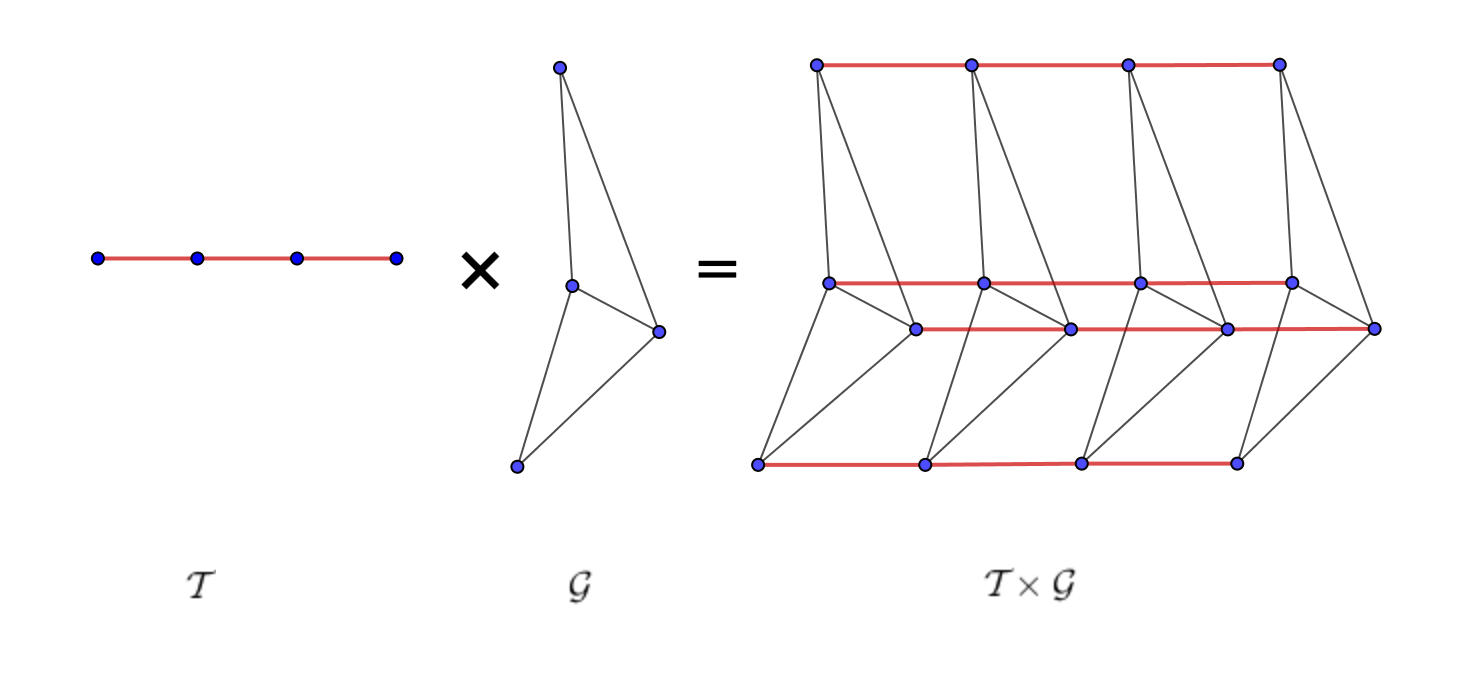}
			\caption{ Cartesian product  ${\mathcal T}\times {\mathcal G}$ of a line graph  $\mathcal{T}$  %with 4 vertices
and an undirected graph $\mathcal{G}$.} % with 4 vertices.}
			\label{Time_varying_graph}
		\end{center}
		\vspace{-1em}
	\end{figure}

Denote  symmetric normalized Laplacian matrices    and
 orders of the  graph ${\mathcal G}_i, i=1, 2$  by   ${\bf L}_i^{\rm sym}$ and $N_i$
 respectively.
	By the mixed-product property
	\vspace{-.4em}\begin{equation}\label{kroneckerproductrule}
	({\bf A}\otimes {\bf B}) ({\bf C}\otimes {\bf D})=({\bf AC})\otimes ({\bf BD})\vspace{-.4em}\end{equation}
	for Kronecker product of matrices ${\bf A}, {\bf B}, {\bf C}, {\bf D}$ of appropriate sizes \cite{Laubbook}, one may verify that
 ${\bf L}_1^{\rm sym}\otimes {\bf I}_{N_2}$ and  $ {\bf I}_{N_1}\otimes {\bf L}_2^{\rm sym}$
are graph filters   of the Cartesian product graph  ${\mathcal G}_1\times {\mathcal G}_2$.
  In the following proposition, we show that they are commutative.

%	{\color{red} references here \cite{Loukas16, Grassi18}.}
%	From the above definition of the Cartesian product graph ${\mathcal G}_1\times {\mathcal G}_2$, we see that
%	for $v_1, \tilde v_1\in V_1$ and  $v_2, \tilde v_2\in V_2$,
%	there is an edge between $(v_1, v_2)$ and $(\tilde v_1, \tilde v_2)$ in the Cartesian product graph ${\mathcal G}_1\times {\mathcal G}_2$ if and only if either $(v_1, \tilde v_1)\in E_1$ and $v_2=\tilde v_2$, or
%	$\tilde v_1=v_1$ and $(v_2, \tilde v_2)\in E_2$.

	\begin{proposition}\label{cartesianproduct.pr}
{\rm  Let ${\mathcal G}_1=(V_1, E_1)$ and ${\mathcal G}_2=(V_2, E_2)$ be two finite  graphs with
		normalized Laplacian matrices ${\bf L}_1^{\rm sym}$ and ${\bf L}_2^{\rm sym}$ respectively. Then
		${\bf L}_1^{\rm sym}\otimes {\bf I}_{\#V_2}$ and  $ {\bf I}_{\#V_1}\otimes {\bf L}_2^{\rm sym}$ are commutative graph shifts of the Cartesian product graph  ${\mathcal G}_1\times {\mathcal G}_2$.
}	\end{proposition}
	
	\begin{proof}  Let $N_i=\# V_i, i=1, 2$ and set ${\bf C}_1={\bf L}_1^{\rm sym}\otimes {\bf I}_{N_2}$ and ${\bf C}_2={\bf I}_{N_1}\otimes {\bf L}_2^{\rm sym}$.
By the mixed-product property \eqref{kroneckerproductrule} of Kronecker product,
Then
	$${\bf C_1}{\bf C}_2 = {\bf L}_1^{\rm sym}\otimes  {\bf L}_2^{\rm sym}=
{\bf C}_2  {\bf C}_1.$$
%	
%	$$({\bf C}_2\otimes {\bf I}_{N})( {\bf I}_{M}\otimes  {\bf C}_1)={\bf C}_2\otimes  {\bf C}_1=( {\bf I}_{M}\otimes  {\bf C}_1)({\bf C}_2\otimes {\bf I}_{N}).$$
	\end{proof}

	\subsection{Joint spectrum  of commutative shifts}\label{jointspectrum.appendix}

		%\section{fff}

%In Theorem \ref{polynomialfilter.thm} of , we extend the above result to
%polynomial filters of multiple graph shifts.
 % satisfying \eqref{commutativityS}.
		%Then
		% they can be upper-triangularized  simultaneously by \cite[Theorem 2.3.3.]{horn1990matrix},
		% i.e.,
		%\begin{equation}
		%\label{upperdiagonalization} \hat{\bf S}_k={\bf U}^*{\bf S}_k{\bf U},\  1\le k\le d,
		%\end{equation}
		%are upper triangular matrices for some unitary matrix ${\bf U}$.
		%Write $\hat{\bf S}_k=(\hat S_{k}(i,j))_{1\le i, j\le N}, 1\le k\le d$, and set
		% \begin{equation}\label{jointspectrum.def}
		% \Lambda=\big\{\lambda_i=\big(\hat{S}_1(i,i), ..., \hat{ S}_d(i,i)\big), 1\le i\le N\big\}.\end{equation}
		%  where $ \hat{ S}_k(i,i), 1\le i\in N$,  are the diagonal entries of the upper triangular matrices $\hat {\bf S}_k, 1\le k\le d$.
%		If a graph filter ${\bf H}$ is a polynomial of  ${\bf S}_1,...,{\bf S}_d$, then
%		it  commutates with  ${\bf S}_k, 1\le k\le d$.
%		In  this appendix, we show that the  necessary condition  \eqref{polynomialfilter.thm.eq1} is also sufficient
% under the  additional assumption that
%		the joint eigenvalues $\pmb \lambda_i, 1\le i\le N$, in  the
%		joint spectrum  $\Lambda$ are distinct.
%		For $d=1$, it is shown in
%\cite[Theorem 1]{aliaksei13}
%		that a filter satisfying
%		is a polynomial filter if the graph filter has distinct  eigenvalues.
		
% The assumption that the graph shifts  ${\mathbf S}_1, \ldots, {\mathbf S}_d$ are commutative,  is indispensable
%   for us to  develop the IOPA and ICPA algorithms  for inverse filtering.

Let   ${\bf S}_1,...,{\bf S}_d$  be commutative graph shifts.
An important property in \cite[Theorem 2.3.3]{horn1990matrix} is that  they can be upper-triangularized  simultaneously over ${\mathbb C}$,
	i.e.,
\begin{equation}
	\label{upperdiagonalization}
\widehat{\bf S}_k={\bf U}{\bf S}_k{\bf U}^{\rm H},\  1\le k\le d,
	\end{equation}
	are upper triangular matrices for some unitary matrix ${\bf U}$.
%, where ${\bf A}^{\rm H}$ is the conjugate transpose of a matrix ${\bf A}$.
	Write
	$\widehat{\bf S}_k=(\widehat S_{k}(i,j))_{1\le i, j\le N}, 1\le k\le d$, and set
	\vspace{-.2em}\begin{equation}\label{jointspectrum.def} \Lambda=\big\{\pmb \lambda_i=\big(\widehat{S}_1(i,i), ..., \widehat{ S}_d(i,i)\big), 1\le i\le N\big\}  \subset {\mathbb C}^d.\vspace{-.2em}\end{equation}
	As $\widehat{S}_k(i, i), 1\le i\le N$, are complex eigenvalues of ${\bf S}_k,\ 1\le k\le d$, we call $\Lambda$ as the {\em joint spectrum} of  ${\bf S}_1, \ldots, {\bf S}_d$.
   The  joint spectrum  $\Lambda$ of commutative shifts
${\bf S}_1, \ldots, {\bf S}_d$
 plays an essential role  in Section \ref{ipaa.section}
  %\ref{iterativeapproximation.section}
in the construction of
%  to
%construct
optimal polynomial approximation filters
  and Chebyshev polynomial approximation filters
 to the inverse filter
of a polynomial filter  of
${\bf S}_1,...,{\bf S}_d$.

We remark that
a sufficient condition for the  graph shifts ${\bf S}_1, \ldots, {\bf S}_d$ to be commutative
is that they can be diagonalized simultaneously, i.e.,
there exists a nonsingular matrix ${\bf P}$ such that ${\bf P}^{-1} {\bf S}_k{\bf P}, 1\le k\le d$, are diagonal matrices,
which a necessary condition is that they can be upper-triangularized  simultaneously, see \eqref{upperdiagonalization}.

\subsection{Polynomial  graph filters of commutative filters}\label{polynomialfilters.appendix}
		
Let  ${\bf S}_1,...,{\bf S}_d$ be commutative graph shifts.
%
%Commutativity of graph shifts ${\bf S}_1,...,{\bf S}_d$
% are building blocks to design
%	polynomial graph  filters  ${\bf H}$
%%	\vspace{-0.6em}\begin{equation}\label{MultiShiftPolynomial}
%%	{\bf H}=h({\bf S}_1, \ldots, {\bf S}_d)=\sum_{ l_1=0}^{L_1} \cdots \sum_{ l_d=0}^{L_d}  h_{l_1,\dots,l_d}{\bf S}_1^{l_1}\cdots {\bf S}_d^{l_d}
%%	\vspace{-0.6em}\end{equation}
%with certain spectral characteristic.  	
%In  this  appendix,  we give a sufficient condition on graph filters  being polynomial of multiple graph shifts.
%		If a graph filter ${\bf H}$ is a polynomial of   ${\bf S}_1,...,{\bf S}_d$, then
%		it  commutates with  ${\bf S}_k, 1\le k\le d$, i.e.,
%		\vspace{-0.6em}\begin{equation}\label{polynomialfilter.thm.eq1} %\label{CommutativityHS}
%		{\bf H}{\bf S}_k={\bf S}_k {\bf H} ,\ 1\leq k\leq d.
%		\vspace{-0.6em}\end{equation}
%		For $d=1$, it is shown in \cite[Theorem 1]{aliaksei13}
%		that any filter satisfying \eqref{polynomialfilter.thm.eq1}
%		is a polynomial filter if the graph shift has distinct  eigenvalues.
In the following theorem,  we show that the  necessary condition  \eqref{polynomialfilter.thm.eq1}
for a filter to be a polynomial of multiple graph shifts is also sufficient
 under the  additional assumption that
		the joint eigenvalues $\pmb \lambda_i, 1\le i\le N$, in  the
		joint spectrum  $\Lambda$ in \eqref{jointspectrum.def} are distinct.

		\begin{theorem}\label{polynomialfilter.thm}
  {\rm
			Let ${\bf S}_1,\ldots, {\bf S}_d$ be commutative graph filters,
%matrices $\hat{\bf S}_1, \ldots, \hat {\bf S}_d$ be as in
%			\eqref{upperdiagonalization},
			and  the joint spectrum ${\Lambda}$  be as in
			\eqref{jointspectrum.def}.
			If all elements ${\pmb \lambda}_i, 1\le i\le N$, in the set ${ \Lambda}$ are distinct,
			then any graph filter ${\bf H}$ satisfying
			\eqref{polynomialfilter.thm.eq1}
			is a polynomial of  ${\bf S}_1,...,{\bf S}_d$, i.e.,
			%\begin{equation} \label{polynomialfilter.thm.eq2}  %\label{CommutativityHS.eq2}
			${\bf H}=h({\bf S}_1,...,{\bf S}_d)$
%			\end{equation}
			for some polynomial $h$.
			%then $\hat{\bf H}= {\bf U}^*{\bf H}{\bf U}$ is an upper triangular.
		} \end{theorem}

		\begin{proof} %  [Proof of Theorem \ref{polynomialfilter.thm}]
			Let ${\bf U}$ be the unitary matrix  in \eqref{upperdiagonalization},
 $\widehat{\bf S}_1, \ldots, \widehat {\bf S}_d$ be upper triangular matrices in
		\eqref{upperdiagonalization},
 and			 $\widehat {\bf H}={\bf U}{\bf H}{\bf U}^{\rm H}= (\widehat H(i,j))_{1\le i, j\le N}$.
			By the assumption on  the set $\Lambda$,
			there exist an interpolating polynomial $h$ such that
			\vspace{-0.4em}\begin{equation}\label{polynomialfilter.thm.pf.eq1}
			h(\widehat{ S}_1(i,i), ..., \widehat{ S}_d(i,i))=\widehat H(i,i), \ 1\le i\le N,
			\vspace{-0.6em}\end{equation}
			see \cite[Theorem 1 on p. 58]{cheney2000course}.
			Set
			\vspace{-0.4em}\begin{equation}\label{HpF}
			{\bf F}= {\bf U}\big({\bf H}-h({\bf S}_1,...,{\bf S}_d)\big) {\bf U}^{\rm H}=
			\widehat {\bf H}-h(\widehat{\bf S}_1, \ldots, \widehat {\bf S}_d).
			\vspace{-0.6em}\end{equation}
			Then it suffices to prove that  %$h$  satisfies \eqref{polynomialfilter.thm.eq2} or equivalently that
			${\bf F}$ is the zero matrix. %  ${\bf O}$.

			Write ${\bf F}=(F(i,j))_{1\le i, j\le N}$.
			By \eqref{polynomialfilter.thm.eq1}, we have that
%			\begin{equation*}\label{upperdiagonalization.lem.pf.eq1}
			${\bf F}\widehat{\bf S}_k=\widehat{\bf S}_k {\bf F}$ for all $1\le k\le d$.
%\end{equation*}
			This together with the upper triangular property for $\widehat{\bf S}_k, 1\le k\le d$, implies that
			\vspace{-0.5em}\begin{equation} \label{upperdiagonalization.lem.pf.eq2}
			\sum_{l=1}^j F(i,l) \widehat S_k(l,j)=\sum_{l=i}^N \widehat  S_k(i, l) F(l, j),\ 1\le i, j\le N.
			\vspace{-0.5em}\end{equation}
			By the assumption on $\Lambda$, we can find $1\le k(i,j)\le d$ for any $1\le i\neq j \le N $  such that
			\vspace{-0.6em}\begin{equation}\label{upperdiagonalization.lem.pf.eq3}
			\widehat S_{k(i,j)}(i, i)\ne \widehat S_{k(i,j)}(j, j).
			\vspace{-0.6em}\end{equation}
			Now we apply  \eqref{upperdiagonalization.lem.pf.eq2} and \eqref{upperdiagonalization.lem.pf.eq3}
			to prove
			\vspace{-0.6em}\begin{equation} \label{upperdiagonalization.lem.pf.eq4}
			{ F}(i,j)=0  %\ \ \text{if } 1\le j< i\le N
			\vspace{-0.6em}\end{equation}
			by induction on $j=1, \ldots, N$ and $i=N, \ldots, 1$.

			% To prove  Theorem \ref{polynomialfilter.thm}, we need a lemma.
			%
			%\begin{lemma}\label{upperdiagonalization.lem}
			%Let ${\bf S}_k, \hat  {\bf S}_k, 1\le k\le d$, $\Lambda$ and ${\bf H}$ be as in
			%Theorem \ref{polynomialfilter.thm}.
			%Then
			%${\bf U}^* {\bf H} {\bf U}$ is an upper triangular matrix
			%for the unitary matrix ${\bf U}$ in \eqref{upperdiagonalization}.
			%\end{lemma}
			%
			%\begin{proof}

			For $i=N$ and $j=1$,  applying \eqref{upperdiagonalization.lem.pf.eq2} with $k$ replaced by $k(N,1)$, we obtain
			\vspace{-0.4em}\begin{equation*} %\label{upperdiagonalization.lem.pf.eq2}
			F(N,1) \widehat S_{k(N,1)}(1,1)= \widehat S_{k(N,1)}(N, N) F(N, 1),
			\vspace{-0.4em}\end{equation*}
			which together with \eqref{upperdiagonalization.lem.pf.eq3} proves \eqref{upperdiagonalization.lem.pf.eq4} for $(i,j)=(N, 1)$.
			Inductively we assume that the conclusion  \eqref{upperdiagonalization.lem.pf.eq4} for all pairs $(i, j)$ satisfying either  $1\le j\le j_0$ and $i=i_0$,			or  $1\le j\le N$ and $i_0< i \le N$.
			
			For the case that $j_0<i_0-1$,
			we have
			\vspace{-0.6em}\begin{eqnarray*}  % \label{upperdiagonalization.lem.pf.eq5}
				& &  F(i_0,j_0+1) \widehat S_{k(i_0,j_0+1)}(j_0+1,j_0+1)=
				\sum_{l=1}^{j_0+1} F (i_0,l) \hat S_{k(i_0,j_0+1)}(l,j_0+1)\nonumber\\
				&= &\sum_{l=i_0}^N \widehat  S_{k(i_0, j_0+1)}(i_0, l) F(l, j_0+1)=   \widehat  S_{k(i_0, j_0+1)}(i_0, i_0) F(i_0, j_0+1),
			\end{eqnarray*}
%			\vspace{-0.6em}\begin{equation*}  % \label{upperdiagonalization.lem.pf.eq5}
%				F(i_0,j_0+1) \widehat S_{k(i_0,j_0+1)}(j_0+1,j_0+1)
%				= \sum_{l=i_0}^N \widehat  S_{k(i_0, j_0+1)}(i_0, l) F(l, j_0+1)
%				=   \widehat  S_{k(i_0, j_0+1)}(i_0, i_0) F(i_0, j_0+1),
%			\end{equation*}
			where the first and third equalities hold by the inductive hypothesis and the second equality is obtained from %the inductive hypothesis and
			\eqref{upperdiagonalization.lem.pf.eq2} with $k$ replaced by $k(i_0,j_0+1)$.
%			where the first and thir equality holds by the inductive hypothesis and the second equality is obtained from
%			\eqref{upperdiagonalization.lem.pf.eq2} with $k$ replaced by $k(i_0,j_0+1)$.
			This together with \eqref{upperdiagonalization.lem.pf.eq3} proves the conclusion \eqref{upperdiagonalization.lem.pf.eq4}
			for  $i=i_0$ and $j=j_0+1\le i_0-1$, and hence the inductive proof can proceed for the case that $j_0<i_0-1$.

			For the case that $j_0=i_0-1$,
			it follows from the  construction of the polynomial $h$ and the upper triangular property for $\widehat {\bf S}_k, 1\le k\le d$,
			that the diagonal entries of ${\bf F}$ are
			\vspace{-0.6em}$$ \widehat H(i, i)- h(\widehat{ S}_1(i,i), ..., \widehat{ S}_d(i,i))=0,\ 1\le i\le N\vspace{-0.4em}$$
			by
			\eqref{polynomialfilter.thm.pf.eq1}.
			Hence the conclusion \eqref{upperdiagonalization.lem.pf.eq4} holds
			for    $i=i_0$ and $j=j_0+1$, and hence
			the inductive proof can proceed for the case that $j_0=i_0-1$.
			
			For the case that $i_0\le j_0\le N-1$, we can follow the argument used in the proof
for the case that $j_0<i_0-1$
%			we have
%			\vspace{-0.6em}\begin{eqnarray*}
%				& & \hat S_{k(i_0, j_0+1)}(i_0,i_0) F(i_0,j_0+1)\\
%				%&=& \sum_{l=i_0}^N \hat S_{k(i_0,j_0+1)}(i_0,l) F (l,j_0+1)\\
%				&= & \sum_{l=1}^{j_0+1} F (i_0, l) \hat S_{k(i_0,j_0+1)}(l,j_0+1)\\
%				&= & F(i_0,j_0+1)\hat S_{k(i_0,j_0+1)}(j_0+1,j_0+1),
%			\end{eqnarray*}
%			where the first equality follows  from the inductive hypothesis and
%			\eqref{upperdiagonalization.lem.pf.eq2} with $k$ replaced by $k(i_0, j_0+1)$, and
%			the second equality holds by the inductive hypothesis.
%			This together with \eqref{upperdiagonalization.lem.pf.eq3}
to establish the conclusion \eqref{upperdiagonalization.lem.pf.eq4}
			for  $i=i_0$ and $j=j_0+1\le N$, and hence the inductive proof can proceed for the case that $i_0\le j_0\le N-1$.
			
			For the case that $j_0=N$ and $i_0\ge 2$,
			we obtain
%			\vspace{-0.6em}\begin{eqnarray*}  % \label{upperdiagonalization.lem.pf.eq5}
%				& & F(i_0-1,1) \widehat S_{k(i_0-1,1)}(1,1)\nonumber\\
%				&= &\sum_{l=i_0-1}^N \widehat  S_{k(i_0-1, l)}(i_0-1, l) F(l, 1)\nonumber\\
%				& = &  \widehat  S_{k(i_0-1, 1)}(i_0-1, i_0-1) F(i_0-1, 1),
%			\end{eqnarray*}
\vspace{-0.6em}\begin{equation*}  % \label{upperdiagonalization.lem.pf.eq5}
				 F(i_0-1,1) \widehat S_{k(i_0-1,1)}(1,1)
				= \sum_{l=i_0-1}^N \widehat  S_{k(i_0-1, l)}(i_0-1, l) F(l, 1)
				 =   \widehat  S_{k(i_0-1, 1)}(i_0-1, i_0-1) F(i_0-1, 1),
			\end{equation*}
			where the first equality follows  from
			\eqref{upperdiagonalization.lem.pf.eq2} with $k$ replaced by $k(i_0-1,1)$ and
			the  second equality holds by the inductive hypothesis.
			This together with \eqref{upperdiagonalization.lem.pf.eq3} proves the conclusion \eqref{upperdiagonalization.lem.pf.eq4}
			for  $i=i_0-1$ and $j=1$, and hence
			the inductive proof can proceed for the case that $j_0=N$ and $i_0\ge 2$.
			
			For the case that $j_0=N$ and $i_0=1$, the inductive proof of the zero matrix property for the matrix ${\bf F}$ is complete.
			This completes the inductive proof.
		\end{proof}

\subsection{Distance between a graph filter and the set of polynomial of commutative graph shifts}\label{distance.appendix}

Let ${\mathcal G}=(V, E)$ be a connected, unweighted and undirected finite graph,  ${\mathcal A}$ be a Banach algebra of graph filters on the graph ${\mathcal G}$ with norm denoted by $\|\cdot\|_{\mathcal A}$,
 ${\bf S}_1, \ldots, {\bf S}_d$ be nonzero commutative graph shifts in ${\mathcal A}$, and ${\mathcal P}$ be the set of all polynomial filters of graph shifts In this Appendix, we consider estimating
the distance
${\rm dist}({\bf H}, {\mathcal P})$ in \eqref{distance.def}
%=\inf_{{\bf P}\in {\mathcal P}} \|{\bf H}-{\bf P}\|_{\mathcal A}$$
between a graph filter  ${\bf H}$
 and the set  ${\mathcal P}$ of polynomial filters. %  given in \eqref{distance.def}.

\begin{theorem}
\label{distance.thm}
{\rm
If the commutative
graph shifts  ${\mathbf S}_1, \ldots, {\bf S}_d$ can be diagonalized  simultaneously by a unitary  matrix    and elements in their joint spectrum  $\Lambda$ are distinct, then there exist positive constants $C_0$ and $C_1$ such that
 \begin{equation}\label{distance.thm.eq1} C_0  \Big(\sum_{k=1}^d \|[{\bf H}, {\bf S}_k]\|_{\mathcal A}^2\Big)^{1/2} \le {\rm dist}({\bf H}, {\mathcal P})\le
 C_1 \Big(\sum_{k=1}^d \|[{\bf H}, {\bf S}_k]\|_{\mathcal A}^2\Big)^{1/2}, \ {\bf H}\in {\mathcal A},
 \end{equation}
where $[{\bf H}, {\bf S}_k]={\bf H}{\bf S}_k-{\bf S}_k {\bf H}, 1\le k\le d$.
%hold for all ${\bf H}\in {\mathcal A}$.
}
\end{theorem}	

\begin{proof}  Take
 ${\bf H}\in {\mathcal A}$. For any ${\bf P}\in {\mathcal P}$, we have
  \begin{equation*}
 \|[{\bf H}, {\bf S}_k]\|_{\mathcal A} \le
  \|({\bf H}-{\mathbf P}) {\bf S}_k\|_{\mathcal A}+
  \| {\bf S}_k({\bf H}-{\mathbf P})\|_{\mathcal A}\le 2 \|{\bf S}_k\|_{\mathcal A} \|{\bf H}-{\bf P}\|_{\mathcal A}, \
1\le k\le d. \end{equation*}
%Therefore
%\begin{equation}\label{lip.proof0}
%{\rm dist}({\bf H}, {\mathcal P}) \ge \max_{1\le k\le d, c\in {\mathbb R}} \frac{\|[{\bf H}, {\bf S}_k]\|}{2 \|{\bf S}_k-c{\bf I}\|}.
%\end{equation}
% \begin{eqnarray*}
% \|[{\bf H}, {\bf S}_k]\| & \hskip-0.08in \le  & \hskip-0.08in % \|[{\bf H}, {\bf S}_k- c{\bf I} ]\|\le
%  \|({\bf H}-{\mathbf P}) ({\bf S}_k-c{\bf I})\|+
%  \| ({\bf S}_k-c{\bf I})({\bf H}-{\mathbf P})\|\nonumber\\
%& \hskip-0.08in \le  & \hskip-0.08in 2 \|{\bf S}_k-c{\bf I}\| \|{\bf H}-{\bf P}\|, \
%1\le k\le d. \end{eqnarray*}
Therefore
\begin{equation*}\label{lip.proof0}
{\rm dist}({\bf H}, {\mathcal P}) \ge \max_{1\le k\le d} \frac{\|[{\bf H}, {\bf S}_k]\|_{\mathcal A}}{2 \|{\bf S}_k\|_{\mathcal A}},
\end{equation*}
and the first inequality in \eqref{distance.thm.eq1} follows.

Now we prove the second  inequality in \eqref{distance.thm.eq1}. Let ${\bf U}$
be the unitary matrix to diagonalize   ${\mathbf S}_1, \ldots, {\bf S}_d$  simultaneously,
	i.e., \eqref{upperdiagonalization} holds for some diagonal matrices
$\widehat{\bf S}_k={\rm diag} (\widehat S(i,i))_{i\in V}, 1\le k\le d$.
Then one may verify that  polynomial filters of
  graph shifts  ${\mathbf S}_1, \ldots, {\bf S}_d$ can  also be diagonalized by the unitary matrix ${\bf U}$. Moreover by the distinct assumption on elements in the joint spectrum  $\Lambda$ of the graph shifts, we have
\begin{equation}\label{P2.def0}
{\mathcal P}= \{ {\bf U}^{\rm H}  {\bf D} {\bf U}, \ {\bf D}\ {\rm are\  diagonal\ matrices} \}.
\end{equation}
 Set ${\bf U}{\bf H}{\bf U}^{\rm H}=(\widehat H(i,j))_{i,j\in V}$.
% ${\bf S}_1, \ldots, {\bf S}_d$.
 and denote the Frobenius norm of a matrix ${\bf A}$ by $\|{\bf A}\|_F$. Therefore  it follows from
\eqref{P2.def0} that
\begin{equation}\label{lip.proof1}
\inf_{{\bf P}\in {\mathcal P}}\|{\bf H}-{\bf P}\|_F=\inf_{{\bf D}\ {\rm are\ diagonal\ matrices}}\|{\bf U}^{\rm H}{\bf H}{\bf U}-{\bf D}\|_F
=\Big(\sum_{i,j\in V, j\ne i} |\widehat H(i,j)|^2\Big)^{1/2}.
\end{equation}
 On the other hand, we have
 $${\bf U} [{\bf H}, {\bf S}_k] {\bf U}^{\rm H}= \Big(\widehat H(i,j)  (\widehat S_{k}(j,j)- \widehat S_{k}(i,i))\Big)_{i,j\in V},\ 1\le k\le d.$$
 This implies that
 \begin{eqnarray}  \label{lip.proof2}
 \sum_{k=1}^d \|[{\bf H}, {\bf S}_k]\|_F^2 & \hskip-0.08in = & \hskip-0.08in  \sum_{k=1}^d \|{\bf U}^{\rm H}|[{\bf H}, {\bf S}_k]{\bf U}\|_F^2
 =\sum_{i,j\in V, j\ne i} |\widehat H(i,j)|^2 \Big(\sum_{k=1}^d \big|\widehat S_{k}(j,j)- \widehat S_{k}(i,i)\big|^2\Big)\nonumber\\
 & \hskip-0.08in \ge & \hskip-0.08in  \inf_{i,j\in V, j\ne i} \Big(\sum_{k=1}^d \big|\widehat S_{k}(j,j)- \widehat S_{k}(i,i)\big|^2\Big)
 \times \inf_{{\bf P}\in {\mathcal P}}\|{\bf H}-{\bf P}\|_F^2,  %\sum_{i,j\in V, j\ne i} |\widehat H(i,j)|^2
 \end{eqnarray}
 where the last inequality follows from \eqref{lip.proof1}.
 Then the second inequality in  \eqref{distance.thm.eq1} follows from \eqref{lip.proof2},
  the equivalence of  norms on a finite-dimensional linear space and the  distinct assumption on the joint spectrum $\Lambda$.
\end{proof}

We believe that the estimate  \eqref{distance.thm.eq1} % on the distance between a graph filter
% and the set  of polynomial of commutative graph shifts  ${\bf S}_1, \ldots, {\bf S}_d$
should hold without the  simultaneous diagonalization assumption on commutative graph shifts ${\bf S}_1, \ldots, {\bf S}_d$.

%\end{appendix}

	%%%%%%%%%%%%%%%%%%%%%%%%%%%%%%%%%%%%%%%%%%%%%%%%%
%%%%%%%%%%%%%%%%%%%%%%%%%%%%%%%%%%%%%%%%%	
		
	%	\vspace{-0.8em}

			%\end{appendix}
			%\end{appendices}
%%%%%%%%%%%%%%%%%%%%%%%%%%%%%%%%%%%%%%%%%%%%%%%%%%%%%%%%%%
%%%%%%%%%%%%%%%%%%%%%%%%%%%%%%%%%%%%%%%%%%%%%%%%%%%%%%%%%

		%\end{appendix}
		
		%
		%
		%Since, $\Lambda$ contains $N$ distinct d-tuples then for each $i\in V$ there is a $q\in \{1,\dots,d\}$ such that $\lambda_i^q\neq \lambda_{i+1}^q$. For $i\in V$ and corresponding $q\in \{1,\dots, d\}$ we have
		%$$S_q(i,i)F(i,i+1)=F(i,i+1)S_q(i+1,i+1).$$
		% By the fact that $S_q(i,i)=\lambda_i^q$ and $\lambda_i^q\neq \lambda_{i+1}^q$, we get $F(i,i+1)=0$.
		%
		%Now, lets apply induction method on m to prove $F(i,i+m)=0$ for all $i\in V$. Assume $F(i,i+z)=0$ for all $i\in V$ and $z\leq m$. Using similar technique as in $m=1$ case, for each $i\in V$ and $m$ there is a $r\in \{1,\dots, d\}$ such that $\lambda_i^r\neq \lambda_{i+m+1}^r$. Similarly, $(i,i+m+1)^{th}$ entry of $S_rF=FS_r$ is
		%$$\sum_{k\in V }S_r(i,k)F(k,i+m+1)=\sum_{k\in V}F(i,k)S_r(k,i+m+1).$$
		%Since,  $F(i,i+z)=0$ for all $i\in V$ and $z\leq m$, we have
		%$$S_r(i,i)F(i,i+m+1)=F(i,i+m+1)S_r(i+m+1,i+m+1),$$
		%which is
		%$$F(i,i+m+1)(\lambda_{i}^r-\lambda_{i+m+1}^r)=0.$$
		%We obtain $F(i,i+m+1)=0$ for $i\in V$.
		%Then by induction we prove that the matrix ${\bf F}$ is zero matrix and  ${\bf H}=p({\bf S}_1,...,{\bf S}_d)$.\\

\bigskip

{\bf Acknowledgement}:\  This work is partially supported by the National Natural Science Foundation of China (61761011, 62171146, 12171490)
and the  National Science Foundation (DMS-1816313).   The authors would like to thank anonymous reviewers to provide many constructive comments
for the improvement of the paper.
On behalf of all authors, the corresponding author states that there is no conflict of interest.

\end{document}